\documentclass[letterpaper,twocolumn,10pt]{article}
\usepackage{usenix-2020-09}

\usepackage{tikz}
\usepackage{amsmath}

\usepackage{filecontents}
\usepackage{graphicx}
\usepackage{amsthm}
\usepackage{amssymb}
\usepackage{thmtools}
\usepackage{color}

\usepackage[linesnumbered, lined, boxed]{algorithm2e}

\usepackage[normalem]{ulem}

\newtheorem{definition}{Definition}
\newtheorem{theorem}{Theorem}
\newtheorem{lemma}{Lemma}
\newtheorem{proposition}{Proposition}
\newtheorem{corollary}{Corollary}

\newcommand{\nats}{\mathbb{N}}
\newcommand{\nnints}{\mathbb{Z}_{\ge0}}
\newcommand{\reals}{\mathbb{R}}
\newcommand{\nnreals}{\mathbb{R}_{\ge0}}

\renewcommand{\epsilon}{\varepsilon}

\newcommand{\calA}{\mathcal{A}}

\newcommand{\calD}{\mathcal{D}}
\newcommand{\calG}{\mathcal{G}}

\newcommand{\calM}{\mathcal{M}}
\newcommand{\calR}{\mathcal{R}}
\newcommand{\calS}{\mathcal{S}}
\newcommand{\calX}{\mathcal{X}}

\newcommand{\E}{\operatorname{\mathbb{E}}}
\newcommand{\V}{\operatorname{\mathbb{V}}}

\newcommand{\bmA}{\mathbf{A}}

\newcommand{\bmG}{\mathbf{G}}
\newcommand{\bmQ}{\mathbf{Q}}
\newcommand{\bma}{\mathbf{a}}
\newcommand{\bmb}{\mathbf{b}}

\newcommand{\bme}{\mathbf{e}}

\newcommand{\bd}{\bar{d}}
\newcommand{\bbma}{\bar{\bma}}
\newcommand{\hd}{\hat{d}}
\newcommand{\hf}{\hat{f}}
\newcommand{\hg}{\hat{g}}
\newcommand{\hr}{\hat{r}}
\newcommand{\hw}{\hat{w}}
\newcommand{\td}{\tilde{d}}
\newcommand{\spantwo}[1]{\multicolumn{2}{|c|}{#1}}

\newcommand{\IMDB}{\textsf{IMDB}}
\newcommand{\Orkut}{\textsf{Orkut}}
\newcommand{\alg}{\textsf}
\newcommand{\Lap}{\textrm{Lap}}

\usepackage{hyperref}
\newcommand{\footremember}[2]{%
    \thanks{#2}
    \newcounter{#1}
    \setcounter{#1}{\value{footnote}}%
}
\newcommand{\footrecall}[1]{%
    \footnotemark[\value{#1}]%
}

\hypersetup{
  colorlinks,
  linkcolor={black},
  citecolor={red!70!black},
  urlcolor={blue!70!black}
}

\newif\ifconferenceon\conferenceonfalse
\ifconferenceon
\newcommand{\conference}[1]{#1}
\newcommand{\arxiv}[1]{}
\else
\newcommand{\conference}[1]{}
\newcommand{\arxiv}[1]{#1}
\usepackage{balance}
\fi

\begin{document}

\title{\Large \bf Locally Differentially Private Analysis of Graph Statistics}

\author{
{\rm Jacob Imola}\footremember{contributions}{The first and second authors made equal contributions.}\\
UC San Diego
\and
{\rm Takao Murakami}\footrecall{contributions}\\
AIST
\and
{\rm Kamalika Chaudhuri}\\
UC San Diego
} % end author

\maketitle

\begin{abstract}

Differentially private analysis of graphs is widely used for releasing statistics from sensitive graphs while still preserving user privacy. Most existing algorithms however are in a centralized privacy model, where a trusted data curator holds the entire graph. As this model raises a number of privacy and security issues -- such as, the trustworthiness of the curator and the possibility of data breaches, it is desirable to consider algorithms in a more decentralized local model where no server holds the entire graph.

In this work, we consider a local model, and present algorithms for counting subgraphs -- a fundamental task for analyzing the connection patterns in a graph -- with LDP (Local Differential Privacy). For triangle counts, we present algorithms that use one and two rounds of interaction, and show that an additional round can significantly improve the utility. For $k$-star counts, we present an algorithm that achieves an order optimal estimation error in the non-interactive local model. 
We provide new lower-bounds on the estimation error for general graph statistics including triangle counts and $k$-star counts. 
Finally, we perform extensive experiments on two real datasets, and show that it is indeed possible to accurately estimate 
subgraph counts in the local differential privacy model. 

\end{abstract}

\section{Introduction}
\label{sec:intro}

Analysis of network statistics is a useful tool for finding meaningful patterns in graph data, such as social, e-mail,  citation and epidemiological networks. 
For example, the average \textit{degree} (i.e., number of edges connected to a node) in a social graph can reveal the average connectivity. 
\textit{Subgraph counts} 
(e.g., the number of triangles, stars, or cliques) can be used to measure 
centrality properties such as the \textit{clustering coefficient}, 
which represents the probability that two friends of an individual will also be friends of one another \cite{Newman_PRL09}. 
However, the vast majority of graph analytics is carried out on sensitive data, which could be leaked through the results of graph analysis. Thus, there is a need to develop solutions that can analyze these graph properties while still preserving the privacy of 
individuals 
in the network.

The standard way to analyze graphs with privacy is through differentially private graph analysis \cite{Raskhodnikova_Encyclopedia16,DP,Dwork_ICALP06}. 
Differential privacy 
provides individual privacy against adversaries with arbitrary background knowledge, and has currently emerged as the gold standard for private analytics. 
However, a vast majority of differentially private graph analysis algorithms are in the \textit{centralized (or global) model} \cite{Blocki_FOCS12,Chen_PoPETs20,Day_SIGMOD16,Hay_ICDM09,Karwa_PVLDB11,Kasiviswanathan_TCC13,Nissim_STOC07,Raskhodnikova_arXiv15,Raskhodnikova_Encyclopedia16,Song_arXiv18,Wang_PAKDD13,Wang_TDP13}, where a single trusted data curator holds the entire graph and releases sanitized versions of the statistics. 
By assuming a trusted party that can access the entire graph, 
it is possible to release accurate graph statistics 
(e.g., subgraph counts \cite{Karwa_PVLDB11,Kasiviswanathan_TCC13,Song_arXiv18}, degree distribution \cite{Day_SIGMOD16,Hay_ICDM09,Raskhodnikova_arXiv15}, spectra \cite{Wang_PAKDD13}) 
and synthetic graphs \cite{Chen_PoPETs20,Wang_TDP13}. 

In many applications however, a single trusted curator may not be practicable due to security or logistical reasons. A centralized data holder is amenable to security issues such as data breaches and leaks -- a growing threat in recent years \cite{data_breach1,data_breach2}. Additionally, \textit{decentralized social networks} \cite{Paul_CN14,Salve_CSR18} (e.g., Diaspora \cite{Diaspora}) have no central server that contains an entire social graph, and use instead many servers all over the world, each containing the data of users who have chosen to register there.  Finally, a centralized solution is also not applicable to fully decentralized applications, where the server does not automatically hold information connecting users. An example of this is a mobile application that asks each user how many of 
her 
friends 
she has 
seen today, and sends noisy counts to a central server. In this application, the server does not hold any individual edge, but can still aggregate the responses to determine the average mobility in an area. 

The standard privacy solution that does not assume a trusted third party is LDP (Local Differential Privacy) \cite{Duchi_FOCS13,Kasiviswanathan_FOCS08}. 
This is a special case of DP 
(Differential Privacy) 
in the \textit{local model}, where each user obfuscates her personal data by herself and sends the obfuscated data to a (possibly malicious) data collector. 
Since the data collector does not hold the original personal data, it does not suffer from data leakage issues. 
Therefore, LDP has recently attracted attention from both academia \cite{Acharya_AISTATS19,Bassily_STOC15,Bassily_NIPS17,Fanti_PoPETs16,Kairouz_ICML16,Kairouz_JMLR16,Murakami_USENIX19,Qin_CCS16,Wang_USENIX17,Ye_ISIT17} as well as industry \cite{Erlingsson_CCS14,Ding_NIPS17,Thakurta_USPatent17}. 
However, the use of LDP has mostly been in the context of tabular data where each row corresponds to an individual, and little attention has been paid to LDP for more complex data such as graphs (see 
Section~\ref{sec:related} for details). 

\begin{figure}
\centering
\includegraphics[width=0.9\linewidth]{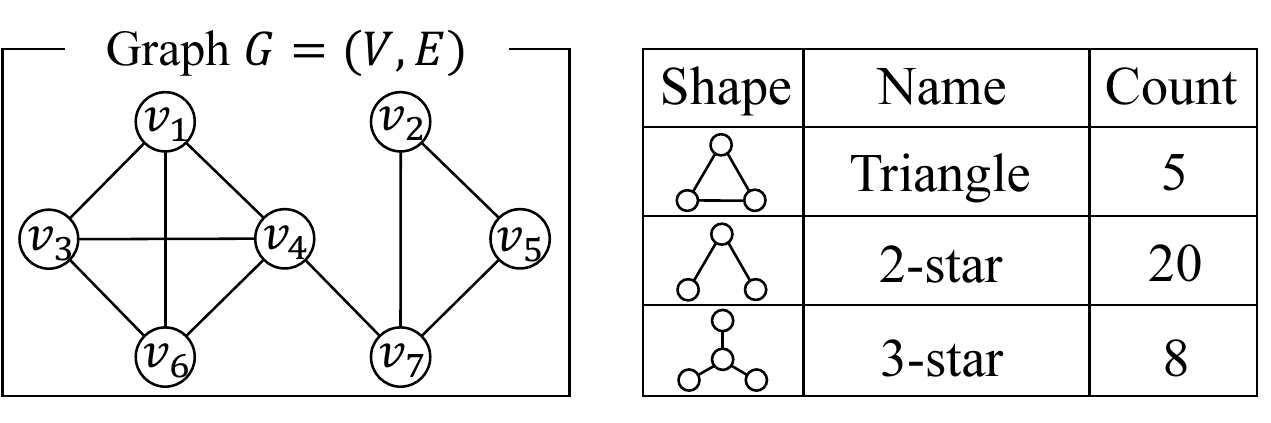}
\vspace{-5.8mm}
\caption{Example of subgraph counts.}
\label{fig:subgraph}
\end{figure}

In this paper, we consider LDP for graph data, and 
provide algorithms and theoretical performance guarantees
for calculating graph statistics in this model. 
In particular, we focus on counting triangles and $k$-stars -- the most basic and useful subgraphs. 
A triangle is a set of three nodes with three edges (we exclude automorphisms; i.e., \#closed triplets $= 3 \times$ \#triangles). 
A $k$-star consists of a central node connected to $k$ other nodes. 
Figure~\ref{fig:subgraph} shows an example of triangles and $k$-stars. 
Counting them is a fundamental task of analyzing the connection patterns in a graph, as 
the clustering coefficient can be calculated from triangle and $2$-star counts as: $\frac{3 \times \text{\#triangles}}{\#2\text{-stars}}$ (in Figure~\ref{fig:subgraph}, $\frac{3 \times 5}{20} = 0.75$). 

When we look to protect privacy of relationship information modeled by edges in a graph, 
we need to pay attention to the fact 
that some relationship information 
could be leaked from subgraph counts. 
For example, 
suppose that user (node) $v_2$ in Figure~\ref{fig:subgraph} knows 
all edges connected to $v_2$ and 
all edges between $v_3, \ldots, v_7$ 
as background knowledge, and that $v_2$ wants to know who are friends with $v_1$. 
Then ``\#2-stars = 20'' reveals 
the fact that $v_1$ has three friends, and ``\#triangles = 5'' reveals 
the fact that the three friends of $v_1$ are $v_3$, $v_4$, and $v_6$. 
Moreover, a central server that holds all friendship information (i.e., all edges) may face data breaches, as explained above. 
Therefore, a private algorithm for counting subgraphs in the local model is highly beneficial to individual privacy.

The main challenge in counting subgraphs in the local model is that existing techniques and their analysis do not directly apply. 
The existing work on LDP for tabular data assumes that 
each person's data 
is independently and identically drawn from an underlying distribution. 
In graphs, this is no longer the case; e.g., 
each triangle is not independent, 
because multiple triangles can involve the same edge; 
each $k$-star is not independent for the same reason. 
Moreover, complex inter-dependencies involving multiple people 
are 
possible in graphs. 
For example, each user cannot count triangles involving herself, because she cannot see edges between other users; 
e.g., 
user 
$v_1$ cannot see an edge between $v_3$ and $v_4$ in Figure~\ref{fig:subgraph}. 

We show that although these complex dependency among users introduces challenges, it also presents opportunities. Specifically, this kind of interdependency also implies that extra interaction between users and a data collector may be helpful depending on the prior responses. 
In this work, we investigate this issue and provide algorithms for accurately calculating subgraph counts under LDP. 

\smallskip
\noindent{\textbf{Our contributions.}}~~In this paper, we provide 
algorithms and corresponding performance guarantees for counting triangles and $k$-stars in graphs under edge Local Differential Privacy. Specifically, our contributions are as follows:
\begin{itemize}
    \item For triangles, we present two algorithms. The first is based on 
    Warner's RR (Randomized Response) \cite{Warner_JASA65} 
    and empirical estimation \cite{Kairouz_ICML16,Murakami_USENIX19,Wang_USENIX17}. 
    We then present a more sophisticated algorithm that uses an additional round of interaction between users and data collector. We provide upper-bounds on the estimation error for each algorithm, and show that the latter can  significantly reduce the estimation error. 
    
    \item For $k$-stars, we present a simple algorithm using the Laplacian mechanism. We analyze the upper-bound on the estimation error for this algorithm, and 
    show that it is order optimal in terms of the number of users among all LDP mechanisms that do not use additional interaction.
    
    \item We provide lower-bounds on the estimation error for 
    %estimating 
    general graph functions including triangle counts and $k$-star counts in the local model. These are stronger than known upper bounds in the centralized model, and illustrate the limitations of the local model over the central.
    
    \item Finally, we evaluate our algorithms on two real datasets, and show that 
    it is indeed possible to accurately estimate 
    subgraph counts in the local model. 
    %Our experimental results show that 
    %although the algorithms in the local model provide high estimation errors than those in the centralized model,  
    In particular, we show that 
    the interactive algorithm for triangle counts and the Laplacian algorithm for the $k$-stars provide small estimation errors when the number of users is large.
\end{itemize}

We implemented our algorithms with C/C++, and published them as open-source software \cite{LDPGraphStatistics}.

\section{Related Work}
\label{sec:related}

\smallskip
\noindent{\textbf{Graph DP.}}~~DP on graphs has been widely studied, with most prior work being in the centralized model \cite{Blocki_FOCS12,Chen_PoPETs20,Day_SIGMOD16,Hay_ICDM09,Karwa_PVLDB11,Kasiviswanathan_TCC13,Nissim_STOC07,Raskhodnikova_arXiv15,Raskhodnikova_Encyclopedia16,Song_arXiv18,Wang_PAKDD13,Wang_TDP13}. 
In this model, a number of algorithms have been proposed for releasing subgraph counts \cite{Karwa_PVLDB11,Kasiviswanathan_TCC13,Song_arXiv18}, degree distributions \cite{Day_SIGMOD16,Hay_ICDM09,Raskhodnikova_arXiv15}, eigenvalues and eigenvectors 
\cite{Wang_PAKDD13}, and synthetic graphs \cite{Chen_PoPETs20,Wang_TDP13}. 

There has also been a handful of work on graph algorithms in the local DP model~\cite{Qin_CCS17,Sun_CCS19,Ye_ICDE20,Ye_TKDE21,Zhang_USENIX20}. 
For example, Qin \textit{et al.} \cite{Qin_CCS17} propose an algorithm for generating synthetic graphs. 
Zhang \textit{et al.} \cite{Zhang_USENIX20} propose an algorithm for software usage analysis under LDP, where a node represents a software component (e.g., function in a code) and an edge represents a control-flow between components. 
Neither of these works focus on subgraph counts. 

Sun \textit{et al.} \cite{Sun_CCS19} propose an algorithm for counting subgraphs in the local model under the assumption that each user allows her friends to see all her connections. 
However, this assumption does not hold in many practical scenarios; e.g., a Facebook user can change her setting so that friends cannot see her connections. Therefore, we assume that each user knows only her friends rather than all of her friends' friends. 
The algorithms in \cite{Sun_CCS19} cannot be applied to this setting.

Ye \textit{et al.} \cite{Ye_ICDE20} propose a one-round algorithm for estimating graph metrics including the clustering coefficient. 
Here they apply Warner's RR (Randomized Response) to an adjacency matrix. 
However, it introduces a very large bias for triangle counts.
In \cite{Ye_TKDE21}, they propose a method for reducing the bias in the estimate of triangle counts. 
However, the method in \cite{Ye_TKDE21} introduces some approximation, and it is unclear whether their estimate is unbiased. 
In this paper, we propose a one-round algorithm for triangles that uses empirical estimation as a post-processing step, and prove that our estimate is unbiased. 
We also show in Appendix~\ref{sec:RR_emp} that our one-round algorithm significantly outperforms the one-round algorithm in \cite{Ye_ICDE20}. 
Moreover, we show in Section~\ref{sec:experiments} that our two-rounds algorithm significantly outperforms our one-round algorithm.

Our work also differs from \cite{Sun_CCS19,Ye_ICDE20,Ye_TKDE21} in that we provide lower-bounds on the estimation error.

\smallskip
\noindent{\textbf{LDP.}}~~Apart from graphs, a number of works have looked at analyzing statistics (e.g., discrete distribution estimation\cite{Acharya_AISTATS19,Fanti_PoPETs16,Kairouz_ICML16,Kairouz_JMLR16,Murakami_USENIX19,Wang_USENIX17,Ye_ISIT17}, 
heavy hitters \cite{Bassily_STOC15,Bassily_NIPS17,Qin_CCS16}) under LDP. 

However, they use LDP in the context of tabular data, and do not consider the kind of complex interdependency in graph data (as described in Section~\ref{sec:intro}). 
For example, the RR with empirical estimation is optimal in the low privacy regimes for estimating a distribution for tabular data \cite{Kairouz_ICML16,Kairouz_JMLR16}. 
We apply the RR and empirical estimation to counting triangles, and show that it is suboptimal and significantly outperformed by a more sophisticated 
two-rounds 
algorithm. 

\smallskip
\noindent{\textbf{Upper/lower-bounds.}}~~Finally, we note that existing work on upper-bounds and lower-bounds cannot be directly applied to our setting. 
For example, there are upper-bounds
\cite{Acharya_AISTATS19,Kairouz_ICML16,Kairouz_JMLR16,Ye_ISIT17,Joseph_ArXiv19,Joseph_ArXiv19_Gauss} and
lower-bounds
\cite{Acharya_arXiv20,Duchi_ArXiv14,Duchi_ArXiv17,Joseph_ArXiv19,Duchi_ArXiv19,
Joseph_ArXiv19_Gauss, Joseph_SODA20} on the estimation error (or sample complexity) in 
distribution estimation of tabular data. 
However, they assume that each original data value is independently sampled from an underlying distribution. 
They cannot be directly applied to our graph setting, because each triangle and each $k$-star involve multiple edges and are not independent (as described in Section~\ref{sec:intro}). 
Rashtchian \textit{et al.} \cite{Rashtchian_arXiv20} provide lower-bounds on communication complexity (i.e., number of queries) of vector-matrix-vector queries for estimating subgraph counts. 
However, their lower-bounds are not on the estimation error, and cannot be applied to our problem.

 %-------------------------------------------------------------------------------
\section{Preliminaries}
\label{sec:preliminaries}

\subsection{Graphs and Differential Privacy}
\label{sub:graphs_CDP}
\noindent{\textbf{Graphs.}}~~Let $\nats$, $\nnints$, $\reals$, and $\nnreals$ be the sets of natural numbers, non-negative integers, real numbers, and non-negative real numbers, respectively. 
For $a \in \nats$, let $[a] = \{1, 2, \ldots, a\}$. 

We consider an undirected graph $G=(V,E)$, where 
$V$ is a set of nodes (i.e., users) and $E$ is a set of edges. 
Let $n \in \nats$ be the number of users in $V$, and let $v_i \in V$ the $i$-th user; i.e., $V=\{v_1,\ldots,v_n\}$. 
An edge $(v_i, v_j) \in E$ represents a relationship between users $v_i \in V$ and $v_j \in V$. 
The number of edges connected to a single node is called the \textit{degree} of the node. 
Let $d_{max} \in \nats$ be the \textit{maximum degree} (i.e., maximum number of edges connected to a node) in graph $G$. 
Let $\calG$ be the set of possible graphs 
$G$ on $n$ users. 
A graph $G \in \calG$ can be represented as a symmetric adjacency matrix $\bmA=(a_{i,j}) \in \{0,1\}^{n \times n}$, where $a_{i,j}=1$ if $(v_i,v_j) \in E$ and $a_{i,j}=0$ otherwise.

\smallskip
\noindent{\textbf{Types of DP.}}~~
DP (Differential Privacy) \cite{DP,Dwork_ICALP06} is known as a gold standard for data privacy. 
According to the underlying architecture, DP can be divided into two 
types: 
\textit{centralized DP} and \textit{LDP (Local DP)}. 
Centralized DP assumes the centralized model, where a ``trusted'' data collector 
collects the original personal data from all users and obfuscates a query (e.g., counting query, histogram query) on the set of personal data. 
LDP assumes the local model, where each user does not trust even the data collector. 
In this model, each user obfuscates a query on her personal data by herself and sends the obfuscated data to the data collector. 

If the data are represented as a graph, we can consider two types of DP: 
\textit{edge DP} and \textit{node DP} \cite{Hay_ICDM09,Raskhodnikova_Encyclopedia16}. 
Edge DP considers two neighboring graphs $G, G' \in \calG$ that differ in one edge. 
In contrast, node DP considers two neighboring graphs $G, G' \in \calG$ in which $G'$ is obtained from $G$ by adding or removing one node along with its adjacent edges. 

Although Zhang \textit{et al.} \cite{Zhang_USENIX20} consider node DP in the local model where each node represents a software component, we consider a totally different problem where each node represents a user. 
In the latter case, 
node DP requires us to hide the \textit{existence of each user} along with her all edges. 
However, many applications in the local model send the identity of each user to a server. 
For example, 
we can consider a mobile application 
that sends to a server how many friends a user met today along with her user ID. 
In this case, the user may not mind sending her user ID, 
but may want to hide her edge information (i.e., who she met today). 
Although we cannot use node DP in such applications, we can use edge DP to deny the presence/absence of each edge (friend). 
Thus we focus on edge DP in the same way as \cite{Qin_CCS17,Sun_CCS19,Ye_ICDE20,Ye_TKDE21}. 

Below we explain edge DP in the centralized model. 

\smallskip
\noindent{\textbf{Centralized DP.}}~~We call edge DP in the centralized model \textit{edge centralized DP}. 
Formally, it is defined as follows:

\begin{definition} [$\epsilon$-edge centralized DP] \label{def:edge_CDP} 
Let $\epsilon \in \nnreals$. 
A randomized algorithm $\calM$ with domain $\calG$ provides \emph{$\epsilon$-edge centralized DP} 
if for any two 
neighboring 
graphs $G, G' \in \calG$ that differ in one edge and any $S \subseteq \mathrm{Range}(\calM)$, 
\begin{align}
\Pr[\calM(G) \in S] \leq e^\epsilon \Pr[\calM(G') \in S].
\label{eq:edge_CDP}
\end{align}
\end{definition}
Edge centralized DP guarantees that an adversary who has observed the output of $\calM$ cannot determine whether it is come from $G$ or $G'$ with a certain degree of confidence. 
The parameter $\epsilon$ is called the \textit{privacy budget}. 
If $\epsilon$ is close to zero, then $G$ and $G'$ are almost equally likely, which means that an edge in $G$ is strongly protected. 

We also note that edge DP can be 
used 
to protect $k\in\nats$ edges by using 
the notion of group privacy \cite{DP}.
Specifically, if $\calM$ provides $\epsilon$-edge centralized DP, then for any two graphs $G, G' \in \calG$ that differ in $k$ edges and any $S \subseteq \mathrm{Range}(\calM)$, we obtain: 
$\Pr[\calM(G) \in S] \leq e^{k\epsilon} \Pr[\calM(G') \in S]$; i.e., $k$ edges are protected with privacy budget $k\epsilon$.

\subsection{Local Differential Privacy}
\label{sub:LDP}
LDP (Local Differential Privacy) \cite{Kasiviswanathan_FOCS08,Duchi_FOCS13} is a privacy metric to protect personal data of each user in the local model. 
LDP has been originally introduced to protect each user's data record that is independent from the other records. 
However, in a graph, each edge is connected to two users. 
Thus, when we define edge DP in the local model, 
we should consider what we want to protect. 
In this paper, we consider two definitions of edge DP in the local model: \textit{edge LDP} in \cite{Qin_CCS17} and 
\textit{relationship DP} 
introduced in this paper. 
Below, we will explain these two definitions in detail.

\smallskip
\noindent{\textbf{Edge LDP.}}~~Qin \textit{et al.} \cite{Qin_CCS17} defined edge LDP based on a user's \textit{neighbor list}. 
Specifically, 
let $\bma_i = (a_{i,1}, \ldots, a_{i,n}) \in \{0,1\}^n$ be a neighbor list of user $v_i$. 
Note that $\bma_i$ is the $i$-th row of the adjacency matrix $\bmA$ of graph $G$. 
In other words, graph $G$ can be represented as neighbor lists $\textbf{a}_1, \ldots, \textbf{a}_n$. 

Then edge LDP is defined as follows: 

\begin{definition} [$\epsilon$-edge LDP \cite{Qin_CCS17}] \label{def:edge_LDP} 
Let $\epsilon \in \nnreals$. 
For any $i \in [n]$, let $\calR_i$ with domain $\{0,1\}^n$ be a randomized algorithm of user $v_i$. 
$\calR_i$ 
provides \emph{$\epsilon$-edge LDP} 
if for any two neighbor lists 
$\bma_i, \bma'_i \in \{0,1\}^n$ 
that differ in one bit and any 
$S \subseteq \mathrm{Range}(\calR_i)$, 
\begin{align}
\Pr[\calR_i(\bma_i) \in S] \leq e^\epsilon \Pr[\calR_i(\bma'_i) \in S].
\label{eq:edge_LDP}
\end{align}
\end{definition}
Edge LDP in Definition~\ref{def:edge_LDP} protects 
a single bit in a neighbor list with privacy budget $\epsilon$. 
As with edge centralized DP, edge LDP can also be 
used 
to protect $k \in \nats$ 
bits in a neighbor list 
by using group privacy; i.e., $k$ bits in a neighbor list are protected with privacy budget $k\epsilon$.

\smallskip
\noindent{\textbf{RR (Randomized Response).}}~~As a simple example of a randomized algorithm 
$\calR_i$ 
providing $\epsilon$-edge LDP, we explain 
Warner's RR (Randomized Response) \cite{Warner_JASA65} applied to a neighbor list, 
which is called 
the randomized neighbor list in \cite{Qin_CCS17}. 

Given a neighbor list 
$\bma_i \in \{0,1\}^n$, 
this algorithm 
outputs a noisy neighbor lists 
$\bmb = (b_1, \ldots, b_n) \in \{0,1\}^n$ 
by flipping each bit in 
$\bma_i$ 
with probability $p = \frac{1}{e^\epsilon + 1}$; i.e., for each $j \in [n]$, 
$b_j \neq a_{i,j}$ with probability $p$ and $b_j = a_{i,j}$ with probability $1-p$. 
Since 
$\Pr[\calR(\bma_i) \in S]$ and $\Pr[\calR(\bma'_i) \in S]$ 
in (\ref{eq:edge_LDP}) differ by $e^\epsilon$ for 
$\bma_i$ and $\bma'_i$ 
that differ in one bit, 
this algorithm 
provides $\epsilon$-edge LDP. 

\smallskip
\noindent{\textbf{Relationship DP.}}~~In graphs such as social networks, 
it is usually the case that two users share knowledge of the presence of an edge between them. 
To hide their mutual edge, 
we must consider
that both user's outputs can leak information. 
We introduce a DP definition called relationship DP that hides \textit{one entire edge in graph $G$ during the whole process:}

\begin{definition} [$\epsilon$-relationship DP] 
\label{def:entire_edge_LDP} 
Let $\epsilon \in \nnreals$. 
A tuple of randomized algorithms $(\calR_1, \ldots, \calR_n)$, 
each of which is with domain $\{0,1\}^n$, 
provides 
\emph{$\epsilon$-relationship DP} 
if for any two 
neighboring 
graphs $G, G' \in \calG$ that differ in one edge and any $S \subseteq \mathrm{Range}(\calR_1) \times \ldots \times \mathrm{Range}(\calR_n)$, 
\begin{align}
&\Pr[(\calR_1(\bma_1), \ldots, \calR_n(\bma_n)) \in S] \nonumber\\
&\leq e^\epsilon \Pr[(\calR_1(\bma'_1), \ldots, \calR_n(\bma'_n)) \in S],
\label{eq:entire_edge_LDP}
\end{align}
where $\bma_i$ (resp.~$\bma'_i$) $\in \{0,1\}^n$ is the $i$-th row of the adjacency matrix of graph $G$ (resp.~$G'$).
\end{definition}

Relationship DP is the same as \textit{decentralized DP} in \cite{Sun_CCS19} except that the former (resp.~latter) assumes that each user knows only her friends (resp.~all of her friends' friends).

Edge LDP assumes that 
user $v_i$'s edge connected to user $v_j$ 
and 
user $v_j$'s edge connected to user $v_i$ 
are different secrets, with user $v_i$ knowing the former and user $v_j$ knowing the latter. 
Relationship DP assumes that the two secrets are the same.

Note that 
the threat model of relationship DP is 
different from 
that of 
LDP -- 
some amount of trust must be given to the other users 
in relationship DP. 
Specifically, user $v_i$ must trust user $v_j$ to not leak information
about their shared edge. If $k \in \nats$ users decide not to follow their protocols, 
then up to $k$ edges incident to user $v_i$ may be compromised. This trust model
is stronger than 
LDP, 
which assumes nothing about what other users 
do,
but is much weaker than centralized DP in which 
all edges are 
in the hands of the central party.

Other than the differing threat models, relationship DP and edge LDP are quite closely related:

\begin{proposition} \label{prop:edge_LDP_entire_edge_LDP} 
If randomized algorithms $\calR_1, \ldots, \calR_n$ provide $\epsilon$-edge LDP, 
then $(\calR_1, \ldots, \calR_n)$ provides $2\epsilon$-relationship DP.
\end{proposition}

\begin{proof}
The existence of edge $(v_i, v_j) \in E$ affects two elements $a_{i,j}, a_{j,i} \in \{0,1\}$ in the adjacency matrix $\bmA$. 
  Then by group privacy~\cite{DP}, Proposition~\ref{prop:edge_LDP_entire_edge_LDP} holds.
\end{proof}

Proposition~\ref{prop:edge_LDP_entire_edge_LDP} states that when we want to protect one edge as a whole, the privacy budget is at most doubled. 
Note, however, that 
some randomized algorithms do not have this doubling issue. 
For example, we can apply the RR to the $i$-th neighbor list $\bma_i$ so that $\calR_i$ outputs noisy bits 
$(b_1, \ldots, b_{i-1}) \in \{0,1\}^{i-1}$ 
for only users 
$v_1, \ldots, v_{i-1}$ with smaller user IDs; 
i.e., 
for each 
$j \in \{1, \ldots, i-1\}$, 
$b_j \neq a_{i,j}$ with probability $p = \frac{1}{e^\epsilon + 1}$ and $b_j = a_{i,j}$ with probability $1-p$. 
In other words, we can extend 
the RR for a neighbor list 
so that $(\calR_1, \ldots, \calR_n)$ outputs only 
the lower triangular part 
of the noisy adjacency matrix. 
Then all of $\calR_1, \ldots, \calR_n$ provide $\epsilon$-edge LDP. 
In addition, the existence of edge $(v_i, v_j) \in E$ 
$(i > j)$ 
affects only one element $a_{i,j}$ in 
the lower triangular part 
of $\bmA$. 
Thus, $(\calR_1, \ldots, \calR_n)$ provides $\epsilon$-relationship DP (not $2\epsilon$). 

Our proposed algorithm in Section~\ref{sub:two_rounds} also has this property; i.e., 
it provides both $\epsilon$-edge LDP and $\epsilon$-relationship DP.

\subsection{Global Sensitivity}
\label{sub:sensitivity}
In this paper, we use the notion of global sensitivity \cite{DP} to provide edge centralized DP or edge LDP.

Let $\calD$ be the set of possible input data of a randomized algorithm. 
In edge centralized DP, $\calD = \calG$. 
In edge LDP, $\calD = \{0,1\}^n$. 
Let $f: \calD \rightarrow \reals$ be a function that takes data $D \in \calD$ as input and outputs some statistics $f(D) \in \reals$ about the data. 
The most basic method for providing DP is to add the Laplacian noise proportional to the global sensitivity \cite{DP}.

\begin{definition} [Global sensitivity] \label{def:global_sen} 
The global sensitivity of a function $f: \calD \rightarrow \reals$ is given by:
\begin{align*}
GS_f = \underset{D,D' \in \calD: D \sim D'}{\max} |f(D) - f(D')|,
\end{align*}
where $D \sim D'$ represents that $D$ and $D'$ are neighbors; i.e., they differ in one edge in edge centralized DP, and differ in one bit in edge LDP.
\end{definition}

In graphs, the global sensitivity $GS_f$ can be very large. 
For example, adding one edge may result in the increase of triangle (resp.~$k$-star) counts by $n-2$ (resp.~$\binom{n}{k-1}$). 

One way to significantly reduce the global sensitivity is to use \textit{graph projection} \cite{Day_SIGMOD16,Kasiviswanathan_TCC13,Raskhodnikova_arXiv15}, which removes some neighbors from a neighbor list so that the maximum degree $d_{max}$ is upper-bounded by a predetermined value $\td_{max} \in \nnints$. 
By using the graph projection with $\td_{max} \ll n$, we can enforce small global sensitivity; e.g., the global sensitivity of triangle (resp.~$k$-star) counts is at most $\td_{max}$ (resp.~$\binom{\td_{max}}{k-1}$) after the projection. 

Ideally, we would like to set $\td_{max} = d_{max}$ to avoid removing neighbors from a neighbor list (i.e., to avoid the loss of utility). 
However, the maximum degree $d_{max}$ can leak some information about the original graph $G$. 
In this paper, we address this issue by privately estimating $d_{max}$ with edge LDP and then using the private estimate of $d_{max}$ as $\td_{max}$.
This technique 
is also known as 
\textit{adaptive clipping} 
in differentially private stochastic gradient descent (SGD) \cite{Pichapati_arXiv19,Thakkar_arXiv19}.

\subsection{Graph Statistics and Utility Metrics}
\label{sub:graph_statistics}

\noindent{\textbf{Graph statistics.}}~~We consider a graph function that takes a graph $G \in \calG$ as input and outputs some graph statistics. 
Specifically, 
let $f_\triangle: \calG \rightarrow \nnints$ be a graph function that outputs the number of triangles in $G$. 
For $k \in \nats$, let $f_{k\star}: \calG \rightarrow \nnints$ be a graph function that outputs the number of $k$-stars in $G$. 
For example, if a graph $G$ is as shown in Figure~\ref{fig:subgraph}, then $f_\triangle(G) = 5$, $f_{2\star}(G) = 20$, and $f_{3\star}(G) = 8$. The clustering coefficient can also be calculated from $f_\triangle(G)$ and $f_{2\star}(G)$ as: $\frac{3 f_\triangle(G)}{f_{2\star}(G)} = 0.75$.

\begin{table}[t]
\caption{Basic notations in this paper.} 
\centering
\hbox to\hsize{\hfil
\begin{tabular}{l|l}
\hline
Symbol		&	Description\\
\hline
$n$         &	    Number of users.\\
$G=(V,E)$   &	    Graph with $n$ nodes (users) $V$ and edges $E$.\\
$v_i$       &       $i$-th user in $V$.\\
$d_{max}$   &       Maximum degree of $G$.\\
$\td_{max}$   &       Upper-bound on $d_{max}$ (used for projection).\\
$\calG$     &       Set of possible graphs on $n$ users.\\
$\bmA=(a_{i,j})$	    &		Adjacency matrix.\\
$\bma_i$	&		$i$-th row of $\bmA$ (i.e., neighbor list of $v_i$).\\
$\calR_i$     &       Randomized algorithm on $\bma_i$.\\
$f_\triangle(G)$   &  Number of triangles in $G$.\\
$f_{k\star}(G)$    &  Number of $k$-stars in $G$.\\
\hline
\end{tabular}
\hfil}
\label{tab:notations}
\end{table}

Table~\ref{tab:notations} shows the basic notations used in this paper.

\smallskip
\noindent{\textbf{Utility metrics.}}~~We use the $l_2$ loss (i.e., squared error) \cite{Kairouz_ICML16,Wang_USENIX17,Murakami_USENIX19} and the relative error \cite{Bindschaedler_SP16,Chen_CCS12,Xiao_SIGMOD11} as utility metrics. 

Specifically, let 
$\hf(G) \in \reals$ be an estimate of graph statistics $f(G) \in \reals$. 
Here $f$ can be instantiated by 
$f_\triangle$ or $f_{k\star}$; 
i.e., 
$\hf_\triangle(G)$ and $\hf_{k\star}(G)$ are the estimates of $f_\triangle(G)$ and $f_{k\star}(G)$, respectively. 
Let $l_2^2$ be the $l_2$ loss function, which maps the estimate $\hf(G)$ and the true value $f(G)$ to the $l_2$ loss; i.e., $l_2^2(\hf(G), f(G)) = (\hf(G) - f(G))^2$. 
Note that when we use a randomized algorithm providing edge LDP (or edge centralized DP), $\hf(G)$ depends on the randomness in the algorithm. 
In our theoretical analysis, we analyze the expectation of the $l_2$ loss over 
the randomness, as with \cite{Kairouz_ICML16,Wang_USENIX17,Murakami_USENIX19}. 

When $f(G)$ is large, the $l_2$ loss can also be large. 
Thus in our experiments, we also evaluate the relative error, along with the $l_2$ loss. 
The relative error is defined as: $\frac{|\hf(G) - f(G)|}{\max\{f(G), \eta\}}$, where $\eta \in \nnreals$ is a very small positive value. 
Following the convention \cite{Bindschaedler_SP16,Chen_CCS12,Xiao_SIGMOD11}, we set $\eta = 0.001n$ 
for $f_\triangle$ and $f_{k\star}$. 

\section{Algorithms}
\label{sec:algorithms}
In the local model, 
there 
are several ways 
to 
model how the data collector interacts with 
the users~\cite{Duchi_FOCS13,Joseph_SODA20,Qin_CCS17}.
The simplest model 
would be 
to assume that 
the data collector sends 
a 
query $\calR_i$ to each user $v_i$ once, 
and then 
each user $v_i$ independently sends an answer $\calR_i(\bma_i)$ to the data collector. 
In this model, there is one-round interaction between each user and the data collector. 
We call this the
\textit{one-round LDP model}. 
For example, the RR 
for a neighbor list in Section~\ref{sub:LDP} assumes this model.

However, in certain cases it may be possible 
for the data collector to send a query to each user multiple times. 
This could allow for more powerful queries that result in more accurate 
subgraph counts 
\cite{Sun_CCS19} 
or more accurate synthetic graphs~\cite{Qin_CCS17}. 
We call this the \textit{multiple-rounds LDP model}.

In 
Sections~\ref{sub:non-interactive_k_stars} and \ref{sub:non-interactive_triangles}, 
we consider the problems of computing $f_{k\star}(G)$ 
and $f_\triangle(G)$ 
for a graph $G \in \calG$ in the 
one-round 
LDP model. 
Our algorithms and bounds highlight limitations of the
one-round 
LDP model. Compared to the centralized graph DP model, the
one-round 
LDP model cannot compute $f_{k\star}(G)$ as accurately.
Furthermore, the algorithm for $f_\triangle(G)$ does not perform 
well. 
In Section~\ref{sub:two_rounds}, we propose a more sophisticated algorithm for computing  $f_\triangle(G)$ in the two-rounds LDP model, and show that it provides much smaller expected $l_2$ loss than the algorithm in the one-round LDP model.
In Section~\ref{sub:lower_bounds}, we show a general result about lower bounds on the expected $l_2$ loss of graph statistics in LDP. 
The proofs of all statements in Section~\ref{sec:algorithms} are given in 
\conference{the full version \cite{Imola_arXiv21}}\arxiv{Appendix~\ref{sec:proof}}.

\subsection{One-Round Algorithms for $k$-Stars}
\label{sub:non-interactive_k_stars}

\noindent{\textbf{Algorithm.}}~~We begin with the problem of computing $f_{k\star}(G)$ in the 
one-round 
LDP model. 
For this model, we introduce a simple algorithm using the Laplacian mechanism, and prove that this algorithm can achieve order optimal expected $l_2$ loss among all one-round LDP algorithms. 

\setlength{\algomargin}{4mm}
\begin{algorithm}
  \SetAlgoLined
  \KwData{Graph $G$ 
  %described by distributed 
  represented as 
  neighbor lists $\bma_1, \ldots, \bma_n \allowbreak \in \{0,1\}^n$, privacy budget $\epsilon \in \nnreals$, $\td_{max} \in \nnints$.}
  %$d_{max}$, $k$.}
  \KwResult{Private estimate of $f_{k\star}(G)$.}
  %\colorB{$\texttt{GraphProjection}(\bma_1, \ldots, \bma_n, \td_{max})$\;}
  %$(\bma_1, \ldots, \bma_n) \leftarrow \texttt{GraphProjection}(\bma_1, \ldots, \bma_n)$\;
  %$\Delta \leftarrow \binom{d_{max}}{k-1}$\;
  $\Delta \leftarrow \binom{\td_{max}}{k-1}$\;
  \For{$i=1$ \KwTo $n$}{
    $\bma_i \leftarrow \texttt{GraphProjection}(\bma_i, \td_{max})$\;
    \tcc{$d_i$ is a degree of user $v_i$.}
    %\tcc{degree of vertex $i$.}
    %$d_i \leftarrow \sum_{j=1}^n a_i^j$\;
    $d_i \leftarrow \sum_{j=1}^n a_{i,j}$\;
    %\colorB{\If{$d_i > \td_{max}$}{$d_i = \td_{max}$\;}}
    $r_i \leftarrow \binom{d_i}{k}$\;
    $\hat{r}_i \leftarrow r_i + \Lap\left(\frac{\Delta}{\epsilon}\right)$\;
    $release(\hat{r}_i)$\;
  }
  \KwRet{$\sum_{i=1}^n \hat{r}_i$}
  \caption{\alg{LocalLap$_{k\star}$}\label{alg:k-stars}}
\end{algorithm}

Algorithm~\ref{alg:k-stars} shows the one-round algorithm for $k$-stars. 
It takes as input a graph $G$ (represented as neighbor lists $\bma_1, \ldots, \bma_n \in \{0,1\}^n$), the privacy budget $\epsilon$, and a non-negative integer $\td_{max} \in \nnints$. 

The parameter $\td_{max}$ plays a role as an upper-bound on the maximum degree $d_{max}$ of $G$. 
Specifically, let $d_i \in \nnints$ be the degree of user $v_i$; i.e., the number of ``$1$''s in her neighbor list $\bma_i$. 
In line 3, user $v_i$
uses a function (denoted by \texttt{GraphProjection}) that performs graph projection \cite{Day_SIGMOD16,Kasiviswanathan_TCC13,Raskhodnikova_arXiv15} for $\bma_i$ as follows. 
If $d_i$ exceeds $\td_{max}$, it randomly 
selects $\td_{max}$ neighbors out of $d_i$ neighbors; otherwise, it uses $\bma_i$ as it is. 
This guarantees that each user's degree never exceeds $\td_{max}$; i.e., $d_i \leq \td_{max}$ after line 4. 

After the graph projection, 
user $v_i$ 
counts the number of $k$-stars $r_i \in \nnints$ of which she is a center (line 5), and 
adds the Laplacian noise 
to 
$r_i$ 
(line 6). 
Here, since adding one edge results in the increase of at most $\binom{\td_{max}}{k-1}$ $k$-stars, the 
sensitivity of 
$k$-star counts for user $v_i$ 
is at most $\binom{\td_{max}}{k-1}$ (after graph projection). 
Therefore, we add $\Lap(\frac{\Delta}{\epsilon})$ to $r_i$, where $\Delta = \binom{\td_{max}}{k-1}$ and 
for $b \in \nnreals$ 
$\Lap(b)$ is a random variable that represents the Laplacian noise with mean $0$ and scale $b$. 
The final answer of Algorithm~\ref{alg:k-stars} is
simply the sum of all the noisy $k$-star counts. 
We denote this algorithm by \alg{LocalLap$_{k\star}$}.

The value of $\td_{max}$ greatly affects the utility. 
If $\td_{max}$ is too large, a large amount of the Laplacian noise would be added. 
If $\td_{max}$ is too small, a great number of neighbors would be reduced by 
graph projection. 
When we have some prior knowledge about the maximum degree $d_{max}$, we can set $\td_{max}$ to an appropriate value. 
For example, 
the maximum number of connections allowed on Facebook is $5000$~\cite{Facebook_Limit}. 
In this case, we can set $\td_{max}=5000$, and then graph projection does nothing. 
Given that the number of Facebook monthly active users is over $2.7$ billion \cite{Facebook_reports20}, $\td_{max}=5000$ is much smaller than $n$. 
For another example, 
if we know that the degree is smaller than $1000$ for most users, then we can set $\td_{max} = 1000$ and perform graph projection for 
users whose degrees exceed $\td_{max}$. 

In some applications, the data collector may not have such prior knowledge about $\td_{max}$. 
In this case, we can 
privately estimate $d_{max}$ by allowing an additional round between each user and the data collector, and use the private estimate of $d_{max}$ as $\td_{max}$. 
We describe how to privately estimate $d_{max}$ with edge LDP at the end of Section~\ref{sub:non-interactive_k_stars}.

\smallskip
\noindent{\textbf{Theoretical properties.}}~~\alg{LocalLap$_{k\star}$} 
has the following guarantees:

\begin{theorem}\label{thm:k-stars_LDP}
  %If the maximum degree $d_{max}$ of $G$ is at most $\td_{max}$, 
  %Algorithm~\ref{alg:k-stars} 
  \alg{LocalLap$_{k\star}$}
  provides $\epsilon$-edge LDP.
\end{theorem}

\begin{theorem}\label{thm:k-stars}
  Let
  % $A(G,k,d_{max}, \epsilon)$ 
  $\hf_{k\star}(G, \epsilon, \td_{max})$ 
  be the output of 
  %Algorithm~\ref{alg:k-stars}. 
  \alg{LocalLap$_{k\star}$}. 
  %If $\td_{max} \geq d_{max}$, 
  Then, 
  %then 
  for all 
  %$d_{max},k \in \nats,\epsilon \in \nnreals$, 
  $k \in \nats,\epsilon \in \nnreals,\td_{max} \in \nnints$, 
  and $G \in \calG$
  such that the maximum degree $d_{max}$ of $G$ 
  is at most 
  $\td_{max}$, 
  %$d_{max} \in \nats$, 
  %$l_2^2(A(G,k,d_{max}, \epsilon), f_{k\star}(G)) = 
  %O\left( \frac{nd_{max}^{2k-2}}{\epsilon^2} \right)$. 
  $\mathbb{E}[l_2^2(\hf_{k\star}(G, \epsilon, \td_{max}), f_{k\star}(G))] = 
  %O\left( \frac{nd_{max}^{2k-2}}{\epsilon^2} \right)$. 
  O\left( \frac{n \td_{max}^{2k-2}}{\epsilon^2} \right)$. 
  % Furthermore, Algorithm~\ref{alg:k-stars} provides $\epsilon$-edge LDP.
\end{theorem}

The factor of $n$ in the 
expected $l_2$ loss 
of 
\alg{LocalLap$_{k\star}$} 
comes from the fact that we are adding 
the Laplacian noise 
$n$ times. 
In the centralized model, this factor of $n$ is not there, because the central data collector sees all $k$-stars; i.e., the data collector knows $f_{k\star}(G)$. 
The 
sensitivity of $f_{k\star}$ is 
at most $2\binom{\td_{max}}{k-1}$ (after graph projection) under edge centralized DP. 
Therefore, we can consider an algorithm that simply adds the Laplacian noise $\Lap(2\binom{\td_{max}}{k-1}/\epsilon)$ to $f_{k\star}(G)$, and outputs $f_{k\star}(G) + \Lap(2\binom{\td_{max}}{k-1}/\epsilon)$. 
We denote this algorithm by \alg{CentralLap$_{k\star}$}. 
Since the bias of the Laplacian noise is $0$, 
\alg{CentralLap$_{k\star}$} attains the expected $l_2$ loss ($=$ variance) of $O\left(\frac{\td_{max}^{2k-2}}{\epsilon^2}\right)$. 

It seems impossible to avoid this factor of $n$ in the 
one-round 
LDP model, as the data collector will be dealing with $n$ independent answers to
queries. Indeed, this is the case---we prove that the expected $l_2$ error of \alg{LocalLap$_{k\star}$} 
is order optimal among all 
one-round 
LDP algorithms, and 
the 
one-round 
LDP model cannot
improve 
the factor of $n$.

\begin{corollary}\label{cor:kstars-lb}
  Let 
  %$A(G,k,d_{max},\epsilon)$ 
  $\hf_{k\star}(G,\td_{max},\epsilon)$
  be any 
  %non-interactive graph 
  one-round 
  LDP algorithm that 
  computes $f_{k\star}(G)$ satisfying $\epsilon$-edge LDP. Then, for all
  $k,n,\td_{max} \in \nats$ and $\epsilon \in \nnreals$ such that 
  %$n$ and $\td_{max}$ are even, 
  $n$ is even, 
  %and $\td_{max} \geq 3$, 
  there exists a set of graphs 
  %$\mathcal{G}$
  $\calA \subseteq \mathcal{G}$ 
  on $n$ 
  %vertices 
  nodes 
  such that 
  the maximum degree of each 
  %$G \in \mathcal{G}$ 
  $G \in \calA$ 
  is
  at most $\td_{max}$,
  %between $\td_{max}-3$ and $\td_{max}$,
  %$d_{max}$, 
  and 
  %$l_2^2(A(G,k,d_{max}, \epsilon), f_{k\star}(G)) 
  $\frac{1}{|\calA|}\sum_{G \in \calA} \E[l_2^2(\hf_{k\star}(G,\td_{max},\epsilon), f_{k\star}(G))] 
  \geq 
  \Omega\left(\frac{e^{2\epsilon}}{(e^{2\epsilon}+1)^2}\td_{max}^{2k-2}n \right)$.
\end{corollary}

This is a corollary of a more general result of Section~\ref{sub:lower_bounds}. Thus,
any algorithm computing $k$-stars cannot avoid the factor of $n$ in its $l_2^2$
loss. $k$-stars 
is an example where the non-interactive graph LDP model is strictly weaker than
the centralized DP model.

Nevertheless, we note that \alg{LocalLap$_{k\star}$} can accurately calculate $f_{k\star}(G)$ for a large number of users $n$. 
Specifically, the relative error decreases with increase in $n$ 
because \alg{LocalLap$_{k\star}$} has a factor of $n$ (not $n^2$) in the expected $l_2$ error, i.e., 
$\mathbb{E}[(\hf_{k\star}(G, \epsilon, \td_{max}) - f_{k\star}(G))^2] = O(n)$ and $f_{k\star}(G)^2 \geq \Omega(n^2)$ (when we ignore $\td_{max}$ and $\epsilon$). 
In our experiments, we show that the relative error of \alg{LocalLap$_{k\star}$} is 
small when $n$ is large.

\smallskip
\noindent{\textbf{Private calculation of $d_{max}$.}}~~By allowing an additional round between each user and the data collector, we can privately estimate $d_{max}$ and use the private estimate of $d_{max}$ as $\td_{max}$. 
Specifically, 
we divide the privacy budget $\epsilon$ into 
$\epsilon_0 \in \nnreals$ and $\epsilon_1 \in \nnreals$; i.e., $\epsilon = \epsilon_0 + \epsilon_1$. 
We first estimate $d_{max}$ with $\epsilon_0$-edge LDP and then run \alg{LocalLap$_{k\star}$} with the remaining privacy budget $\epsilon_1$. 
Note that \alg{LocalLap$_{k\star}$} with the private calculation of $d_{max}$ results in a two-rounds LDP algorithm.

We consider the following simple algorithm. 
At the first round, 
each user $v_i$ adds the Laplacian noise $\Lap(\frac{1}{\epsilon_0})$ to her degree $d_i$. 
Let $\hd_i \in \reals$ be the noisy degree of $v_i$; i.e., $\hd_i = d_i + \Lap(\frac{1}{\epsilon_0})$. 
Then user $v_i$ sends $\hd_i$ to the data collector. 
Let $\hd_{max} \in \reals$ be the maximum value of the noisy degree; i.e., $\hd_{max} = \max\{\hd_1, \ldots, \hd_n\}$. 
We call $\hd_{max}$ the \textit{noisy max degree}. 
The data collector calculates the noisy max degree $\hd_{max}$ 
as 
an estimate of $d_{max}$, 
and sends $\hd_{max}$ back to all users. 
At the second round, we run \alg{LocalLap$_{k\star}$} 
with input $G$, 
$\epsilon$, and $\lfloor \hd_{max} \rfloor$.

At the first round, the calculation of $\hd_{max}$  provides $\epsilon_0$-edge LDP because each user's degree has the sensitivity $1$ under edge LDP. 
At the second round, Theorem~\ref{thm:k-stars_LDP} guarantees that 
\alg{LocalLap$_{k\star}$} provides $\epsilon_1$-edge LDP. 
Then by the composition theorem~\cite{DP}, this two-rounds algorithm provides $\epsilon$-edge LDP in total ($\epsilon =\epsilon_0 + \epsilon_1$). 

In our experiments, we show that this algorithm provides the utility close to \alg{LocalLap$_{k\star}$} with the true max degree $d_{max}$. 

\subsection{One-Round Algorithms for Triangles.}
\label{sub:non-interactive_triangles}
\noindent{\textbf{Algorithm.}}~~Now, we focus our attention on the more challenging $f_\triangle$ query. This
query is more challenging in the graph LDP model because no user is aware of any
triangle; i.e., user $v_i$ is not aware of any triangle formed by $(v_i, v_j, v_k)$, because $v_i$ cannot see any edge $(v_j, v_k) \in E$ in graph $G$. 

One way to count $f_\triangle(G)$ with edge LDP is 
to apply the RR (Randomized Response) 
to a neighbor list. 
For example, user $v_i$ applies the RR to 
$a_{i,1}, \ldots, a_{i,i-1}$ (which corresponds to users $v_1, \ldots, v_{i-1}$ with smaller user IDs) in her neighbor list $\bma_i$; i.e., 
we apply the RR to the lower triangular part of adjacency matrix $\bmA$, as described in Section~\ref{sub:LDP}. 
Then the data collector constructs a noisy graph $G'=(V,E') \in \calG$ from the lower triangular part of the noisy adjacency matrix, and 
estimates the number of triangles 
from $G'$. 
However, 
simply counting
the triangles in 
$G'$ 
can introduce a significant bias 
because $G'$ is denser than $G$ especially when $\epsilon$ is small. 

\begin{figure}
\centering
\includegraphics[width=0.9\linewidth]{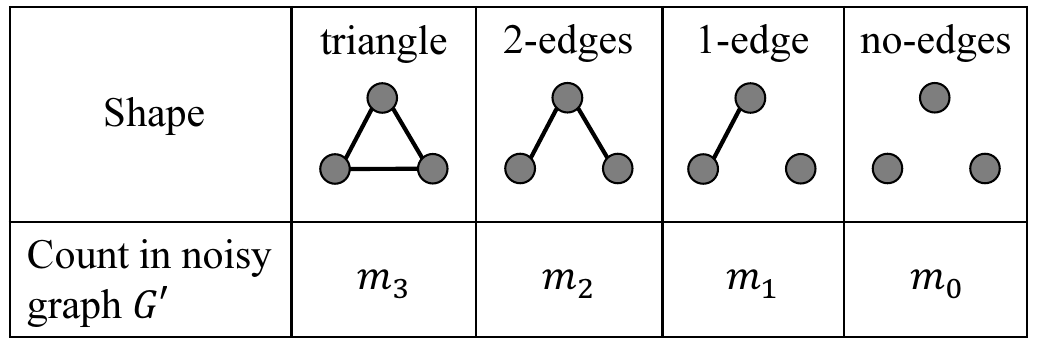}
\vspace{-4mm}
\caption{Four types of subgraphs with three nodes.}
\label{fig:triplet_shape}
\end{figure}

Through 
clever post-processing 
known as 
empirical estimation 
\cite{Kairouz_ICML16,Murakami_USENIX19,Wang_USENIX17},
we are able to obtain an unbiased estimate of $f_\triangle(G)$ 
from $G'$. 
Specifically, a subgraph with three nodes can be divided into four types depending on the number of edges. 
Three nodes with three edges form a triangle. 
We refer to three nodes with two edges, one edge, and no edges as \textit{2-edges},  \textit{1-edge}, and  \textit{no-edges}, respectively. 
Figure~\ref{fig:triplet_shape} shows their shapes. 
$f_\triangle(G)$ can be expressed using $m_3$, $m_2$, $m_1$, and $m_0$ as follows:

\begin{proposition}\label{prop:triangle_emp}
  %Let $\mu = e^\epsilon$. 
  Let $G'=(V,E')$ be a noisy graph generated by applying the RR to the lower triangular part of $\bmA$.
  Let $m_3, m_2, m_1, m_0 \in \nnints$ be respectively the number of triangles, 2-edges, 1-edge, and no-edges in $G'$. 
  Then 
  \begin{align}
      %\textstyle{\mathbb{E}\left[ \frac{\mu^3}{(\mu-1)^3} m_3 - \frac{\mu^2}{(\mu-1)^3} m_2 + \frac{\mu}{(\mu-1)^3} m_1 - \frac{1}{(\mu-1)^3} m_0 \right] = f_\triangle(G).}
      \textstyle{\mathbb{E}\left[ \frac{e^{3\epsilon}}{(e^\epsilon-1)^3} m_3 \hspace{-0.5mm}-\hspace{-0.5mm} \frac{e^{2\epsilon}}{(e^\epsilon-1)^3} m_2 \hspace{-0.5mm}+\hspace{-0.5mm} \frac{e^\epsilon}{(e^\epsilon-1)^3} m_1 \hspace{-0.5mm}-\hspace{-0.5mm} \frac{1}{(e^\epsilon-1)^3} m_0 \right] \hspace{-0.5mm} = \hspace{-0.5mm} f_\triangle(G).}
      \label{eq:triangle_emp}
  \end{align}
\end{proposition}

Therefore, the data collector can count $m_3$, $m_2$, $m_1$, and $m_0$ from $G'$, and calculate an unbiased estimate of $f_\triangle(G)$ by (\ref{eq:triangle_emp}). 
In Appendix~\ref{sec:RR_emp}, we show that the $l_2$ loss is significantly reduced by this empirical estimation.

\setlength{\algomargin}{4mm}
\begin{algorithm}
  \SetAlgoLined
  \KwData{Graph $G$ 
  %described by distributed 
  represented as 
  neighbor lists $\bma_1,
    \ldots, \bma_n
  \in \{0,1\}^n$, privacy budget $\epsilon \in \nnreals$.}
  \KwResult{Private estimate of $f_\triangle(G)$.}
  \For{$i=1$ \KwTo $n$}{
    %$R_i \leftarrow (RR_{\epsilon}(a_i^1), \ldots, RR_{\epsilon}(a_i^{i-1}))$\;
    $R_i \leftarrow (RR_{\epsilon}(a_{i,1}), \ldots, RR_{\epsilon}(a_{i,i-1}))$\;
    $release(R_i)$\;
  }
  %$G' \leftarrow \texttt{UndirectedGraph}(R_1, \ldots, R_{i-1})$\;
  $G'=(V,E') \leftarrow \texttt{UndirectedGraph}(R_1, \ldots, R_n)$\;
  %\tcc{Counts $T_3,T_2,T_1,T_0$ in $P$.}
  \tcc{Counts $m_3,m_2,m_1,m_0$ in $G'$.}
  %$\hat{\textbf{m}} \leftarrow \texttt{CountSubgraphs}(G',3)$\;
  $(m_3, m_2, m_1, m_0) \leftarrow \texttt{Count}(G')$\;
  %$\mu \leftarrow e^\epsilon$\;
  %\KwRet{$\frac{1}{(\mu-1)^3}(\mu^3 m_3 -\mu^2 m_2 + \mu m_1 - m_0)$}
  \KwRet{$\frac{1}{(e^\epsilon-1)^3}(e^{3\epsilon} m_3 - e^{2\epsilon} m_2 + e^\epsilon m_1 - m_0)$}

  %\caption{CountSubgraphsRR\label{alg:subgraph-rr}}
  \caption{\alg{LocalRR$_\triangle$}\label{alg:subgraph-rr}}
\end{algorithm}

Algorithm~\ref{alg:subgraph-rr} shows this algorithm. 
In line 2, user $v_i$ applies the RR with privacy budget $\epsilon$ (denoted by $RR_\epsilon$) to $a_{i,1}, \ldots, a_{i,i-1}$ 
in her neighbor list $\bma_i$, and outputs $R_i = (RR_\epsilon(a_{i,1}), \ldots, RR_\epsilon(a_{i,i-1}))$. 
In other words, we apply the RR to the lower triangular part of $\bmA$ and there is no overlap between edges sent by users. 
In line 5, the data collector uses a function (denoted by \texttt{UndirectedGraph}) 
that converts the bits of $(R_1, \ldots, R_n)$ into an undirected graph $G'
= (V, E')$ by adding edge $(v_i,v_j)$ with $i>j$ to $E'$ if and only if the $j$-th bit of
$R_i$ is $1$. 
Note that $G'$ is biased, as explained above. 
In line 6, the data collector uses a function (denoted by 
\texttt{Count}) that calculates $m_3$, $m_2$, $m_1$, and $m_0$ from $G'$. 
Finally, the data collector outputs the expression inside the expectation on
the left-hand side of (\ref{eq:triangle_emp}), which is an unbiased estimator for 
$f_\triangle(G)$ by Proposition~\ref{prop:triangle_emp}.
We denote this algorithm by \alg{LocalRR$_\triangle$}.

\smallskip
\noindent{\textbf{Theoretical properties.}}~~\alg{LocalRR$_\triangle$} 
provides the following guarantee.

\begin{theorem}\label{thm:subgraph-rr_LDP}
  \alg{LocalRR$_\triangle$} provides $\epsilon$-edge LDP and $\epsilon$-relationship DP.
\end{theorem}

\alg{LocalRR$_\triangle$} does not have the doubling issue (i.e., it provides not $2\epsilon$ but $\epsilon$-relationship DP), because we apply the RR to the lower triangular part of $\bmA$, as explained in Section~\ref{sub:LDP}.

Unlike the RR and empirical estimation for tabular data \cite{Kairouz_ICML16}, the expected $l_2$ loss of \alg{LocalRR$_\triangle$} is complicated. 
To simplify the utility analysis, we assume that $G$ is generated from the Erd\"os-R\'enyi graph distribution $\bmG(n,\alpha)$ with edge existence probability $\alpha$; i.e., each edge in $G$ with $n$ nodes is independently generated with probability $\alpha \in [0,1]$.

\begin{theorem}\label{thm:subgraph-rr}
  Let $\bmG(n,\alpha)$ be the Erd\"os-R\'enyi graph distribution with edge existence probability $\alpha \in [0,1]$. 
  Let $p = \frac{1}{e^\epsilon+1}$ and 
  $\beta = \alpha(1-p) + (1-\alpha)p$. 
  Let 
  %$A(G, \epsilon)$ 
  $\hf_{\triangle}(G, \epsilon)$ 
  be the output of 
  %Algorithm~\ref{alg:subgraph-rr}.
  \alg{LocalRR$_\triangle$}.
  %Suppose 
  %$G \sim \bmG(n,\frac{d_{max}}{n})$, 
  If 
  $G \sim \bmG(n,\alpha)$, 
  %the Erd\"os-R\'enyi graph
  %distribution with parameter $\frac{d_{max}}{n}$.  Then, 
  then for all 
  $\epsilon \in \nnreals$, 
  %$l_2^2(A(G,\epsilon),
  $\mathbb{E}[l_2^2(\hf_{\triangle}(G, \epsilon),
  f_\triangle(G))] = 
  %O(\frac{e^{6\epsilon}}{(e^\epsilon-1)^6}d_{max}n^3)$.
  O\left(\frac{e^{6\epsilon}}{(e^\epsilon-1)^6}\beta n^4\right)$.
  %Furthermore, Algorithm~\ref{alg:subgraph-rr} provides $\epsilon$-edge LDP.
\end{theorem}

Note that we assume the Erd\"os-R\'enyi model only for the utility analysis of \alg{LocalRR$_\triangle$}, and do not assume this model for the other algorithms. 
The upper-bound of \alg{LocalRR$_\triangle$} in Theorem~\ref{thm:subgraph-rr} is less ideal than the upper-bounds of the other algorithms in that it does not consider all possible graphs $G \in \calG$. 
Nevertheless, we also show that the $l_2$ loss of \alg{LocalRR$_\triangle$} is roughly consistent with Theorem~\ref{thm:subgraph-rr} in our experiments using two real datasets (Section~\ref{sec:experiments}) and 
the Barab\'{a}si-Albert graphs \cite{NetworkScience}, which have power-law degree distribution (Appendix~\ref{sec:BAGraph}). 

The parameters $\alpha$ and $\beta$ are edge existence probabilities in the original graph $G$ and noisy graph $G'$, respectively. 
Although $\alpha$ is very small in a sparse graph, $\beta$ can be large for small $\epsilon$. 
For example, if $\alpha \approx 0$ and $\epsilon=1$, then $\beta \approx \frac{1}{e+1} = 0.27$. 

Theorem~\ref{thm:subgraph-rr} states that for large $n$, the $l_2$ loss of \alg{LocalRR$_\triangle$} 
($=O(n^4)$) 
is much larger than the $l_2$ loss of \alg{LocalRR$_k\star$} ($=O(n)$). 
This follows from the fact that user $v_i$ 
is not aware of any triangle formed by $(v_i, v_j, v_k)$, as explained above. 

In contrast, counting $f_\triangle(G)$ in the centralized model is much easier because the data collector sees all triangles in $G$; i.e., the data collector knows $f_\triangle(G)$. 
The 
sensitivity of $f_\triangle$ 
is 
at most $\td_{max}$ (after graph projection). 
Thus, 
we can consider a simple algorithm that 
outputs $f_{\triangle}(G) + \Lap(\td_{max}/\epsilon)$. 
We denote this algorithm by \alg{CentralLap$_{\triangle}$}. 
\alg{CentralLap$_{\triangle}$} attains the expected $l_2$ loss ($=$ variance) of $O\left(\frac{\td_{max}^2}{\epsilon^2}\right)$.

The large $l_2$ loss of \alg{LocalRR$_\triangle$} is caused by the fact that 
each edge is released independently with
some probability of being flipped. 
In other words, 
there are three independent random
variables that influence 
any triangle in $G'$. 
The next algorithm,
using interaction, 
reduces 
this influencing number 
from three to one 
by using the fact that 
a user 
knows 
the existence of two edges for any triangle that involves the user.

\subsection{Two-Rounds Algorithms for Triangles}
\label{sub:two_rounds}

\noindent{\textbf{Algorithm.}}~~Allowing for 
two-rounds interaction, 
we are able to compute $f_{\triangle}$ with
a significantly improved $l_2$ loss, albeit with a higher per-user
communication overhead.
As described in Section~\ref{sub:non-interactive_triangles}, it is impossible for user $v_i$ to see edge $(v_j, v_k) \in E$ in graph $G=(V,E)$ at the first round. 
However, if 
the data collector publishes a noisy graph $G'=(V,E')$ calculated by \alg{LocalRR$_\triangle$} at the first round, then 
user $v_i$ can see a noisy edge $(v_j, v_k) \in E'$ in the noisy graph $G'$ at the second round. 
Then user $v_i$ can count the number of \textit{noisy triangles} formed by
$(v_i, v_j, v_k)$ such that $(v_i,v_j) \in E$, $(v_i,v_k) \in E$, and $(v_j,v_k)
\in E'$, and send the noisy triangle counts with the Laplacian noise to the data
collector in an analogous way to \alg{LocalLap$_{k\star}$}.
Since user $v_i$ always knows that two edges $(v_i,v_j)$ and $(v_i,v_k)$ exist in $G$, 
there is only one noisy edge in any noisy triangle 
(whereas all edges are noisy in \alg{LocalRR$_\triangle$}).
This is an intuition behind our proposed two-rounds algorithm. 

As with the RR in Section~\ref{sub:non-interactive_triangles}, simply counting the noisy triangles can introduce a bias. 
Therefore, we calculate an empirical estimate of $f_\triangle(G)$ from the noisy triangle counts. 
Specifically, 
the following is the empirical estimate of $f_\triangle(G)$: 

\begin{proposition}\label{prop:triangle_emp_2rounds}
  Let $G'=(V,E')$ be a noisy graph generated by applying the RR with privacy budget $\epsilon_1 \in \nnreals$   to the lower triangular part of $\bmA$.
  Let $p_1 = \frac{1}{e^{\epsilon_1}+1}$. 
  Let $t_i \in \nnints$ be the number of triplets $(v_i, v_j, v_k)$ such that 
  %$j < i < k$, 
  $j < k < i$, 
  $(v_i,v_j) \in E$, $(v_i,v_k) \in E$, and $(v_j,v_k) \in E'$.
  Let $s_i \in \nnints$ be the number of triplets $(v_i, v_j, v_k)$ such that 
  %$j < i < k$, 
  $j < k < i$, 
  $(v_i,v_j) \in E$, and $(v_i,v_k) \in E$. 
  Let $w_i = t_i - p_1 s _i$. 
  Then 
  \begin{align}
      %\textstyle{\mathbb{E}[f_\triangle(G)] = \frac{1}{1-2p_1} \sum_{i=1}^n (t_i - p_1 s_i)}.
      \textstyle{\mathbb{E}\left[ \frac{1}{1-2p_1} \sum_{i=1}^n w_i \right] = f_\triangle(G).}
      \label{eq:triangle_emp_2rounds}
  \end{align}
\end{proposition}

Note that in Proposition~\ref{prop:triangle_emp_2rounds}, 
we count only triplets $(v_i, v_j, v_k)$ with 
$j < k < i$ 
to use only the lower triangular part of $\bmA$. 
$t_i$ is the number of noisy triangles user $v_i$ can see at the second round. 
$s_i$ is the number of $2$-stars of which user $v_i$ is a center. 
Since $t_i$ and $s_i$ can reveal information about an edge in $G$, user $v_i$ adds the Laplacian noise to $w_i$ $(= t_i - p_1 s _i)$ in (\ref{eq:triangle_emp_2rounds}), and sends it to the data collector. 
Then the data collector calculates an unbiased estimate of $f_\triangle(G)$ by (\ref{eq:triangle_emp_2rounds}).

\setlength{\algomargin}{5mm}
\begin{algorithm}
  \SetAlgoLined
  \KwData{Graph $G$ 
  %described by distributed 
  represented as 
  neighbor lists $\bma_1,
    \ldots, \bma_n
    \in \{0,1\}^n$, two privacy budgets
  %$\epsilon_0,\epsilon_1 > 0$.}
  $\epsilon_1,\epsilon_2 > 0$, $\td_{max} \in \nnints$.}
  %\KwResult{Private count of number of triangles in $G$.}
  \KwResult{Private estimate of $f_\triangle(G)$.}
  %$\rho \leftarrow \frac{1}{e^{\epsilon_0}+1}$\;
  $p_1 \leftarrow \frac{1}{e^{\epsilon_1}+1}$\;
  \tcc{First round.}
  \For{$i=1$ \KwTo $n$}{
    %$ans_i \leftarrow 0$\;
    %$R_i \leftarrow (RR_{\epsilon_0}(a_i^1), \ldots, RR_{\epsilon_0}(a_i^{i-1}))$\;
    $R_i \leftarrow (RR_{\epsilon_1}(a_{i,1}), \ldots, RR_{\epsilon_1}(a_{i,i-1}))$\;
    $release(R_i)$\;
  }
  %$P_{i-1} \leftarrow \texttt{UndirectedGraph}(R_1, \ldots, R_{i-1})$\;
  $G'=(V,E') \leftarrow \texttt{UndirectedGraph}(R_1, \ldots, R_{i-1})$\;
  \tcc{Second round.}
  \For{$i=1$ \KwTo $n$}{
    $\bma_i \leftarrow \texttt{GraphProjection}(\bma_i, \td_{max})$\;
    %$t_i \leftarrow |\{(u,v) : u,v \in [i], u<v<i, a_{i}^u = a_{i}^v = 1,(u,v) \in P_{i-1}\}|$\;
    $t_i \leftarrow |\{(v_i,v_j,v_k) : 
    %j<i<k, a_{j,i} = a_{i,k} = 1, 
    j<k<i, a_{i,j} = a_{i,k} = 1, 
    (v_j,v_k) \in E'\}|$\;
    %$s_i \leftarrow |\{(u,v) : u,v \in [i], u<v<i, a_{i}^u = a_{i}^v = 1\}|$\;
    $s_i \leftarrow |\{(v_i,v_j,v_k) : 
    %j<i<k, a_{j,i} = a_{i,k} = 1\}|$\;
    j<k<i, a_{i,j} = a_{i,k} = 1\}|$\;
    %$w_i \leftarrow t_i - \rho s_i + Lap(\frac{d_{max}(1-\rho)}{\epsilon_1})$\;
    %$w_i \leftarrow t_i - p_1 s_i + \Lap(\frac{\td_{max}(1-p_1)}{\epsilon_2})$\;
    $w_i \leftarrow t_i - p_1 s_i$\;
    $\hw_i \leftarrow w_i + \Lap(\frac{\td_{max}}{\epsilon_2})$\;
    %$release(R_i, w_i)$\;
    $release(\hw_i)$\;
  }
  %$ans \leftarrow \frac{1}{1-2\rho}\sum_{i=1}^n w_i$\;
  \KwRet{$\frac{1}{1-2p_1}\sum_{i=1}^n \hw_i$}
  %\caption{CountSubgraphsTwoRound\label{alg:subgraph-interactive}}
  \caption{\alg{Local2Rounds$_\triangle$}}\label{alg:subgraph-interactive}
\end{algorithm}

Algorithm~\ref{alg:subgraph-interactive} contains the formal
description of this process. 
It takes as input a graph $G$, 
the privacy budgets $\epsilon_1, \epsilon_2 \in \nnreals$ at the first and second rounds, respectively, 
and 
a non-negative integer $\td_{max} \in \nnints$. 
At the first round, we 
apply the RR to the lower triangular part of $\bmA$ 
(i.e., there is no overlap between edges sent by users) 
and use the \texttt{UndirectedGraph} function to 
obtain a noisy graph $G'=(V,E')$ by the RR in the same way as Algorithm~\ref{alg:subgraph-rr}. 
Note that $G'$ is biased. 
We calculate an unbiased estimate of $f_\triangle(G)$ from $G'$ at the second round.  

At the second round, each user $v_i$ 
calculates $\hw_i = w_i + \Lap(\frac{\td_{max}}{\epsilon_2})$ 
by adding the Laplacian noise to $w_i$ 
in Proposition~\ref{prop:triangle_emp_2rounds} 
whose 
sensitivity is at most 
$\td_{max}$ 
(as we will prove in Theorem~\ref{thm:local2rounds_LDP}). 
Finally, we output $\frac{1}{1-2p_1}\sum_{i=1}^n \hw_i$, which is an unbiased estimate of $f_\triangle(G)$ by Proposition~\ref{prop:triangle_emp_2rounds}. 
We call this algorithm \alg{Local2Rounds$_\triangle$}.

\smallskip
\noindent{\textbf{Theoretical properties.}}~~\alg{Local2Rounds$_\triangle$} 
has the following 
guarantee.

\begin{theorem}\label{thm:local2rounds_LDP}
  %If the maximum degree $d_{max}$ of $G$ is at most $\td_{max}$, 
  \alg{Local2Rounds$_\triangle$}
  provides $(\epsilon_1 + \epsilon_2)$-edge LDP and $(\epsilon_1 + \epsilon_2)$-relationship DP.
\end{theorem}

As with \alg{LocalRR$_\triangle$}, \alg{Local2Rounds$_\triangle$} does not have the doubling issue; i.e., it provides $\epsilon$-relationship DP (not $2\epsilon$). 
This follows from the fact that we 
use only the lower triangular part of $\bmA$; 
i.e., we assume 
$j<k<i$ 
in counting $t_i$ and $s_i$. 

\begin{theorem}\label{thm:local2rounds}
  Let 
  %$A(G, d_{max}, \epsilon_0, \epsilon_1)$ 
  $\hf_{\triangle}(G, \epsilon_1, \epsilon_2, \td_{max})$ 
  be the output of 
  %Algorithm~\ref{alg:subgraph-interactive}. 
  \alg{Local2Rounds$_\triangle$}. 
  Then, for all 
  %$d_{max} \in \nats$,
  $\epsilon_1,\epsilon_2 \in \nnreals$, 
  $\td_{max} \in \nnints$,
  and $G\in \calG$ such that the
  maximum degree $d_{max}$ of $G$ is 
  at most 
  $\td_{max}$,
  %$l_2^2(A(G,d_{max},\epsilon_0,\epsilon_1) 
  $\mathbb{E}[l_2^2(\hf_{\triangle}(G, \epsilon_1, \epsilon_2, \td_{max}), f_\triangle(G))] 
  \leq
    O\left(\frac{e^{\epsilon_1}}{(1-e^{\epsilon_1})^2} \left(\td_{max}^3 n +
    \frac{e^{\epsilon_1}}{\epsilon_2^2}\td_{max}^2 n\right)\right)$.
\end{theorem}

Theorem~\ref{thm:local2rounds} means that for triangles, the $l_2$ loss 
is reduced from $O(n^4)$ to $O(\td_{max}^3n)$ by introducing an additional round.

\smallskip
\noindent{\textbf{Private calculation of $d_{max}$.}}~~As with $k$-stars, we can privately calculate $d_{max}$ 
by using the method described in Section~\ref{sub:non-interactive_k_stars}. 
Furthermore, the private calculation of $d_{max}$ does not increase the number of rounds; i.e., we can run \alg{Local2Rounds$_\triangle$} with the private calculation of $d_{max}$ in two rounds. 

Specifically, let $\epsilon_0 \in \nnreals$ be the privacy budget for the private calculation of $d_{max}$. 
At the first round, each user $v_i$ adds $\Lap(\frac{1}{\epsilon_0})$ to her degree $d_i$, 
and sends the noisy degree $\hd_i$ ($=d_i + \Lap(\frac{1}{\epsilon_0})$) to the data collector, along with the outputs $R_i = (RR_\epsilon(a_{i,1}), \ldots, RR_\epsilon(a_{i,i-1}))$ of the RR. 
The data collector calculates the noisy max degree $\hd_{max}$ ($= \max\{\hd_1,
\ldots, \hd_n\}$) as an estimate of $d_{max}$, and sends it back to all users. 
At the second round, we run \alg{Local2Rounds$_\triangle$} with input $G$ (represented as $\bma_1, \ldots, \bma_n$), $\epsilon_1$, $\epsilon_2$, and $\lfloor \hd_{max} \rfloor$.

At the first round, the calculation of $\hd_{max}$ provides $\epsilon_0$-edge LDP. 
Note that it provides $2\epsilon_0$-relationship DP (i.e., it has the doubling issue) because one edge $(v_i,v_j) \in E$ affects both of the degrees $d_i$ and $d_j$ by 1. 
At the second round, \alg{LocalLap$_{k\star}$} provides $(\epsilon_1 + \epsilon_2)$-edge LDP and 
$(\epsilon_1 + \epsilon_2)$-relationship DP (Theorem~\ref{thm:local2rounds_LDP}). 
Then by the composition theorem~\cite{DP}, this two-rounds algorithm provides $(\epsilon_0 + \epsilon_1 + \epsilon_2)$-edge LDP and $(2\epsilon_0 + \epsilon_1 + \epsilon_2)$-relationship DP. 
Although the total privacy budget is larger for relationship DP, the difference ($=\epsilon_0$) can be very small. 
In fact, we set $(\epsilon_0, \epsilon_1, \epsilon_2) = (0.1, 0.45, 0.45)$ or $(0.2, 0.9, 0.9)$ in our experiments (i.e., the difference is $0.1$ or $0.2$), and show that this algorithm provides almost the same utility as \alg{Local2Rounds$_\triangle$} with the true max degree $d_{max}$. 

\smallskip
\noindent{\textbf{Time complexity.}}~~We also note that \alg{Local2Rounds$_\triangle$} has an advantage over \alg{LocalRR$_\triangle$} in terms of the time complexity. 

Specifically, \alg{LocalRR$_\triangle$} is inefficient because the data collector has to count the number of triangles $m_3$ in the noisy graph $G'$. 
Since the noisy graph $G'$ is dense (especially when $\epsilon$ is small) and there are $\binom{n}{3}$ subgraphs with three nodes in $G'$, the number of triangles is $m_3 = O(n^3)$. 
Then, the time complexity of \alg{LocalRR$_\triangle$} is 
also $O(n^3)$, which is not practical for a graph with a large number of users $n$.
In fact, we implemented \alg{LocalRR$_\triangle$} ($\epsilon=1$) with C/C++ and measured its running time using one node of a supercomputer (ABCI: AI Bridging Cloud Infrastructure \cite{ABCI}). 
When $n=5000$, $10000$, $20000$, and $40000$, the running time was $138$, $1107$, $9345$, and $99561$ seconds, respectively; i.e., the running time was almost cubic in $n$. 
We can also estimate the running time for larger $n$. 
For example, when $n=1000000$, \alg{LocalRR$_\triangle$} ($\epsilon=1$) would require about $35$ years $(=1107 \times 100^3 /(3600 \times 24 \times 365))$. 

In contrast, the time complexity of \alg{Local2Rounds$_\triangle$} 
is 
$O(n^2 + n d_{max}^2)$\footnote{When we 
evaluate 
\alg{Local2Rounds$_\triangle$} in our experiments, 
we can apply the RR to only edges that are required at the second round; i.e., $(v_j,v_k) \in G'$ in line 8 of Algorithm~\ref{alg:subgraph-interactive}. 
Then the time complexity of \alg{Local2Rounds$_\triangle$} 
can be reduced to 
$O(n d_{max}^2)$ in total. 
We also confirmed that when $n=1000000$, the running time of \alg{Local2Rounds$_\triangle$} was $311$ seconds on one node of the ABCI. 
Note, however, that this does \textit{not} protect individual privacy, because it reveals the fact that users $v_j$ and $v_k$ are friends with $u_i$ to the data collector.}. 
The factor of $n^2$ comes from the fact that the size of the noisy graph $G'$ is $O(n^2)$. 
This also causes a large communication overhead, as explained below.

\smallskip
\noindent{\textbf{Communication overhead.}}~~In 
\alg{Local2Rounds$_\triangle$}, each user need to see the noisy graph $G'$ 
of size $O(n^2)$ 
to count $t_i$ and $s_i$. 
This results in a per-user communication overhead of $O(n^2)$. 
Although we 
do not simulate the communication overhead in our experiments that use \alg{Local2Rounds$_\triangle$}, 
the $O(n^2)$ overhead 
might 
limit its application in very large graphs. 
An interesting avenue of future work is how to compress the graph size (e.g., via graph projection or random projection) to reduce both the time complexity and the communication overhead.

\begin{table*}
  \centering
\begin{tabular}{|l|l|l|l|l|l|l|}
  \hline
  & Centralized & \spantwo{One-round local} & Two-rounds local \\
  \hline
  & Upper Bound & Lower Bound & Upper Bound & Upper Bound \\ \hline

  $f_{k\star}$
  & $O\left( \frac{d_{max}^{2k-2}}{\epsilon^2} \right)$  
  &  $\Omega\left( \frac{e^{2\epsilon}}{(e^{2\epsilon}+1)^2}d_{max}^{2k-2}n \right)$ 
  &  $O\left( \frac{d_{max}^{2k-2}}{\epsilon^2}n \right)$ 
  &  $O\left( \frac{d_{max}^{2k-2}}{\epsilon^2}n \right)$ \\ \hline

 $f_\triangle$ 
  &  $O\left(\frac{d_{max}^2}{\epsilon^2}\right)$ 
  &  $\Omega\left( \frac{e^{2\epsilon}}{(e^{2\epsilon}+1)^2}d_{max}^2n \right)$
  &  $O\left(\frac{e^{6\epsilon}}{(e^{\epsilon}-1)^6}n^4\right)$ 
  (when $G \sim \bmG(n,\alpha)$)
  %(randomized bound)
  &  $O\left(\frac{e^\epsilon}{(e^\epsilon-1)^2}(d_{max}^3 n +
  \frac{e^\epsilon}{\epsilon^2}d_{max}^2 n)\right)$ \\ \hline

\end{tabular}
\vspace{-2mm}
\caption{Bounds on $l_2$ losses for privately estimating $f_{k\star}$ and
$f_{\triangle}$ with $\epsilon$-edge LDP. For upper-bounds, we assume that  $\td_{max}=d_{max}$. 
For the centralized model, we use the Laplace mechanism. For
the one-round $f_\triangle$ algorithm, we apply Theorem~\ref{thm:subgraph-rr} 
with constant $\alpha$. For the two-round protocol $f_\triangle$ algorithm, we
apply Theorem~\ref{thm:local2rounds} with
$\epsilon_1=\epsilon_2=\frac{\epsilon}{2}$. }\label{tab:perf}
\end{table*}

\subsection{Lower Bounds}
\label{sub:lower_bounds}

We show a general lower bound on the $l_2$ loss of private estimators $\hat{f}$ of
real-valued functions $f$ in the one-round LDP model. Treating $\epsilon$ as a
constant, we have shown that 
when $\td_{max}=d_{max}$, 
the expected $l_2$ loss of \alg{LocalLaplace$_{k\star}$} is 
$O(nd_{max}^{2k-2})$
(Theorem~\ref{thm:k-stars}). 
However, in
the centralized 
model, we can use
the Laplace mechanism with sensitivity 
$2\binom{d_{max}}{k-1}$ 
to obtain $l_2^2$ errors of $O(d_{max}^{2k-2})$ 
for $f_{k\star}$. 
Thus, we ask
if 
the factor of $n$ is 
necessary 
in the one-round LDP model.

We 
answer this question affirmatively.
We show for many types of queries $f$, there is a lower bound on 
$l_2^2(f(G), \hf(G))$ 
for any private estimator $\hf$ of the form
\begin{equation}\label{eq:one-round-lower}
  \hf(G) = \tilde{f}(\calR_1(\bma_1), \ldots, \calR_n(\bma_n)),
\end{equation}
where 
$\calR_1, \ldots, \calR_n$ satisfy 
$\epsilon$-edge LDP or $\epsilon$-relationship DP 
and $\tilde{f}$ is an aggregate function that takes $\calR_1(\bma_1), \ldots, \calR_n(\bma_n)$ as input and outputs $\hf(G)$. 
Here we assume that $\calR_1, \ldots, \calR_n$ 
are independently run, meaning that they are in the one-round
setting.
For our 
lower bound,
we 
require 
that 
input edges to $f$ 
be ``independent'' in the sense that 
adding an edge to an input graph $G$  
independently 
change 
$f$ by at least $D \in \reals$. 
The specific structure of input graphs we require is as follows:

\begin{definition}\label{def:mono-cube}[$(n,D)$-independent cube for $f$]
  Let 
  %$D \in \reals$. 
  $D \in \nnreals$. 
  For 
  %$k \in \nats$, 
  $\kappa \in \nats$, 
  let $G=(V,E) \in \calG$ be a graph on 
  %$n = 2k$ 
  $n = 2\kappa$ 
  nodes, and let 
  %$M = \{(v_{i_1}, v_{i_2}),(v_{i_3},v_{i_4}),\ldots,(v_{i_{2k-1}},v_{i_{2k}})\}$ 
  $M = \{(v_{i_1}, v_{i_2}),(v_{i_3},v_{i_4}),\ldots,(v_{i_{2k-1}},v_{i_{2\kappa}})\}$ 
  for integers $i_j \in [n]$ 
  %denote 
  be 
  a set of edges such that each of $i_1, \ldots, i_{2\kappa}$ is distinct (i.e., perfect matching on the nodes).
  %Furthermore,
  Suppose that 
  %the perfect matching 
  $M$ is disjoint from 
  $E$; 
  %$G$; 
  i.e., 
  % , meaning 
  %$(v_{i_{2j-1}}, v_{i_{2j-2}}) \notin G$. 
  %$(v_{i_{2j-1}}, v_{i_{2j}}) \notin G$ 
  $(v_{i_{2j-1}}, v_{i_{2j}}) \notin E$ 
  for any 
  %$j\in[k]$. 
  $j\in[\kappa]$. 
  Let 
  $\calA = \{(V, E \cup N): N \subseteq M\}$.
  %Notice 
  Note that 
  $\calA$ is a set of 
  %$2^k$ 
  $2^\kappa$ 
  graphs.
  We say $\calA$ 
  %defines 
  %forms 
  is 
  an \emph{$(n,D)$-independent cube for $f$} if for all
  %$G_1, G_2 \in \calA$ such that $|G_1| + 1 = |G_2|$, 
  $G'=(V,E') \in \calA$, 
  %$G \in \calA$, 
  we have
  \[
    % f(G) = f(G') + \sum_{e \in M} \textbf{1}[e \in G] C_e
    f(G') = f(G) + \sum_{e \in E' \cap M} C_e,
  \]
  where $C_e \in \reals$ satisfies $|C_e| \geq D$ for any $e \in M$.
\end{definition}

\begin{figure}[t]
  \centering
  \includegraphics[width=0.9\linewidth]{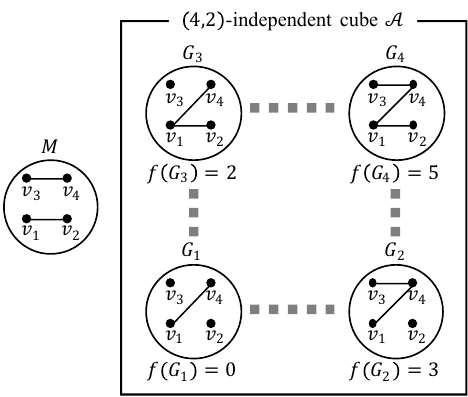}
  \vspace{-2mm}
  \caption{
    $(4,2)$-independent cube $\calA$ for $f$. 
    In this example, $M = \{(v_1,v_2),(v_3,v_4)\}$, $G_1=(V,E)$, $\calA = \{(V, E \cup N): N \subseteq M\}$, 
    $C_{(v_1,v_2)}=2$, and $C_{(v_3,v_4)}=3$.
    %Definition~\ref{def:mono-cube} requires that $f$ increase by at least $D=2$ along the dotted gray lines, which it does.
    Adding $(v_1,v_2)$ and $(v_3,v_4)$ increase $f$ by $2$ and $3$, respectively.
  }\label{fig:mono-cube}
\end{figure}

\begin{figure}[t]
  \centering
  \includegraphics[width=0.9\linewidth]{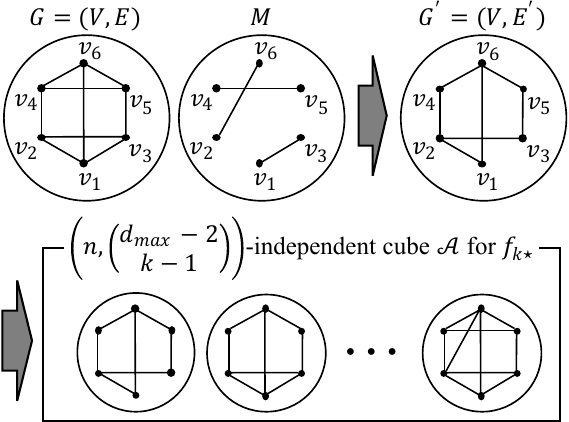}
  \vspace{-2mm}
  \caption{
    Construction of an independent cube for a $k$-star function ($n=6$, $d_{max}=4$). 
    From a 
    %$(d_{max} - 1)$-
    $3$-regular graph $G=(V,E)$ and $M=\{(v_1,v_3),(v_2,v_6),(v_4,v_5)\}$, we make a graph $G'=(V,E')$ such that $E' = E \setminus M$. 
    Then $\calA = \{(V, E' \cup N): N \subseteq M\}$ forms an
    $(n, 2\binom{d_{max}-2}{k-1})$-independent cube for $f_{k\star}$.
  }\label{fig:mono-cube_kstar}
\end{figure}

Such a set of inputs has an ``independence'' property because,
regardless of which edges from $M$ has been added before, adding edge $e \in M$
always changes $f$ by $C_e$. 
Figure~\ref{fig:mono-cube} shows an example of a $(4,2)$-independent cube for $f$. 

We can also construct 
a independent cube for 
a $k$-star function 
as follows. 
Assume that $n$ is even. 
It is well known in graph theory that if $n$ is even, then 
for any $d\in[n-1]$, there exists a 
$d$-regular graph where every node has degree $d$ \cite{Ganesan_arXiv18}. 
Therefore, there exists a 
$(d_{max}-1)$-regular graph $G=(V,E)$ 
of size $n$. 
Pick an arbitrary perfect matching $M$ on the nodes. Now, let 
$G' = (V,E')$ such that $E' = E \setminus M$. 
Every node in $G'$ has degree between $d_{max}-2$ and $d_{max}-1$. 
Adding an edge in $M$ to $G'$ will produce at least
$2\binom{d_{max}-2}{k-1}$ new $k$-stars.
Thus, $\calA = \{(V, E' \cup N): 
N \subseteq M\}$ forms an
$(n, 2\binom{d_{max}-2}{k-1})$-independent cube for $f_{k\star}$. 
Note that the maximum degree of each graph in $\calA$ is at most $d_{max}$. 
Figure~\ref{fig:mono-cube_kstar} shows how to construct an independent cube for a $k$-star function when $n=6$ and $d_{max}=4$. 

Using the structure that the $(n,D)$-independent cube imposes on $f$,
we can prove a lower bound:
\begin{theorem}\label{thm:lower-bound}
  Let 
  %$\hf(\bmA)$ 
  $\hf(G)$ 
  have the form of~\eqref{eq:one-round-lower}, 
  where $\calR_1, \ldots, \calR_n$ are independently run.
  %Assume that 
  %$(\calR_1, \ldots, \calR_n)$ 
  %satisfy 
  %provide 
  %$\epsilon$-relationship DP. 
  %and are independently run.
  Let $\cal{A}$ be an $(n,D)$-independent cube for $f$. 
  If 
  $(\calR_1, \ldots, \calR_n)$ 
  provides 
  $\epsilon$-relationship DP, 
  then 
  %, with 
  %$\bmA$ 
  %$G$ 
  %uniformly drawn from $\calA$, 
  we have
  \[
    % \E_{\bmA, \calR_1, \ldots, \calR_n}[l_2^2(f(\bmA), \hf(\bmA))] =
    % \Omega(\frac{e^{\epsilon}}{(e^{\epsilon}+1)^2}nD^2).
    \frac{1}{\calA} \sum_{G \in \calA} \E[l_2^2(f(G), \hf(G))] =
    \Omega\left(\frac{e^{\epsilon}}{(e^{\epsilon}+1)^2}nD^2\right).
  \]
\end{theorem}

A corollary of Theorem~\ref{thm:lower-bound} is that if $\calR_1, \ldots,
\calR_n$ satisfy $\epsilon$-edge LDP, then they satisfy $2\epsilon$
-relationship DP and thus for edge LDP we have a lower bound of
$\Omega\left(\frac{e^{2\epsilon}}{(e^{2\epsilon}+1)^2}nD^2\right)$.

Theorem~\ref{thm:lower-bound}, combined with the fact that there exists an
$(n,2\binom{d_{max}-2}{k-1})$-independent cube for 
a $k$-star function 
implies Corollary~\ref{cor:kstars-lb}. 
In Appendix~\ref{sub:cube_triangle}, we also construct an $(n, \frac{d_{max}}{2}-2)$
independent cube 
for $f_\triangle$ and establish a lower bound of 
$\Omega(\frac{e^{2\epsilon}}{(e^{2\epsilon}+1)^2} nd_{max}^2)$ for
$f_\triangle$. 

The upper and lower bounds on the $l_2$ losses 
shown in
this section appear in Table~\ref{tab:perf}.

\section{Experiments}
\label{sec:experiments}

Based on our theoretical results in Section~\ref{sec:algorithms}, we would like to pose the following questions:
\begin{itemize}
    \item For triangle counts, how much does the two-rounds interaction help over a single round in practice?
    \item What is the privacy-utility trade-off 
    %for subgraph counts in the local model 
    of our LDP algorithms 
    (i.e., how beneficial are our LDP algorithms)?
    %\item Is it possible to accurately estimate triangle counts and $k$-star counts in the local model?
\end{itemize}
We conducted experiments to answer to these questions. 

\subsection{Experimental Set-up}
\label{sub:setup}
We used the following two large-scale datasets:

\smallskip
\noindent{\textbf{IMDB.}}~~The Internet Movie Database (denoted by \IMDB{}) 
\cite{IMDB_GD05} 
includes a bipartite graph between $896308$ actors and $428440$ movies. 
We assumed actors as users. 
From the bipartite graph, we extracted 
a graph $G^*$ with $896308$ nodes (actors), where an edge between two actors represents that they have played in the same movie. 
There are $57064358$ edges in $G^*$, and the average degree in $G^*$ is $63.7$ $(=\frac{57064358}{896308})$.

\smallskip
\noindent{\textbf{Orkut.}}~~The 
Orkut online social network dataset (denoted by \Orkut{})  \cite{snapnets} includes a graph $G^*$ with $3072441$ users and $117185083$ edges. 
The average degree in $G^*$ is $38.1$ $(=\frac{117185083}{3072441})$. 
Therefore, \Orkut{} is more sparse than \IMDB{} (whose average degree in $G^*$ is $63.7$). 

\smallskip
For each dataset, we randomly selected $n$ users from the whole graph $G^*$, and extracted a graph $G=(V,E)$ with 
$n$ users. 
Then we estimated the number of triangles $f_\triangle(G)$, the number of $k$-stars $f_{k\star}(G)$, and the clustering coefficient ($=\frac{3 f_\triangle(G)}{f_{2\star}(G)}$) using $\epsilon$-edge LDP (or $\epsilon$-edge centralized DP) algorithms in Section~\ref{sec:algorithms}. 
Specifically, we used the following algorithms:

\smallskip
\noindent{\textbf{Algorithms for triangles.}}~~For algorithms for estimating $f_\triangle(G)$, we used the following three algorithms: 
(1) the RR (Randomized Response) with the empirical estimation method in the local model (i.e., \alg{LocalRR$_\triangle$} in Section~\ref{sub:non-interactive_triangles}), 
(2) the two-rounds algorithm in the local model (i.e., \alg{Local2Rounds$_\triangle$} in Section~\ref{sub:two_rounds}), and 
(3) the Laplacian mechanism in the centralized model (i.e., \alg{CentralLap$\triangle$} in Section~\ref{sub:non-interactive_triangles}).

\smallskip
\noindent{\textbf{Algorithms for $k$-stars.}}~~For algorithms for estimating $f_{k\star}(G)$, we used the following two algorithms: 
(1) the Laplacian mechanism in the local model (i.e., \alg{LocalLap$_k\star$} in Section~\ref{sub:non-interactive_k_stars}) and 
(2) the Laplacian mechanism in the centralized model (i.e., \alg{CentralLap$_k\star$} in Section~\ref{sub:non-interactive_k_stars}). 

\smallskip
For each algorithm, we evaluated the $l_2$ loss and the relative error (as described in Section~\ref{sub:graph_statistics}), while changing the values of $n$ and $\epsilon$. 
To stabilize the performance, we attempted $\gamma \in \nats$ ways to randomly select $n$ users from $G^*$, and averaged the utility value over all the $\gamma$ ways to randomly select $n$ users. 
When we changed $n$ from $1000$ to $10000$, we set $\gamma = 100$ because the variance was large. For other cases, we set $\gamma = 10$. 

In Appendix~\ref{sec:BAGraph}, 
we also report experimental results using artificial graphs based on the Barab\'{a}si-Albert model \cite{NetworkScience}.

\begin{figure}[t]
\centering
\includegraphics[width=0.99\linewidth]{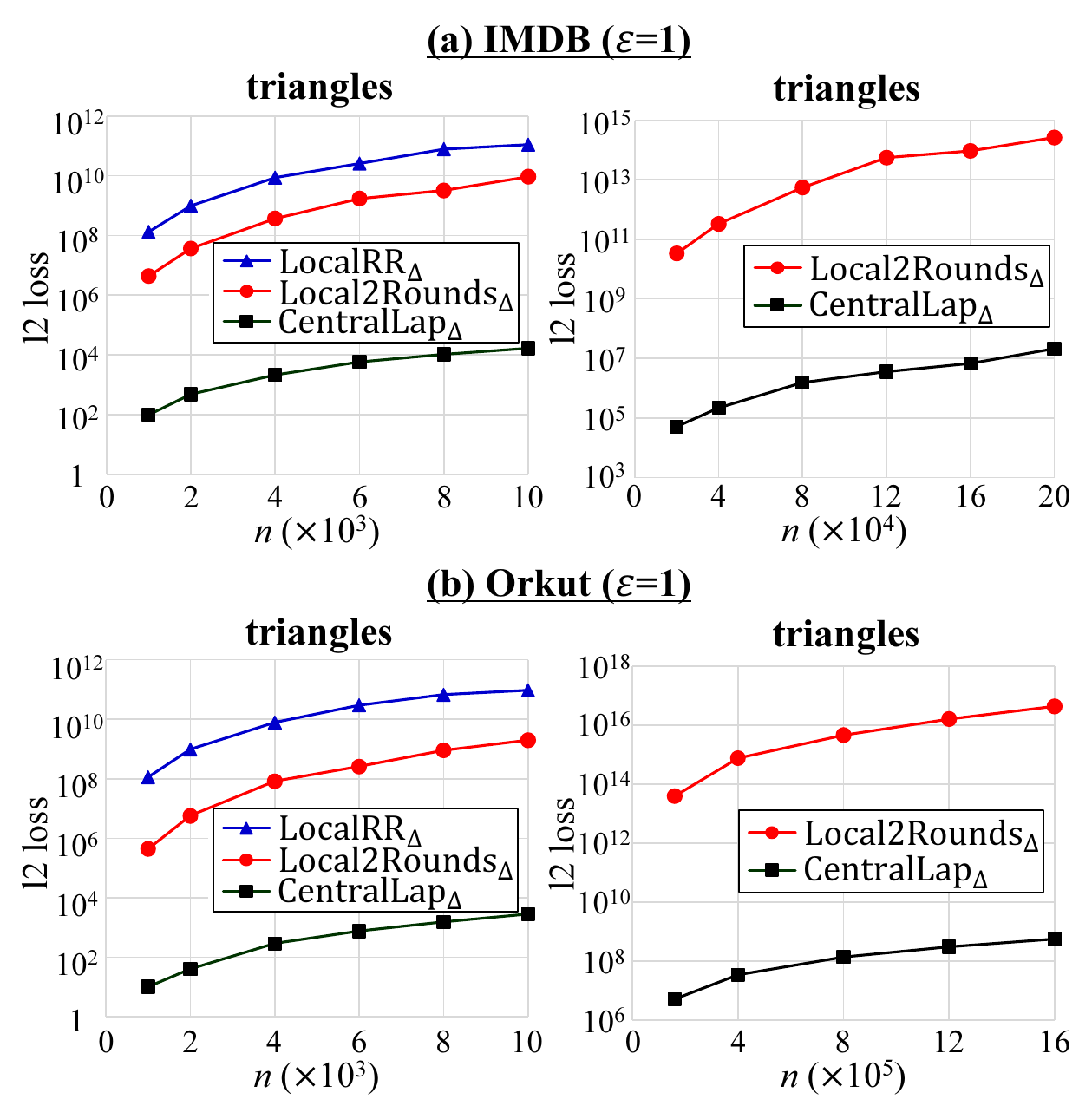}
\vspace{-4mm}
\caption{Relation between the number of users $n$ and the $l_2$ loss in triangle counts when $\epsilon = 1$ ($\epsilon_1 = \epsilon_2 = \frac{1}{2}$, $\td_{max} = d_{max}$). 
Here we do not evaluate \alg{LocalRR$_\triangle$} when $n > 10000$, because it is inefficient (see Section~\ref{sub:two_rounds} ``Time complexity'').}
\label{fig:res1_n_l2loss_tri}
\end{figure}

\begin{figure}[t]
\centering
\includegraphics[width=0.99\linewidth]{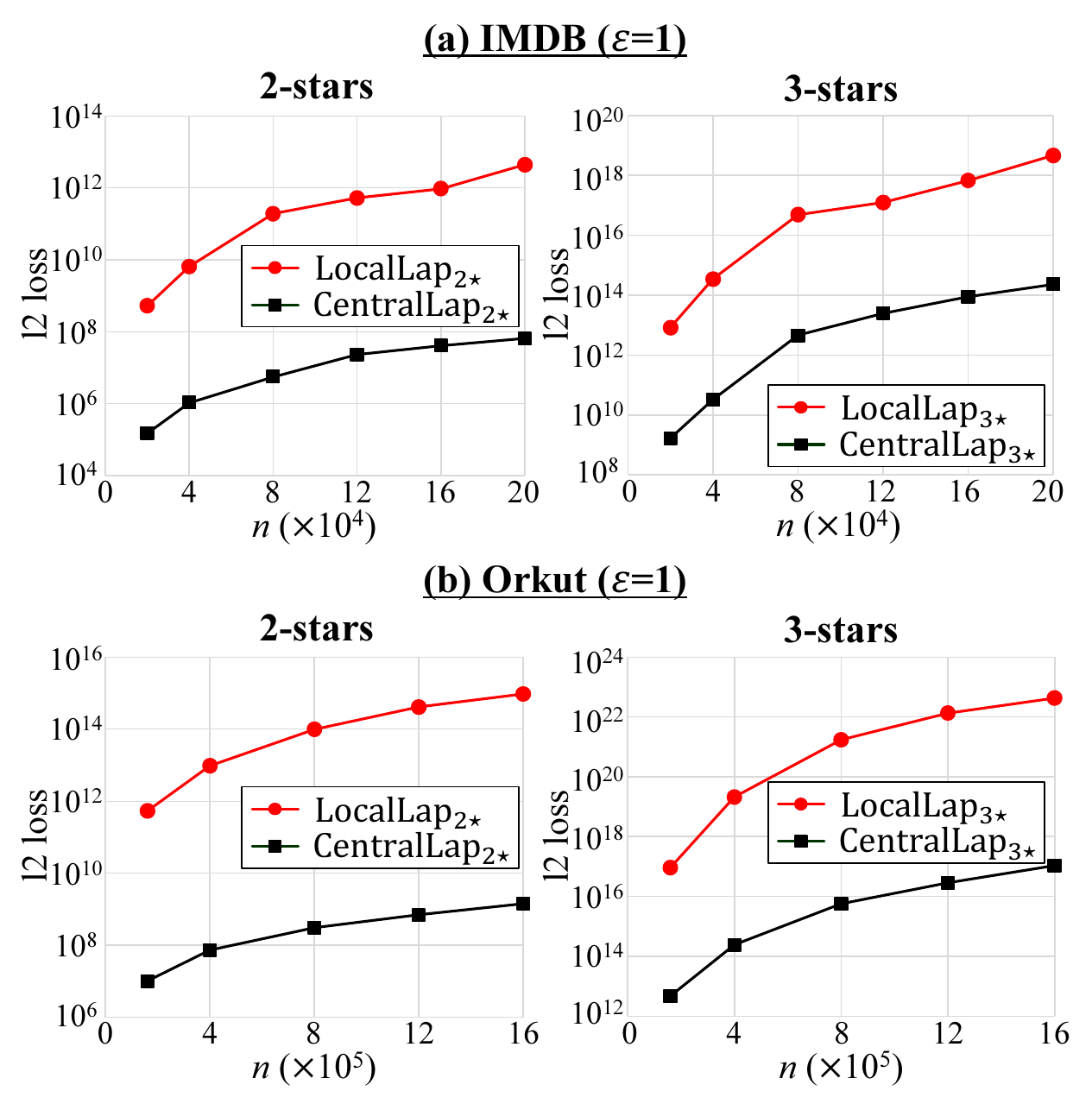}
\vspace{-5mm}
\caption{Relation between the number of users $n$ and the $l_2$ loss in $k$-star counts when $\epsilon=1$ ($\epsilon_1 = \epsilon_2 = \frac{1}{2}$, $\td_{max} = d_{max}$).}
\label{fig:res1_n_l2loss_kst}
\end{figure}

\subsection{Experimental Results}
\label{sub:results}
\noindent{\textbf{Relation between $n$ and the $l_2$ loss.}}~~We first evaluated the $l_2$ loss of the estimates of 
$f_\triangle(G)$, 
$f_{2\star}(G)$, 
and $f_{3\star}(G)$ 
while changing the number of users $n$. 
Figures~\ref{fig:res1_n_l2loss_tri} and \ref{fig:res1_n_l2loss_kst} shows the results ($\epsilon=1$). 
Here 
we did not evaluate \alg{LocalRR$_\triangle$} when $n$ was larger than $10000$, because \alg{LocalRR$_\triangle$} was inefficient 
(as described in Section~\ref{sub:two_rounds} ``Time complexity''). 
In \alg{Local2Rounds$_\triangle$}, we set $\epsilon_1 = \epsilon_2 = \frac{1}{2}$. 
As for $\td_{max}$, 
we set $\td_{max} = d_{max}$ (i.e., we assumed that $d_{max}$ is publicly available and did not perform graph projection) 
because we want to examine how well our theoretical results hold in our experiments. 
We also evaluate the effectiveness of the private calculation of $d_{max}$ at the end of Section~\ref{sub:results}. 

Figure~\ref{fig:res1_n_l2loss_tri} shows that \alg{Local2Rounds$_\triangle$} significantly outperforms \alg{LocalRR$_\triangle$}. 
Specifically, the $l_2$ loss of \alg{Local2Rounds$_\triangle$} is smaller than that of \alg{LocalRR$_\triangle$} by a factor of about $10^2$. 
The difference between \alg{Local2Rounds$_\triangle$} and \alg{LocalRR$_\triangle$} is larger in \Orkut{}. 
This is because \Orkut{} is more sparse, as described in Section~\ref{sub:setup}. 
For example, when $n=10000$, the maximum degree $d_{max}$ in $G$ was $73.5$ and $27.8$ on average in \IMDB{} and \Orkut{}, respectively. 
Recall that for a fixed $\epsilon$, 
the expected $l_2$ loss of \alg{Local2Rounds$_\triangle$} and \alg{LocalRR$_\triangle$} 
can be expressed as $O(nd_{max}^3)$ and $O(n^4)$, respectively. 
Thus \alg{Local2Rounds$_\triangle$} significantly outperforms \alg{LocalRR$_\triangle$}, especially in sparse graphs.

Figures~\ref{fig:res1_n_l2loss_tri} and \ref{fig:res1_n_l2loss_kst} show that the $l_2$ loss is roughly consistent with 
our upper-bounds in terms of $n$. 
Specifically, 
\alg{LocalRR$_\triangle$}, 
\alg{Local2Rounds$_\triangle$},  \alg{CentralLap$_\triangle$}, 
\alg{LocalLap$_{k\star}$}, and 
\alg{CentralLap$_{k\star}$} achieve 
the expected $l_2$ loss of $O(n^4)$, $O(nd_{max}^3)$, $O(d_{max}^2)$, $O(nd_{max}^{2k-2})$, and $O(d_{max}^{2k-2})$, respectively. 
Here note that 
each user's degree increases roughly in proportion to $n$ (though the degree is much smaller than $n$), 
as we randomly select $n$ users from the whole graph $G^*$. Assuming that $d_{max} = O(n)$,  Figures~\ref{fig:res1_n_l2loss_tri} and \ref{fig:res1_n_l2loss_kst} are roughly consistent with the upper-bounds. 
The figures also show the limitations of the local model in terms of the utility when compared to the centralized model.

\smallskip
\noindent{\textbf{Relation between $\epsilon$ and the $l_2$ loss.}}~~Next we evaluated the $l_2$ loss 
when we changed the privacy budget $\epsilon$ in edge LDP. 
Figure~\ref{fig:res2_eps_l2loss} shows the results for triangles and $2$-stars ($n=10000$). 
Here we omit the result of $3$-stars because it is similar to that of $2$-stars. 
In \alg{Local2Rounds$_\triangle$}, we set $\epsilon_1 = \epsilon_2 = \frac{\epsilon}{2}$. 

Figure~\ref{fig:res2_eps_l2loss} shows that the $l_2$ loss is roughly consistent with 
our upper-bounds in terms of $\epsilon$. 
For example, when we decrease $\epsilon$ from $0.4$ to $0.1$, the $l_2$ loss 
increases by a factor of about $5000$, $200$, and $16$ for both the datasets in \alg{LocalRR$_\triangle$}, \alg{Local2Rounds$_\triangle$}, and \alg{CentralLap$_\triangle$}, respectively. 
They are 
roughly consistent with our theoretical results 
that for small $\epsilon$, the expected $l_2$ loss of \alg{LocalRR$_\triangle$}, \alg{Local2Rounds$_\triangle$}, and \alg{CentralLap$_\triangle$} is  $O(\epsilon^{-6})$\footnote{We used $e^\epsilon \approx \epsilon + 1$ to derive the upper-bound of \alg{LocalRR$_\triangle$} for small $\epsilon$.}, $O(\epsilon^{-4})$, and $O(\epsilon^{-2})$, respectively. 

\begin{figure}[t]
\centering
\includegraphics[width=0.99\linewidth]{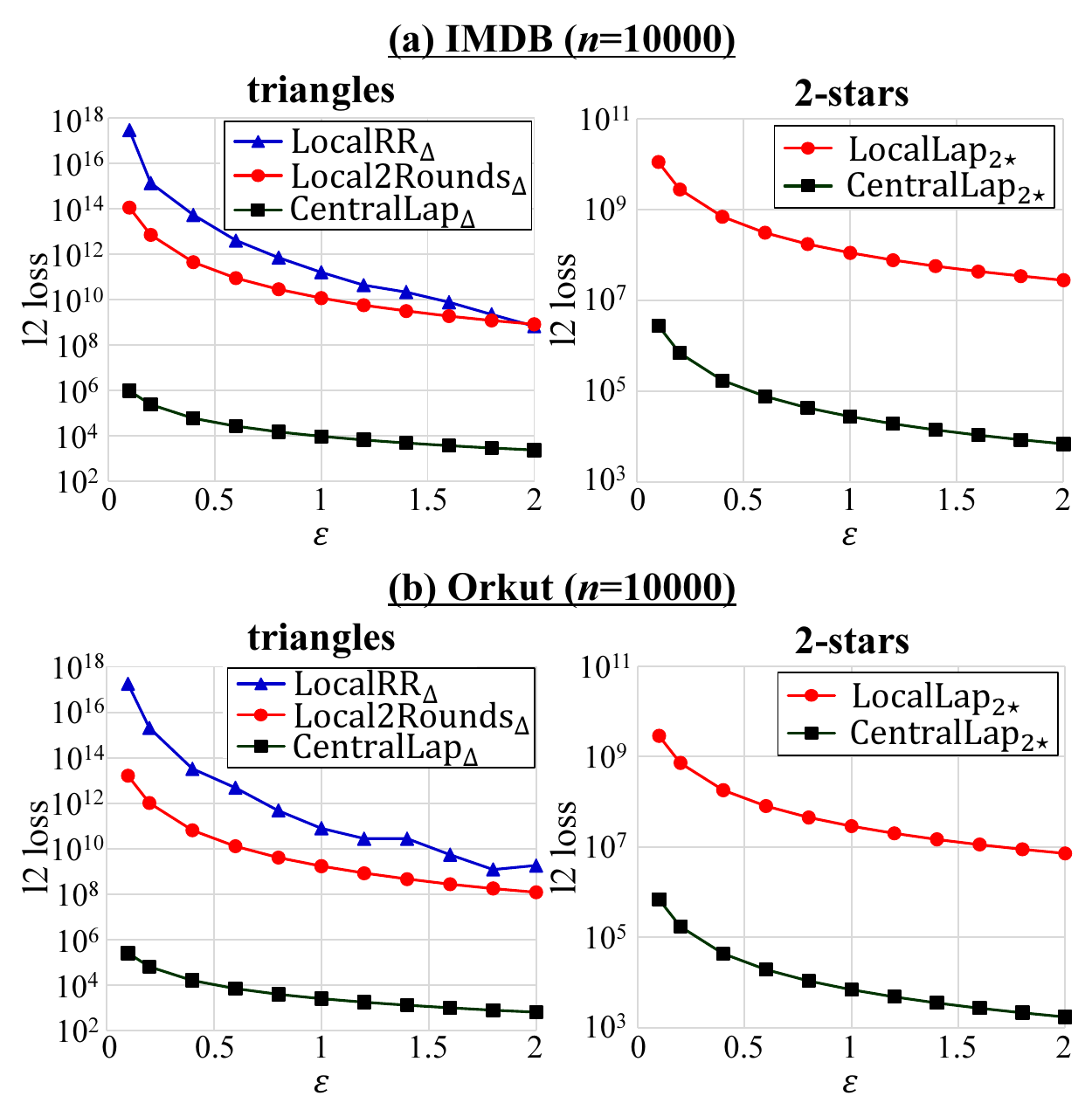}
\vspace{-5mm}
\caption{Relation between $\epsilon$ in edge LDP and the $l_2$ loss when $n=10000$ ($\epsilon_1 = \epsilon_2 = \frac{\epsilon}{2}$, $\td_{max} = d_{max}$).}
\label{fig:res2_eps_l2loss}
\end{figure}

\begin{figure}[t]
\centering
\includegraphics[width=0.99\linewidth]{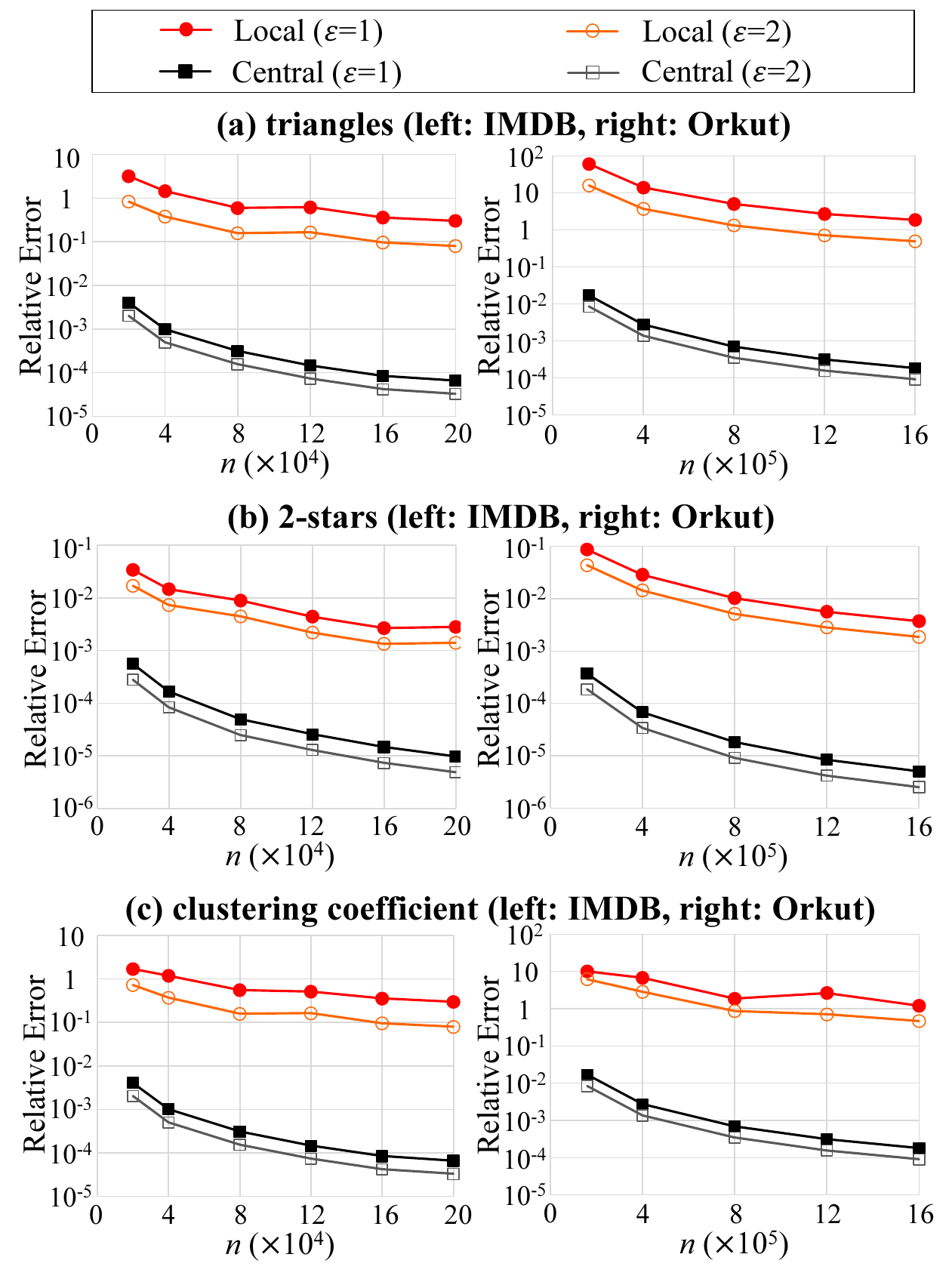}
\vspace{-5mm}
\caption{Relation between 
$n$ and the relative error. In the local model, we used \alg{Local2Rounds$_\triangle$} ($\epsilon = 1$ or $2$) and \alg{LocalLap$_k\star$} ($\epsilon = 1$ or $2$) for estimating triangle counts $f_\triangle(G)$ and $k$-star counts $f_{k\star}(G)$, respectively ($\td_{max} = d_{max}$).}
\label{fig:res3_n_relerr}
\end{figure}

Figure~\ref{fig:res2_eps_l2loss} also shows that 
\alg{Local2Rounds$_\triangle$} significantly outperforms \alg{LocalRR$_\triangle$} especially when $\epsilon$ is small, which is also consistent with our theoretical results. 
Conversely, the difference between \alg{LocalRR$_\triangle$} and \alg{Local2Rounds$_\triangle$} is small when $\epsilon$ is large. 
This is because when $\epsilon$ is large, the RR outputs the true value with high probability. 
For example, when 
$\epsilon \geq 5$, 
the RR outputs the true value with 
$\frac{e^\epsilon}{e^\epsilon+1} > 0.993$. 
However, \alg{LocalRR$_\triangle$} with 
such a large value of $\epsilon$ 
does not guarantee strong privacy, because it outputs the true value in most cases. 
\alg{Local2Rounds$_\triangle$} significantly outperforms \alg{LocalRR$_\triangle$} 
when we want to estimate $f_\triangle(G)$ or $f_{k\star}(G)$ 
with a strong privacy guarantee; e.g., $\epsilon \leq 1$ \cite{DP_Li}. 

\smallskip
\noindent{\textbf{Relative error.}}~~As the number of users $n$ increases, the numbers of triangles $f_\triangle(G)$ and $k$-stars $f_{k\star}(G)$ increase. 
This causes the increase of the $l_2$ loss. 
Therefore, we also evaluated the relative error, as described in Section~\ref{sub:graph_statistics}. 

Figure~\ref{fig:res3_n_relerr} shows the relation between $n$ and the relative error 
(we omit the result of $3$-stars because it is similar to that of $2$-stars). 
In the local model, we used \alg{Local2Rounds$_\triangle$} and \alg{LocalLap$_k\star$} for estimating 
$f_\triangle(G)$ 
and 
$f_{k\star}(G)$, 
respectively 
(we did not use \alg{Local2RR$_\triangle$}, because it is both inaccurate and inefficient). 
For both algorithms, we set $\epsilon = 1$ or $2$ 
($\epsilon_1 = \epsilon_2 = \frac{\epsilon}{2}$ in \alg{Local2Rounds$_\triangle$}) and $\td_{max} = d_{max}$. 
Then we estimated the clustering coefficient as: $\frac{3\hf_{\triangle}(G, \epsilon_1, \epsilon_2, d_{max})}{\hf_{k\star}(G, \epsilon, d_{max})}$, where 
$\hf_{\triangle}(G, \epsilon_1, \epsilon_2, d_{max})$ and 
$\hf_{k\star}(G, \epsilon, d_{max})$ are the estimates of $f_\triangle(G)$ and $f_{k\star}(G)$, respectively. 
If the estimate of the clustering coefficient is smaller than $0$ (resp.~larger than $1$), we set the estimate to $0$ (resp.~$1$) because the clustering coefficient is always between $0$ and $1$. 
In the centralized model, we used \alg{CentralLap$_\triangle$} and \alg{CentralLap$_k\star$} ($\epsilon=1$ or $2$, $\td_{max} = d_{max}$) and calculated the clustering coefficient in the same way. 

Figure~\ref{fig:res3_n_relerr} shows that for all cases, the relative error decreases with increase in $n$. 
This is because both $f_\triangle(G)$ and $f_{k\star}(G)$ significantly increase with increase in $n$. 
Specifically, let $f_{\triangle,v_i}(G) \in \nnints$ the number of triangles that involve user $v_i$, and $f_{k\star,v_i}(G) \in \nnints$ be the number of $k$-stars of which user $v_i$ is a center. 
Then $f_\triangle(G) = \frac{1}{3}\sum_{i=1}^n f_{\triangle,v_i}(G)$ and $f_{k\star,v_i}(G) = \sum_{i=1}^n f_{k\star,v_i}(G)$. 
Since both $f_{\triangle,v_i}(G)$ and $f_{k\star,v_i}(G)$ increase with increase in $n$, both $f_\triangle(G)$ and $f_{k\star}(G)$ increase \textit{at least} in proportion to $n$. 
Thus $f_\triangle(G)^2 \geq \Omega(n^2)$ and $f_{k\star}(G)^2 \geq \Omega(n^2)$. 
In contrast, \alg{Local2Rounds$_\triangle$}, \alg{LocalLap$_k\star$}, \alg{CentralLap$_\triangle$}, and \alg{CentralLap$_k\star$} achieve the expected $l_2$ loss of $O(n)$, $O(n)$, $O(1)$, and $O(1)$, respectively (when we ignore $d_{max}$ and $\epsilon$), all of which are smaller than $O(n^2)$. 
Therefore, the relative error decreases with increase in $n$. 

This result demonstrates that we can accurately estimate graph statistics for large $n$ in the local model. 
In particular, the relative error is smaller in \IMDB{} 
because \IMDB{} is denser and includes a larger number of triangles and $k$-stars; i.e., the denominator of the relative error is large. 
For example, when $n=200000$ and $\epsilon=1$, the relative error is 
$0.30$ and 
$0.0028$ 
for triangles and $2$-stars, 
respectively. 
Note that the clustering coefficient requires $2\epsilon$ 
because we need to estimate both $f_\triangle(G)$ and $f_{k\star}(G)$. 
Yet, we can still accurately calculate the clustering coefficient with a moderate privacy budget; 
e.g., the relative error of the clustering coefficient is 
$0.30$ 
when the privacy budget is $2$ (i.e., $\epsilon = 1$).
If $n$ is larger, then $\epsilon$ would be smaller at the same value of the relative error. 

\begin{figure}[t]
\centering
\includegraphics[width=0.99\linewidth]{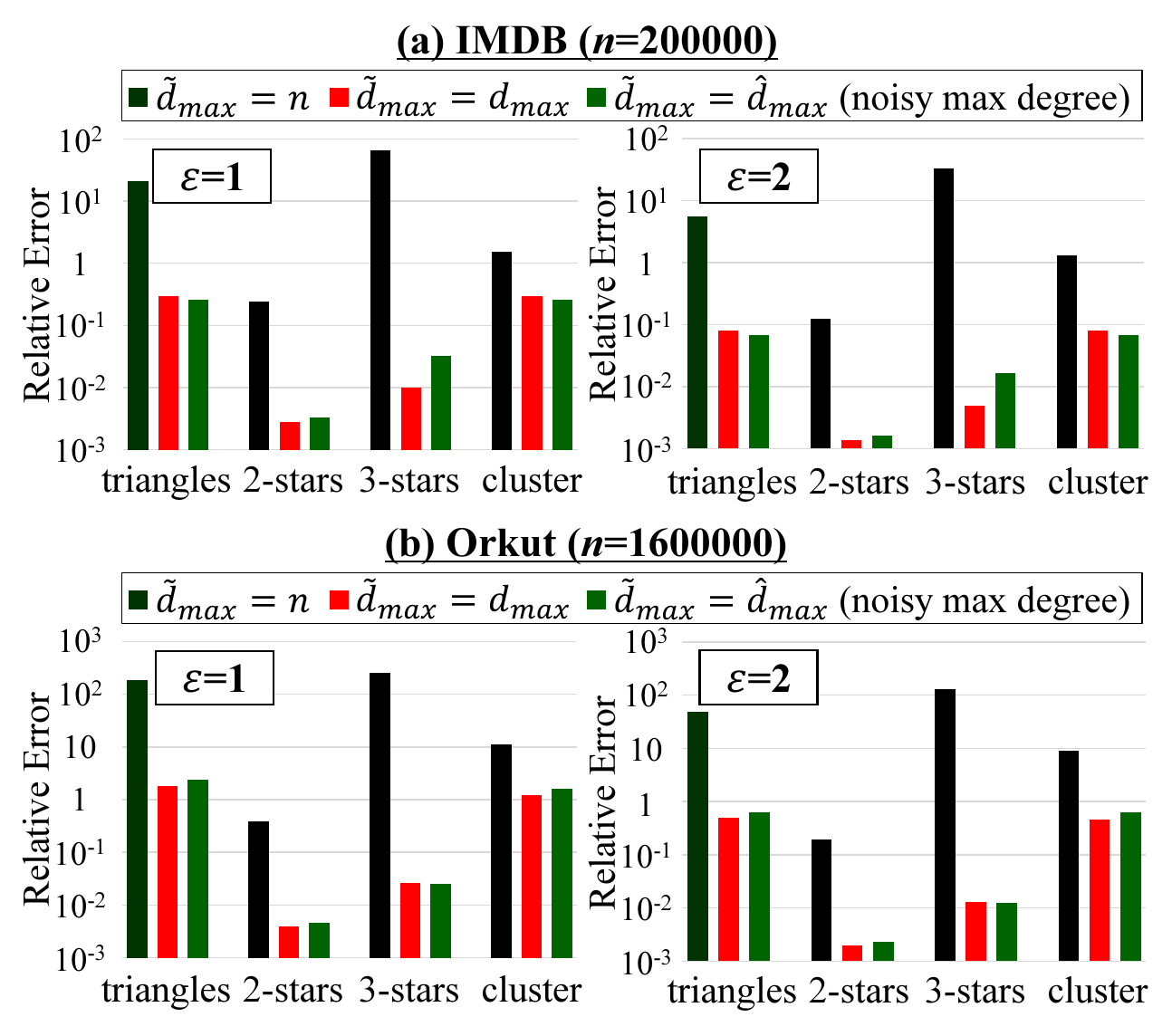}
\vspace{-4mm}
\caption{Relative error 
when $\td_{max} = n$ (\#users), 
$d_{max}$ (max degree), or 
$\hd_{max}$ (noisy max degree). 
We used \alg{Local2Rounds$_\triangle$} ($\epsilon = 1$ or $2$) and \alg{LocalLap$_k\star$} ($\epsilon = 1$ or $2$) for estimating triangle counts $f_\triangle(G)$ and $k$-star counts $f_{k\star}(G)$, respectively. 
}
\label{fig:res4_noisy_local}
\end{figure}

\smallskip
\noindent{\textbf{Private calculation of $d_{max}$.}}~~We have so far assumed that $\td_{max} = d_{max}$ (i.e., $d_{max}$ is publicly available) in our experiments. 
We 
finally evaluate the methods to privately calculate $d_{max}$ with $\epsilon_0$-edge LDP 
(described in Sections~\ref{sub:non-interactive_k_stars} and \ref{sub:two_rounds}). 

Specifically, we used \alg{Local2Rounds$_\triangle$} and \alg{LocalLap$_k\star$} for estimating $f_\triangle(G)$ and $f_{k\star}(G)$, respectively, and evaluated the following three methods for setting $\td_{max}$: 
(i) 
$\td_{max} = n$; 
(ii) 
$\td_{max} = d_{max}$; 
(iii) 
$\td_{max} = \hd_{max}$, where $\hd_{max}$ is the private estimate of $d_{max}$ 
(noisy max degree) in Sections~\ref{sub:non-interactive_k_stars} and \ref{sub:two_rounds}. 

We set $n=200000$ in \IMDB{} and $n=1600000$ in \Orkut{}. 
Regarding the total privacy budget $\epsilon$ in edge LDP for estimating $f_\triangle(G)$ or $f_{k\star}(G)$, we set $\epsilon=1$ or $2$. 
We used $\frac{\epsilon}{10}$ for privately calculating $d_{max}$ (i.e., $\epsilon_0 = \frac{\epsilon}{10}$), and the remaining privacy budget $\frac{9\epsilon}{10}$ as input to \alg{Local2Rounds$_\triangle$} or \alg{LocalLap$_k\star$}. 
In \alg{Local2Rounds$_\triangle$}, we set $\epsilon_1 = \epsilon_2$; i.e., we set $(\epsilon_0, \epsilon_1, \epsilon_2) = (0.1, 0.45, 0.45)$ or $(0.2, 0.9, 0.9)$. 
Then we estimated the clustering coefficient in the same way as Figure~\ref{fig:res3_n_relerr}. 

Figure~\ref{fig:res4_noisy_local} shows the results. 
Figure~\ref{fig:res4_noisy_local} shows that 
the algorithms with $\td_{max} = \hd_{max}$ (noisy max degree) 
achieves the relative error close to (sometimes almost the same as) 
the algorithms with $\td_{max} = d_{max}$ 
and significantly outperforms 
the algorithms with $\td_{max} = n$. 
This means that we can privately estimate $d_{max}$ without a significant loss of utility. 

\smallskip
\noindent{\textbf{Summary of results.}}~~In summary, 
our experimental results showed that the estimation error of triangle counts is significantly reduced by introducing the interaction between users and a data collector. 
The results also showed that 
we can achieve small relative errors 
(much smaller than 1) for subgraph counts 
with privacy budget $\epsilon=1$ or $2$ in edge LDP. 

As described in Section~\ref{sec:intro}, non-private 
subgraph 
counts may reveal some friendship information, and a central server may face data breaches. 
Our LDP algorithms are highly beneficial because they enable us to analyze the connection patterns in a graph 
(i.e., subgraph counts) 
or to understand how likely two friends of an individual will also be a friend 
(i.e., clustering coefficient) 
while strongly protecting individual privacy.

\section{Conclusions}
\label{sec:conclusions}
We presented a series of algorithms for counting triangles and $k$-stars under LDP. 
We 
showed that an additional round can significantly reduce the estimation error in triangles, and the algorithm based on the Laplacian mechanism provides an order optimal error in the non-interactive local model. 
We also showed lower-bounds for general functions including triangles and $k$-stars. 
We conducted experiments using two real datasets, and showed that our algorithms achieve small relative errors, especially when the number of users is large.

As future work, we would like to develop algorithms for other subgraph counts such as cliques and $k$-triangles \cite{Karwa_PVLDB11}. 

\section*{Acknowledgments}
Kamalika Chaudhuri and Jacob Imola would like to thank ONR under N00014-20-1-2334 and UC Lab Fees under LFR 18-548554  for research support. 
Takao Murakami was supported in part by JSPS KAKENHI JP19H04113.

\bibliographystyle{plain}
\bibliography{main}

\appendix

\section{Effectiveness of empirical estimation in \alg{LocalRR$_\triangle$}}
\label{sec:RR_emp}
In Section~\ref{sub:non-interactive_triangles}, we presented \alg{LocalRR$_\triangle$}, which uses the empirical estimation method after the RR. 
Here we show the effectiveness of empirical estimation by comparing \alg{LocalRR$_\triangle$} with the RR without empirical estimation \cite{Qin_CCS17,Ye_ICDE20}. 

As the RR without empirical estimation, we applied the RR to the lower triangular part of the adjacency matrix $\bmA$; i.e., we ran lines 1 to 6 in Algorithm~\ref{alg:subgraph-rr}. 
Then we output the number of noisy triangles $m_3$. 
We denote this algorithm by \alg{RR w/o emp}.

Figure~\ref{fig:res5_RR_wo_emp} shows the $l_2$ loss of \alg{LocalRR$_\triangle$} and \alg{RR w/o emp} when we changed $n$ from $1000$ to $10000$ or $\epsilon$ in edge LDP from $0.1$ to $2$. 
The experimental set-up is the same as Section~\ref{sub:setup}. 
Figure~\ref{fig:res5_RR_wo_emp} shows that \alg{LocalRR$_\triangle$} significantly outperforms \alg{RR w/o emp}, which means that the $l_2$ loss is significantly reduced by empirical estimation. 
As shown in Section~\ref{sec:experiments}, the $l_2$ loss of \alg{LocalRR$_\triangle$} is also significantly reduced by an additional round of interaction.

\begin{figure}[t]
\centering
\includegraphics[width=0.99\linewidth]{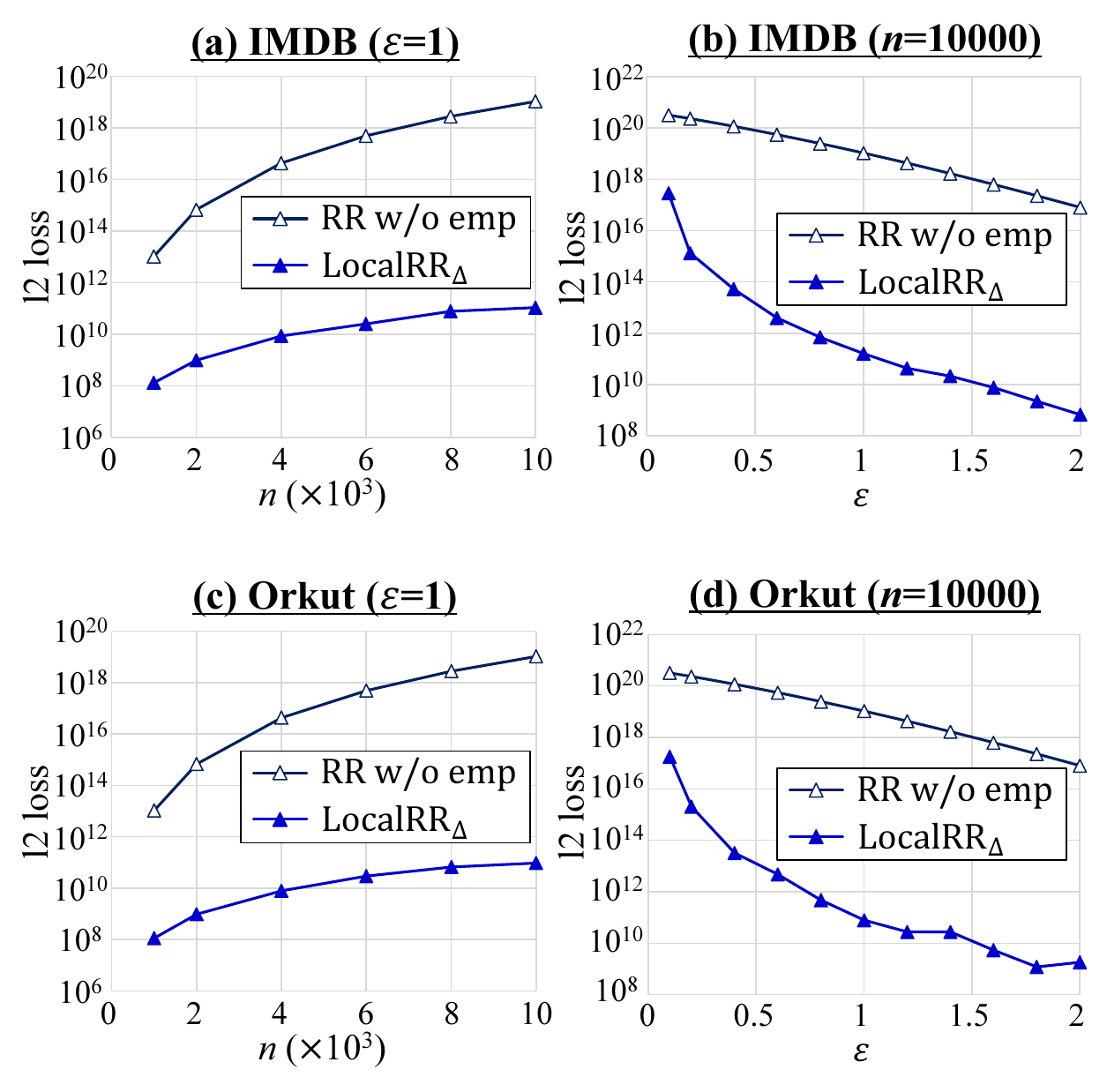}
\vspace{-5mm}
\caption{$l_2$ loss of \alg{LocalRR$_\triangle$} and the RR without empirical estimation (\alg{RR w/o emp}).}
\label{fig:res5_RR_wo_emp}
\end{figure}

\section{Experiments on Barab\'{a}si-Albert Graphs}
\label{sec:BAGraph}
\noindent{\textbf{Experimental set-up.}}~~In Section~\ref{sec:experiments}, we evaluated our algorithms using two real datasets: \IMDB{} and \Orkut{}. 
We also evaluated our algorithms using artificial graphs that have power-law degree distributions. 
We used the BA (Barab\'{a}si-Albert) model \cite{NetworkScience} to generate such graphs.

In the BA model, an artificial graph (referred to as a BA graph)
is grown by adding new nodes one at a time. 
Each new node is connected to $\lambda \in \nats$ existing nodes with probability proportional to the degree of the existing node. 
In our experiments, we used 
NetworkX \cite{Hagberg_SciPy08}, a Python package for graph analysis, to generate BA graphs.

We generated a BA graph $G^*$ with $1000000$ nodes using NetworkX. 
For the attachment parameter $\lambda$, we set $\lambda=10$ or $50$. 
When $\lambda=10$ (resp.~$50$), the average degree of $G^*$ was $10.0$ (resp.~$50.0$). 
For each case, we randomly generated $n$ users from the whole graph $G^*$, and extracted a graph $G=(V,E)$ with the $n$ users. 
Then we estimated the number of triangles $f_\triangle(G)$ and the number of $2$-stars $f_{2\star}(G)$. 
For triangles, we evaluated \alg{LocalRR$_\triangle$}, \alg{Local2Rounds$_\triangle$}, and \alg{CentralLap$\triangle$}. 
For $2$-stars, we evaluated \alg{LocalLap$_2\star$} and \alg{CentralLap$_2\star$}. 
In \alg{Local2Rounds$_\triangle$}, we set $\epsilon_1 = \epsilon_2$.
For $\td_{max}$, we set $\td_{max} = d_{max}$. 

We evaluated the $l_2$ loss while changing $n$ and $\epsilon$. 
We attempted $\gamma \in \nats$ ways to randomly select $n$ users from $G^*$, and averaged the $l_2$ loss over all the $\gamma$ ways to randomly select $n$ users. 
As with Section~\ref{sec:experiments}, we set $\gamma=100$ and changed $n$ from $1000$ to $10000$ while fixing $\epsilon=1$. 
Then we set $\gamma=10$ and changed $\epsilon$ from $0.1$ to $2$ while fixing $n=10000$.

\smallskip
\noindent{\textbf{Experimental results.}}~~Figure~\ref{fig:res6_BAGraph} shows the results. 
Overall, Figure~\ref{fig:res6_BAGraph} has a similar tendency to Figures~\ref{fig:res1_n_l2loss_tri}, \ref{fig:res1_n_l2loss_kst}, and \ref{fig:res2_eps_l2loss}. 
For example, \alg{Local2Rounds$_\triangle$} significantly outperforms \alg{LocalRR$_\triangle$}, especially when the graph $G$ is sparse; i.e., $\lambda = 10$. 
In \alg{Local2Rounds$_\triangle$}, \alg{CentralLap$\triangle$}, \alg{LocalLap$_2\star$}, and \alg{CentralLap$_2\star$}, the $l_2$ loss increases with increase in $\lambda$. 
This is because the maximum degree $d_{max}$ $(= \td_{max})$ increases with increase in $\lambda$.

\begin{figure}[t]
\centering
\includegraphics[width=0.99\linewidth]{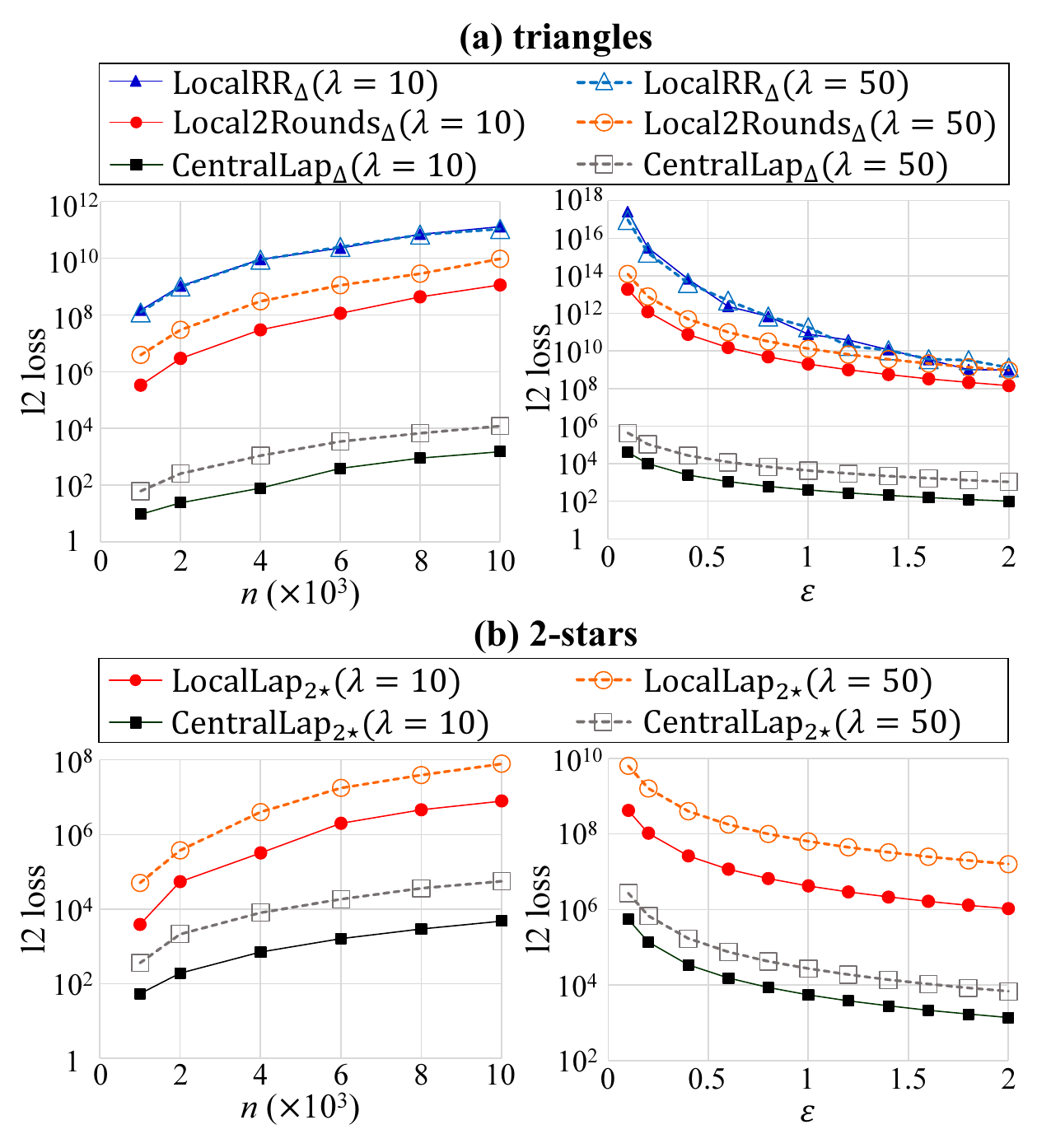}
\vspace{-2mm}
\caption{$l_2$ loss in the Barab\'{a}si-Albert graph datasets 
(left: $\epsilon=1$, right: $n=10000$). 
We set the attachment parameter $\lambda$ in the BA model to $\lambda=10$ or $50$, and $\td_{max}$ to $\td_{max} = d_{max}$.}
\label{fig:res6_BAGraph}
\end{figure}

Figure~\ref{fig:res6_BAGraph} also shows that the $l_2$ loss is roughly consistent with our upper-bounds in Section~\ref{sec:algorithms}. 
For example, recall that \alg{LocalRR$_\triangle$}, \alg{Local2Rounds$_\triangle$}, \alg{CentralLap$_\triangle$}, \alg{LocalLap$_{2\star}$}, and \alg{CentralLap$_{2\star}$} achieve the expected $l_2$ loss of $O(n^4)$, $O(nd_{max}^3)$, $O(d_{max}^2)$, $O(nd_{max}^{2})$, and $O(d_{max}^{2})$, respectively. 
Assuming that $d_{max} = O(n)$, the left panels of Figure~\ref{fig:res6_BAGraph} are roughly consistent with these upper-bounds. 
In addition, the right panels of Figure~\ref{fig:res6_BAGraph} show that when we set $\lambda=10$ and decrease $\epsilon$ from $0.4$ to $0.1$, the $l_2$ loss increases by a factor of about $3800$, $250$, and $16$ in \alg{LocalRR$_\triangle$}, \alg{Local2Rounds$_\triangle$}, and \alg{CentralLap$_\triangle$}, respectively. 
They are roughly consistent with our upper-bounds -- for small $\epsilon$, the expected $l_2$ loss of \alg{LocalRR$_\triangle$}, \alg{Local2Rounds$_\triangle$}, and \alg{CentralLap$_\triangle$} is  $O(\epsilon^{-6})$, $O(\epsilon^{-4})$, and $O(\epsilon^{-2})$, respectively.

In summary, for both the two real datasets and the BA graphs, our experimental results showed the following findings: 
(1) \alg{Local2Rounds$_\triangle$} significantly outperforms \alg{LocalRR$_\triangle$}, especially when the graph $G$ is sparse; 
(2) our experimental results are roughly consistent with our upper-bounds.

\section{Construction of an $(n, \frac{d_{max}}{2}-2)$ independent cube for $f_\triangle$}
\label{sub:cube_triangle}
Suppose that $n$ is even and $d_{max}$ is divisible by $4$.
Since $d_{max} < n$, it is possible to write 
$n = \eta_1 \frac{d_{max}}{2} + \eta_2$ 
for integers 
$\eta_1, \eta_2$
such that 
$\eta_1 \geq 1$ and $1 \leq \eta_2 < \frac{d_{max}}{2}$. 
Because 
$\eta_1 \frac{d_{max}}{2}$ 
and $n$ are even, 
we must have 
$\eta_2$ 
is even.
Now, we can write 
$n = (\eta_1-1) \frac{d_{max}}{2} + (\eta_2 + \frac{d_{max}}{2})$.
Thus, 
we can define 
a graph $G=(V,E)$ on $n$ 
nodes 
consisting of 
$(\eta_1-1)$ 
cliques of 
even 
size
$\frac{d_{max}}{2}$ and one final clique of an even size $\eta_2+\frac{d_{max}}{2} \in (\frac{d_{max}}{2}, d_{max})$ 
with all cliques disjoint. 

Since $G=(V,E)$ consists of even-sized cliques, it contains a perfect matching $M$. 
Figure~\ref{fig:mono-cube_triangle} shows examples of $G$ and $M$, where $n=14$, $d_{max} = 8$, $\eta_1 = 3$, and $\eta_2 = 2$. 
Let 
$G'=(V,E')$ such that $E' = E \setminus M$. 
Let $\calA = \{(V,E' \cup N: N \subseteq M\}$. 
Each edge in $G$ is part of at least $\frac{d_{max}}{2}-2$ triangles. 
For each pair of edges in $M$, the triangles of $G$ of which they are part
are disjoint. 
Thus, 
for any edge $e \in M$, 
removing $e$ 
from 
a graph in $\calA$ 
will remove at least $\frac{d_{max}}{2}-2$ triangles. This implies 
that $\calA$ 
is an $(n,\frac{d_{max}}{2}-2)$ independent cube for $f_\triangle$.

\begin{figure}[t]
  \centering
  \includegraphics[width=0.88\linewidth]{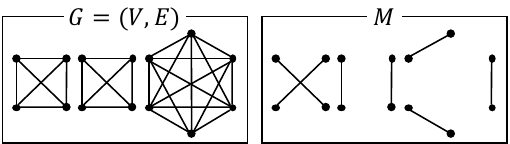}
  \vspace{-2mm}
  \caption{
    Examples of $G$ and $M$ for constructing an independent cube for $f_\triangle$ ($n=14$, $d_{max} = 8$, $\eta_1 = 3$, $\eta_2 = 2$).
  }\label{fig:mono-cube_triangle}
\end{figure}

\arxiv{
\section{Proof of Statements in Section~\ref{sec:algorithms}}
\label{sec:proof}
Here we prove the statements in Section~\ref{sec:algorithms}. 
Our proofs will repeatedly use the well-known bias-variance decomposition \cite{mlpp}, which we briefly explain below. 
We denote the variance of the random variable $X$ by $\mathbb{V}[X]$. 
If we are producing a private, randomized estimate $\hat{f}(G)$ of the graph function $f(G)$, then the expected $l_2$ loss (over the randomness in the algorithm) can be written as: 
\begin{equation}\label{eq:bias-var}
  \E[l_2^2(\hat{f}(G),f(G))] = \left(\E[\hat{f}(G)] - f(G)\right)^2
  + \V[\hat{f}(G)].
\end{equation}
The first term is the bias, and the second term is the variance. 
If the estimate is unbiased (i.e., $\E[\hat{f}(G)] = f(G)$), then the expected $l_2$ loss is equal to the variance.

\subsection{Proof of Theorem~\ref{thm:k-stars_LDP}}
Let $\calR_i$ be \alg{LocalLap$_{k\star}$}. 
Let $d_i,d'_i \in \nnints$ be the number of ``1''s in two neighbor lists $\bma_i,\bma'_i \in \{0,1\}^n$ that differ in one bit. 
Let $r_i = \binom{d_i}{k}$ and $r'_i = \binom{d'_i}{k}$. 
Below we consider two cases about $d_i$: when $d_i < \td_{max}$ and when $d_i \geq \td_{max}$.

\smallskip
\noindent{\textbf{Case 1: $d_i < \td_{max}$.}}~~In this case, both $\bma_i$ and $\bma'_i$ do not change after graph projection, as $d'_i \leq d_i + 1 \leq \td_{max}$. 
Then we obtain:
\begin{align*}
\Pr[\calR_i(\bma_i) = \hr_i] &= \exp\left(-\frac{\epsilon |\hr_i- r_i|}{\Delta}\right) \\
\Pr[\calR_i(\bma'_i) = \hr_i] &= \exp\left(-\frac{\epsilon |\hr_i- r'_i|}{\Delta}\right),
\end{align*}
where $\Delta = \binom{\td_{max}}{k-1}$. 
Therefore, 
\begin{align}
\frac{\Pr[\calR_i(\bma_i) = \hr_i]}{\Pr[\calR_i(\bma'_i) = \hr_i]} 
&= \exp\left( \frac{\epsilon |\hr_i- r'_i|}{\Delta} - \frac{\epsilon |\hr_i- r_i|}{\Delta}\right) \nonumber\\
&\leq  \exp\left( \frac{\epsilon |r'_i- r_i|}{\Delta} \right) \label{eq:Pr_R_i_a'_i_a_i}\\
& \hspace{4mm} (\text{by the triangle inequality}). \nonumber
\end{align}
If $d'_i = d_i + 1$, then $|r'_i- r_i|$ in (\ref{eq:Pr_R_i_a'_i_a_i}) can be written as follows:
\begin{align*}
|r'_i- r_i| 
= \binom{d_i+1}{k} - \binom{d_i}{k} 
= \binom{d_i}{k-1}
< \binom{\td_{max}}{k-1}
= \Delta, 
\end{align*}
Since we add $\Lap(\frac{\Delta}{\epsilon})$ to $r_i$, we obtain:
\begin{align}
\Pr[\calR_i(\bma_i) = \hr_i] \leq e^\epsilon \Pr[\calR_i(\bma'_i) = \hr_i]. 
\label{eq:R_i_a_i_hr_i}
\end{align}
If $d'_i = d_i - 1$, then $|r'_i- r_i| = \binom{d_i}{k} - \binom{d_i-1}{k} = \binom{d_i-1}{k-1} < \Delta$ and (\ref{eq:R_i_a_i_hr_i}) holds. 
Therefore, \alg{LocalLap$_{k\star}$} provides $\epsilon$-edge LDP. 

\smallskip
\noindent{\textbf{Case 2: $d_i \geq \td_{max}$.}}~~Assume 
that $d'_i = d_i + 1$. 
In this case, $d'_i > \td_{max}$. 
Therefore, $d'_i$ becomes $\td_{max}$ after graph projection. 
In addition, 
$d_i$ also becomes $\td_{max}$ after graph projection. 
Therefore, we obtain 
$d_i = d'_i = \td_{max}$ after graph projection. 
Thus 
$\Pr[\calR_i(\bma_i) = \hr_i] = \Pr[\calR_i(\bma'_i) = \hr_i]$. 

Assume that $d'_i = d_i - 1$. 
If $d_i > \td_{max}$, then $d_i = d'_i = \td_{max}$ after graph projection. 
Thus $\Pr[\calR_i(\bma_i) = \hr_i] = \Pr[\calR_i(\bma'_i) = \hr_i]$. 
If $d_i = \td_{max}$, then (\ref{eq:R_i_a_i_hr_i}) holds. 
Therefore, \alg{LocalLap$_{k\star}$} provides $\epsilon$-edge LDP. \qed

\subsection{Proof of Theorem~\ref{thm:k-stars}}
Assuming the maximum degree $d_{max}$ of $G$ is at most $\td_{max}$, the only
randomness in the algorithm will be the Laplace noise since graph projection
will not occur.
Since the Laplacian noise $\Lap(\frac{\Delta}{\epsilon})$ has mean $0$, the estimate $\hf_{k\star}(G, \epsilon, \td_{max})$ is unbiased. 
Then by the bias-variance decomposition \cite{mlpp}, 
the expected $l_2$ loss 
$\mathbb{E}[l_2^2(\hf_{k\star}(G, \epsilon, \td_{max}),\allowbreak f_{k\star}(G))]$ is equal to the variance of $\hf_{k\star}(G, \epsilon, \td_{max})$. 
The variance of $\hf_{k\star}(G, \epsilon, \td_{max})$ can be written as follows:
\begin{align*}
    \mathbb{V}[\hf_{k\star}(G, \epsilon, \td_{max})] 
    &= \mathbb{V}\left[ \sum_{i=1}^n \Lap\left( \frac{\Delta}{\epsilon} \right) \right] \\
    &= \frac{n \Delta^2}{\epsilon^2}.
\end{align*}
Since $\Delta = \binom{\td_{max}}{k-1} = O(\td_{max}^{k-1})$, we obtain:
\begin{align*}
    \mathbb{E}[l_2^2(\hf_{k\star}(G, \epsilon, \td_{max}), f_{k\star}(G))] 
    &= \mathbb{V}[\hf_{k\star}(G, \epsilon, \td_{max})] \\
    &= O\left(\frac{n \td_{max}^{2k-2}}{\epsilon^2}\right).
\end{align*}
\qed

\subsection{Proof of Proposition~\ref{prop:triangle_emp}}
Let $\mu = e^\epsilon$ and $\bmQ \in [0,1]^{4 \times 4}$ be a $4 \times 4$ matrix such that:
\begin{align}
  \bmQ = \frac{1}{(\mu+1)^3} \left(
    \begin{array}{cccc}
      \mu^3 & 3\mu^2 & 3\mu & 1 \\
      \mu^2 & \mu^3+2\mu & 2\mu^2+1 & \mu \\
      \mu & 2\mu^2+1 & \mu^3+2\mu & \mu^2 \\
      1 & 3\mu & 3\mu^2 & \mu^3
    \end{array}
  \right).
  \label{eq:Q_1}
\end{align}
Let $c_3, c_2, c_1, c_0 \in \nnints$ be respectively the number of triangles, 2-edges, 1-edge, and no-edges in $G$. 
Then we obtain:
\begin{align}
(\mathbb{E}[m_3], \mathbb{E}[m_2], \mathbb{E}[m_1],
\mathbb{E}[m_0]) = (c_3, c_2, c_1, c_0) \bmQ.
\label{eq:bmQ}
\end{align}
In other words, $\bmQ$ is a transition matrix from a type of subgraph (i.e., triangle, 2-edges, 1-edge, or no-edge) in $G$ to a type of subgraph in $G'$. 

Let $\hat{c}_3, \hat{c}_2, \hat{c}_1, \hat{c}_0 \in \reals$ be the empirical estimate of $(c_3, c_2, c_1, c_0)$. 
By (\ref{eq:bmQ}), they can be written as follows:
\begin{align}
(\hat{c}_3, \hat{c}_2, \hat{c}_1, \hat{c}_0) = (m_3, m_2, m_1, m_0) \bmQ^{-1}.
\label{eq:hc3_hc2_hc1_hc0}
\end{align}
Let $\bmQ_{i,j}^{-1}$ be the ($i,j$)-th element of $\bmQ^{-1}$. 
By using Cramer's rule, we obtain: 
\begin{align}
\bmQ_{1,1}^{-1} &= \textstyle{\frac{\mu^3}{(\mu-1)^3}},~ \bmQ_{2,1}^{-1} =  \textstyle{-\frac{\mu^2}{(\mu-1)^3}}, \label{eq:bmQ11_bmQ21}\\
\bmQ_{3,1}^{-1} &= \textstyle{\frac{\mu}{(\mu-1)^3}},~ \bmQ_{4,1}^{-1} = \textstyle{-\frac{1}{(\mu-1)^3}}.
\label{eq:bmQ31_bmQ41}
\end{align}
By (\ref{eq:hc3_hc2_hc1_hc0}), (\ref{eq:bmQ11_bmQ21}), and (\ref{eq:bmQ31_bmQ41}), we obtain:
\begin{align*}
\textstyle{\hat{c}_3 = \frac{\mu^3}{(\mu-1)^3} m_3 - \frac{\mu^2}{(\mu-1)^3} m_2 + \frac{\mu}{(\mu-1)^3} m_1 - \frac{1}{(\mu-1)^3} m_0.}
\end{align*}
Since $\mu = e^\epsilon$ and the empirical estimate is unbiased \cite{Kairouz_ICML16,Wang_USENIX17}, we obtain (\ref{eq:triangle_emp}) in Proposition~\ref{prop:triangle_emp}. \qed

\subsection{Proof of Theorem~\ref{thm:subgraph-rr_LDP}}
Since \alg{LocalRR$_\triangle$} applies the RR to the lower triangular part of the adjacency matrix $\bmA$, it provides $\epsilon$-edge LDP for $(R_1, \ldots, R_n)$. 
Lines 5 to 8 in Algorithm~\ref{alg:subgraph-rr} are post-processing of $(R_1, \ldots, R_n)$. 
Thus, by the immunity to post-processing \cite{DP}, \alg{LocalRR$_\triangle$} provides $\epsilon$-edge LDP for the output $\frac{1}{(\mu-1)^3}(\mu^3 m_3 -\mu^2 m_2 + \mu m_1 - m_0)$. 

In addition, the existence of edge $(v_i,v_j) \in E$ $(i>j)$ affects only one element $a_{i,j}$ in the lower triangular part of $\bmA$. 
Therefore, \alg{LocalRR$_\triangle$} provides $\epsilon$-relationship DP.

\subsection{Proof of Theorem~\ref{thm:subgraph-rr}}
\label{sub:proof_thm:subgraph-rr}
By Proposition~\ref{prop:triangle_emp}, the estimate $\hf_{\triangle}(G, \epsilon)$ by \alg{LocalRR$_\triangle$} is unbiased. 
Then by the bias-variance decomposition \cite{mlpp}, 
the expected $l_2$ loss $\mathbb{E}[l_2^2(\hf_{\triangle}(G, \epsilon), f_\triangle(G))]$ is equal to the variance of $\hf_{\triangle}(G, \epsilon)$. 
Let 
$a_3 = \frac{\mu^3}{(\mu-1)^3}$, 
$a_2 = - \frac{\mu^2}{(\mu-1)^3}$, 
$a_1 = \frac{\mu}{(\mu-1)^3}$, and 
$a_0 = - \frac{1}{(\mu-1)^3}$. 
Then the variance of $\hf_{\triangle}(G, \epsilon)$ can be written as follows:
\begin{align}
    \V[\hf_{\triangle}(G, \epsilon)] 
    &= \V[a_3 m_3 + a_2 m_2 + a_1 m_1 + a_0 m_0] \nonumber\\
    &= a_3^2 \V[m_3] + a_2^2 \V_{RR}[m_2] + a_1^2 \V[m_1] + a_0^2 \V[m_0] \nonumber\\
    &\hspace{3.5mm} + \sum_{i=0}^3 \sum_{j=0, j \ne i}^3 2a_i a_j \text{cov}(m_i, m_j),
    \label{eq:V_a3m3_a2m2_a1m1_a0m0}
\end{align}
where $\text{cov}(m_i, m_j)$ represents the covariance of $m_i$ and $m_j$. 
The covariance $\text{cov}(m_i, m_j)$ can be written as follows:
\begin{align}
    \text{cov}(m_i, m_j)
    &\leq \sqrt{\V[m_i] \V[m_j]} \nonumber\\
    &\hspace{4.2mm} (\text{by Cauchy-Schwarz inequality}) \nonumber\\
    &\leq  \max\{ \V[m_i], \V[m_j]\} \nonumber\\
    &\leq \V[m_i] + \V[m_j].
    \label{eq:cov_mi_mj}
\end{align}
By (\ref{eq:V_a3m3_a2m2_a1m1_a0m0}) and (\ref{eq:cov_mi_mj}), we obtain:
\begin{align}
    &\V[\hf_{\triangle}(G, \epsilon)] \nonumber\\
    &\leq (a_3^2 + 4a_3(a_2 + a_1 + a_0)) \V[m_3] \nonumber\\
    & \hspace{4.5mm} + (a_2^2 + 4a_2(a_3 + a_1 + a_0)) \V[m_2] \nonumber\\
    & \hspace{4.5mm} + (a_1^2 + 4a_1(a_3 + a_2 + a_0)) \V[m_1] \nonumber\\
    & \hspace{4.5mm} + (a_0^2 + 4a_0(a_3 + a_2 + a_1)) \V[m_0] \nonumber\\
    &= O\left( \frac{e^{6\epsilon}}{(e^\epsilon-1)^6} (\V[m_3] + \V[m_2] + \V[m_1] + \V[m_0]) \right).
    % &\leq O\left(\frac{e^{6\epsilon}}{(e^\epsilon-1)^6} \V[m_3] + \frac{e^{5\epsilon}}{(e^\epsilon-1)^6}\V[m_2]
    % \right. \nonumber\\
    % & \hspace{10mm} \left. + \frac{e^{4\epsilon}}{(e^\epsilon-1)^6}\V[m_1] + \frac{e^{3\epsilon}}{(e^\epsilon-1)^6}\V[m_0]\right).
    \label{V_RR_m_3}
\end{align}

Below we calculate $\V[m_3]$, $\V[m_2]$, $\V[m_1]$, and $\V[m_0]$ by assuming the Erd\"{o}s-R\'{e}nyi model $\bmG(n, \alpha)$ for $G$:

\begin{lemma}\label{lem:erdos-renyi-variance}
  Let $G \sim \textbf{G}(n,\alpha)$. 
  %with $\alpha < \frac{1}{2}$. 
  Let $p = \frac{1}{e^\epsilon+1}$ and 
  $\beta = \alpha(1-p) + (1-\alpha)p$. 
  Then $\V[m_3] = O(\beta^5 n^4 + \beta^3
  n^3)$, $\V[m_2] = O(\beta^3 n^4 + \beta^2 n^3)$, and 
  $\V[m_1] = \V[m_0] = O(\beta n^4)$.
\end{lemma}
Before going into the proof of Lemma~\ref{lem:erdos-renyi-variance}, we prove Theorem~\ref{thm:subgraph-rr} using Lemma~\ref{lem:erdos-renyi-variance}. 
By (\ref{V_RR_m_3}) and Lemma~\ref{lem:erdos-renyi-variance}, we obtain: 
\begin{align*}
\V[\hf_{\triangle}(G, \epsilon)] = O\left( \frac{e^{6\epsilon}}{(e^\epsilon-1)^6} \beta n^4 \right),
\end{align*}
which proves Theorem~\ref{thm:subgraph-rr}. \qed

We now prove Lemma~\ref{lem:erdos-renyi-variance}:

\begin{proof}[Proof of Lemma~\ref{lem:erdos-renyi-variance}]
Fist we show the variance of $m_3$ and $m_0$. 
Then we show the variance of $m_2$ and $m_1$.

\smallskip
\noindent{\textbf{Variance of $m_3$ and $m_0$.}}~~Since each edge in the original graph $G$ is independently generated with probability $\alpha \in [0,1]$, each edge in the noisy graph $G'$ is independently generated with probability $\beta = \alpha (1-p) + (1 - \alpha) p \in [0,1]$, where $p=\frac{1}{e^\epsilon+1}$. 
Thus $m_3$ is the number of triangles in graph $G' \sim \textbf{G}(n,\beta)$.

For $i,j,k \in [n]$, let $y_{i,j,k} \in \{0,1\}$ be a variable that takes $1$ if and only if 
$(v_i, v_j, v_k)$ forms a triangle. 
Then $\mathbb{E}[m_3^2]$ can be written as follows:
\begin{align}
  \mathbb{E}[m_3^2] = \sum_{i<j<k} ~ \sum_{i'<j'<k'}
  \mathbb{E}[y_{i,j,k} y_{i',j',k'}] 
  \label{eq:E_RR_tm3}
\end{align}
$\mathbb{E}[y_{i,j,k} y_{i',j',k'}]$ in (\ref{eq:E_RR_tm3}) is the probability that both $(v_i,v_j,v_k)$ and $(v_{i'},v_{j'},v_{k'})$ form a triangle. 
This event can be divided into the following four types:
\begin{enumerate}
\item $(i,j,k)=(i',j',k')$. There are $\binom{n}{3}$ such terms in (\ref{eq:E_RR_tm3}). 
For each term, $\mathbb{E}[y_{i,j,k} y_{i',j',k'}] = \beta^3$.
\item $(i,j,k)$ and $(i',j',k')$ have two elements in common. 
There are $\binom{n}{2} (n-2) (n-3) = 12\binom{n}{4}$ such terms in (\ref{eq:E_RR_tm3}). 
For each term, $\mathbb{E}[y_{i,j,k} y_{i',j',k'}] = \beta^5$. 
\item $(i,j,k)$ and $(i',j',k')$ have one element in common. 
There are $n \binom{n-1}{2} \binom{n-3}{2} = 30\binom{n}{5}$ such terms in (\ref{eq:E_RR_tm3}). 
For each term, $\mathbb{E}[y_{i,j,k} y_{i',j',k'}] = \beta^6$. 
\item $(i,j,k)$ and $(i',j',k')$ have no common elements. 
There are $\binom{n}{3} \binom{n-3}{3} = 20\binom{n}{6}$ such terms in in (\ref{eq:E_RR_tm3}). 
For each term, $\mathbb{E}[y_{i,j,k} y_{i',j',k'}] = \beta^6$. 
\end{enumerate}
Moreover, $\mathbb{E}[m_3]^2 = \binom{n}{3}^2 \beta^6$. 
Therefore, the variance of $m_3$ can be written as follows:
\begin{align*}
    \V[m_3] 
    &= \textstyle{\binom{n}{3} \beta^3 + 12\binom{n}{4} \beta^5 + 30\binom{n}{5} \beta^6 + 20\binom{n}{6} \beta^6 - \binom{n}{3}^2 \beta^6} \\
    &= \textstyle{\binom{n}{3} \beta^3 (1-\beta^3) + 12\binom{n}{4}\beta^5(1-\beta)} \\
    &= O(\beta^5 n^4 + \beta^3 n^3).
\end{align*}

By changing $\beta$ to $1-\beta$ and counting triangles, we get a random variable with the same distribution as $m_0$. Thus,
\begin{align*}
  \V[m_0] 
  &= \textstyle{\binom{n}{3} (1-\beta)^3(1-(1-\beta)^3) + 12\binom{n}{4} (1-\beta)^5\beta}
  \\
  &= O(\beta n^4).
\end{align*}
\smallskip
\noindent{\textbf{Variance of $m_2$ and $m_1$.}}~~For $i,j,k \in [n]$, let $z_{i,j,k} \in \{0,1\}$ be a variable that takes $1$ if and only if 
$(v_i, v_j, v_k)$ forms $2$-edges (i.e., exactly one edge is missing in the three nodes). 
Then $\mathbb{E}[m_2^2]$ can be written as follows:
\begin{align}
  \mathbb{E}[m_2^2] = \sum_{i<j<k} \sum_{i'<j'<k'}
  \mathbb{E}[z_{i,j,k} z_{i',j',k'}] 
  \label{eq:E_RR_tm2}
\end{align}
$\mathbb{E}[z_{i,j,k} z_{i',j',k'}]$ in (\ref{eq:E_RR_tm2}) is the probability that both $(v_i,v_j,v_k)$ and $(v_{i'},v_{j'},v_{k'})$ form $2$-edges. 
This event can be divided into the following four types:
\begin{enumerate}
	\item $(i,j,k)=(i',j',k')$. There are $\binom{n}{3}$ such terms in (\ref{eq:E_RR_tm2}). 
    For each term, $\mathbb{E}[z_{i,j,k} z_{i',j',k'}]=3\beta^2(1-\beta)$. 
	\item $(i,j,k)$ and $(i',j',k')$ have two elements in common. 
	There are $\binom{n}{2}(n-2)(n-3) = 12 \binom{n}{4}$ such terms in (\ref{eq:E_RR_tm2}). 
	For example, consider a term in which $i=i'=1$, $j=j'=2$, $k=3$, and $k'=4$. 
	Both $(v_1,v_2,v_3)$ and $(v_1,v_2,v_4)$ form 2-edges if:\\
	(a) $(v_1,v_2), (v_1,v_3), (v_1,v_4) \in E'$, $(v_2,v_3), (v_2,v_4) \notin E'$, \\
	(b) $(v_1,v_2), (v_1,v_3), (v_2,v_4) \in E'$, $(v_2,v_3), (v_1,v_4) \notin E'$, \\
	(c) $(v_1,v_2), (v_2,v_3), (v_1,v_4) \in E'$, $(v_1,v_3), (v_2,v_4) \notin E'$, \\
	(d) $(v_1,v_2), (v_2,v_3), (v_2,v_4) \in E'$, $(v_1,v_3), (v_1,v_4) \notin E'$, or \\
	(e) $(v_1,v_3), (v_1,v_4), (v_2,v_3), (v_2,v_4) \in E'$, $(v_1,v_2) \notin E'$. \\
    Thus, $\mathbb{E}[z_{i,j,k} z_{i',j',k'}]=4\beta^3(1-\beta)^2 + \beta^4(1-\beta)$ for this term. 
    Similarly, $\mathbb{E}[z_{i,j,k} z_{i',j',k'}]=4\beta^3(1-\beta)^2 + \beta^4(1-\beta)$ for the other terms.
	\item $(i,j,k)$ and $(i',j',k')$ have one element in common. 
	There are $n \binom{n-1}{2} \binom{n-3}{2} = 30\binom{n}{5}$ such terms in (\ref{eq:E_RR_tm2}).  For each term, $\mathbb{E}[z_{i,j,k}z_{i',j',k'}]=(3\beta^2(1-\beta))^2 = 9\beta^4(1-\beta)^2$. 
	\item $(i,j,k)$ and $(i',j',k')$ have no common elements. 
	There are $\binom{n}{3}\binom{n-3}{3} = 20\binom{n}{6}$ such terms in (\ref{eq:E_RR_tm2}). 
	For each term, $\mathbb{E}[z_{i,j,k}z_{i',j',k'}]=(3\beta^2(1-\beta))^2 = 9\beta^4(1-\beta)^2$. 
\end{enumerate}
Moreover, $\mathbb{E}[m_2]^2 = (3\binom{n}{3}\beta^2(1-\beta))^2 = 9\binom{n}{3}^2\beta^4(1-\beta)^2$. 
Therefore, the variance of $m_2$ can be written as follows:
\begin{align*}
  \mathbb{V}[m_2] 
  &= \mathbb{E}[m_2^2] - \mathbb{E}[m_2]^2 \nonumber\\
  &= \textstyle{3\binom{n}{3}\beta^2(1-\beta) + 12\binom{n}{4}\left(4\beta^3(1-\beta)^2 + \beta^4(1-\beta)\right)} \nonumber\\
  &\hspace{3.5mm} \textstyle{+ 270\binom{n}{5}\beta^4(1-\beta)^2 + 180\binom{n}{6}\beta^4(1-\beta)^2} \nonumber\\
  &\hspace{3.5mm} \textstyle{- 9\binom{n}{3}^2\beta^4(1-\beta)^2.}
\end{align*}
By simple calculations,
\begin{align*}
\textstyle{270\binom{n}{5} + 180\binom{n}{6} - 9\binom{n}{3}^2 = -108\binom{n}{4}-9\binom{n}{3}.}
\end{align*}
Thus we obtain:
\begin{align*}
\mathbb{V}[m_2] 
&= \textstyle{3\binom{n}{3}\beta^2(1-\beta)\left(1 - 3\beta^2(1-\beta)\right)} \nonumber\\
&\hspace{3.5mm} \textstyle{+ 12\binom{n}{4} \beta^3(1-\beta) \left(4(1-\beta) + \beta - 9\beta(1-\beta) \right)} \\
&= O(\beta^3 n^4 + \beta^2n^3).
\end{align*}
Similarly, the variance of $m_1$ can be written as follows:
\begin{align*}
\mathbb{V}[m_1] 
&= \textstyle{3\binom{n}{3}\beta(1-\beta)^2\left(1 - 3\beta(1-\beta)^2\right)} \nonumber\\
&\hspace{3.5mm} \textstyle{+ 12\binom{n}{4} \beta(1-\beta)^3 \left(4\beta +
(1-\beta) - 9\beta(1-\beta) \right)} \\
&= O(\beta n^4).
\end{align*}
\end{proof}

\subsection{Proof of Proposition~\ref{prop:triangle_emp_2rounds}}
Let $t_* = \sum_{i=1}^n t_i$ and $s_* = \sum_{i=1}^n s_i$. 
Let $s_*^{\wedge}$ be the number of triplets $(v_i,v_j,v_k)$ such that $j<k<i$, $a_{i,j} = a_{i,k} = 1$, and $a_{j,k} = 0$. 
Let $s_*^{\triangle}$ be the number of triplets $(v_i,v_j,v_k)$ such that $j<k<i$, $a_{i,j} = a_{i,k} = a_{j,k} =1$. 
Note that 
$s_* = s_*^{\wedge} + s_*^{\triangle}$ and 
$s_*^{\triangle} = f_\triangle(G)$. 

Consider a triangle $(v_i,v_j,v_k) \in G$. 
This triangle is counted $1-p_1$ ($= \frac{e^{\epsilon_1}}{e^{\epsilon_1}+1}$) times in expectation in $t_*$. 
Consider $2$-edges $(v_i,v_j,v_k) \in G$ (i.e., exactly one edge is missing in the three nodes). 
This is counted $p_1$ ($= \frac{1}{e^{\epsilon_1}+1}$) times in expectation in $t_*$.  
No other events can change $t_*$. 
Therefore, we obtain:
\begin{align*}
\mathbb{E}[t_*] = (1-p_1) s_*^{\triangle} + p_1 s_*^{\wedge}. 
\end{align*}
By $s_* = s_*^{\wedge} + s_*^{\triangle}$ and 
$s_*^{\triangle} = f_\triangle(G)$, we obtain:
\begin{align*}
\mathbb{E}\left[\sum_{i=1}^n w_i \right] 
&= \mathbb{E}\left[\sum_{i=1}^n (t_i - p_1 s_i) \right] \\
&= \mathbb{E}[t_* - p_1 s_*] \\
&= \mathbb{E}[t_*] - p_1 \mathbb{E}[s_*^{\wedge} + s_*^{\triangle}] \\
&= (1-p_1) s_*^{\triangle} + p_1 s_*^{\wedge} - p_1 (s_*^{\wedge} + s_*^{\triangle}) \\
&= (1 - 2 p_1) f_\triangle(G),
\end{align*}
hence 
\begin{align*}
\textstyle{\mathbb{E}\left[ \frac{1}{1-2p_1} \sum_{i=1}^n w_i \right] = f_\triangle(G).}
\end{align*}
\qed

\subsection{Proof of Theorem~\ref{thm:local2rounds_LDP}}
Let $\calR_i$ be \alg{Local2Rounds$_{\triangle}$}. 
Consider two neighbor lists $\bma_i,\bma'_i \in \{0,1\}^n$ that differ in one bit. 
Let $d_i$ (resp.~$d'_i$) $\in \nnints$ be the number of ``1''s in $\bma_i$ (resp.~$\bma'_i$). 
Let $\bbma_i$ (resp.~$\bbma'_i$) $\in \{0,1\}^n$ be neighbor lists obtained by setting all of the $i$-th to the $n$-th elements in $\bma_i$ (resp.~$\bma'_i$) to $0$. 
Let $\bd_i$ (resp.~$\bd'_i$) $\in \nnints$ be the number of ``1''s in $\bbma_i$ (resp.~$\bbma'_i$). 
For example, if $n=6$, $\bma_4=(1,0,1,0,1,1)$, and $\bma'_4=(1,1,1,0,1,1)$, then 
$d_4=4$, $d'_4=5$, $\bbma_4=(1,0,1,0,0,0)$, $\bbma'_4=(1,1,1,0,0,0)$, $\bd_4=2$, and $\bd'_4=3$. 

Furthermore, 
let $t_i$ (resp.~$t'_i$) $\in \nnints$ be the number of triplets $(v_i, v_j, v_k)$ such that $j < k < i$, $(v_i,v_j) \in E$, $(v_i,v_k) \in E$, and $(v_j,v_k) \in E'$ in $\bma_i$ (resp.~$\bma'_i$). 
Let $s_i$ (resp.~$s'_i$) $\in \nnints$ be the number of triplets $(v_i, v_j, v_k)$ such that $j < k < i$, $(v_i,v_j) \in E$, and $(v_i,v_k) \in E$ in $\bma_i$ (resp.~$\bma'_i$). 
Let $w_i = t_i - p_1 s_i$ and $w'_i = t'_i - p_1 s'_i$. 
Below we consider two cases about $d_i$: when $d_i < \td_{max}$ and when $d_i \geq \td_{max}$. 

\smallskip
\noindent{\textbf{Case 1: $d_i < \td_{max}$.}}~~Assume that $d'_i = d_i + 1$. 
In this case, we have either $\bbma'_i = \bbma_i$ or $\bd'_i = \bd_i+1$. 
If $\bbma'_i = \bbma_i$, then $s_i = s'_i$, $t_i = t'_i$, and $w_i = w'_i$, hence $\Pr[\calR_i(\bma_i) = \hw_i] = \Pr[\calR_i(\bma'_i) = \hw_i]$. 
If $\bd'_i = \bd_i+1$, then $s_i$ and $s'_i$ can be expressed as $s_i = \binom{\bd_i}{2}$ and $s'_i = \binom{\bd'_i}{2} = \binom{\bd_i+1}{2}$, respectively. 
Then we obtain:
\begin{align*}
s'_i - s_i = \binom{\bd_i+1}{2} - \binom{\bd_i}{2} = \bd_i.
\end{align*}
In addition, since we consider an additional constraint ``$(v_j,v_k) \in E'$'' in counting $t_i$ and $t'_i$, 
we have $t'_i - t_i \leq s'_i - s_i$. 
Therefore, 
\begin{align*}
|w'_i - w_i| 
&= |t'_i - t_i - p_1 (s'_i - s_i)| \\
&\leq (1 - p_1) \bd_i \\
&\leq (1 - p_1) d_i \\
&< \td_{max} \hspace{5mm} \text{(by $p_1 > 0$ and $d_i < \td_{max}$)}.
\end{align*}
Since we add $\Lap(\frac{\td_{max}}{\epsilon_2})$ to $w_i$, we obtain:
\begin{align}
\Pr[\calR_i(\bma_i) = \hw_i] \leq e^{\epsilon_2} \Pr[\calR_i(\bma'_i) = \hw_i]. 
\label{eq:R_i_a_i_hr_i_2}
\end{align}

Assume that $d'_i = d_i - 1$. 
In this case, we have either $\bbma'_i = \bbma_i$ or $\bd'_i = \bd_i-1$. 
If $\bbma'_i = \bbma_i$, then $\Pr[\calR_i(\bma_i) = \hw_i] = \Pr[\calR_i(\bma'_i) = \hw_i]$. 
If $\bd'_i = \bd_i-1$, then we obtain $s_i - s'_i = \bd_i - 1$ and $t_i - t'_i \leq s_i - s'_i$. 
Thus $|w'_i - w_i| \leq (1-p_1) (\td_i - 1) < \td_{max}$ and (\ref{eq:R_i_a_i_hr_i_2}) holds. 
Therefore, \alg{Local2Rounds$_{\triangle}$} provides $\epsilon_2$-edge LDP at the second round. 
Since \alg{Local2Rounds$_{\triangle}$} provides $\epsilon_1$-edge LDP at the first round (by Theorem~\ref{thm:subgraph-rr_LDP}), it provides $(\epsilon_1 + \epsilon_2)$-edge LDP in total by the composition theorem \cite{DP}. 

\smallskip
\noindent{\textbf{Case 2: $d_i \geq \td_{max}$.}}~~Assume that $d'_i = d_i + 1$. 
In this case, we obtain $d_i = d'_i = \td_{max}$ after graph projection. 

Note that $\bma_i$ and $\bma'_i$ can differ in \textit{zero or two bits} after graph projection. 
For example, consider the case where $n=8$, $\bma_5=(1,1,0,1,0,1,1,1)$, $\bma'_5=(1,1,1,1,0,1,1,1)$, and $\td_{max}=4$. 
If the permutation is 1,4,6,8,2,7,5,3, then $\bma_5=\bma'_5=(1,0,0,1,0,1,0,1)$ after graph projection. 
However, if the permutation is 3,1,4,6,8,2,7,5, then $\bma_5$ and $\bma'_5$ become $\bma_5=(1,0,0,1,0,1,0,1)$ and $\bma'_5=(1,0,1,1,0,1,0,0)$, respectively; i.e., they differ in the third and eighth elements. 

If $\bma_i=\bma'_i$, then $\Pr[\calR_i(\bma_i) = \hw_i] = \Pr[\calR_i(\bma'_i) = \hw_i]$. 
If $\bma_i$ and $\bma'_i$ differ in two bits, $\bbma_i$ and $\bbma'_i$ differ in \textit{at most two bits} (because we set all of the $i$-th to the $n$-th elements in $\bma_i$ and $\bma'_i$ to $0$). 
For example, we can consider the following three cases:
\begin{itemize}
    \item If $\bma_5=(1,0,0,1,0,1,0,1)$ and $\bma'_5=(1,0,0,1,0,1,1,0)$, then $\bbma_5=\bbma'_5=(1,0,0,1,0,0,0,0)$. 
    \item If $\bma_5=(1,0,0,1,0,1,0,1)$ and $\bma'_5=(1,0,1,1,0,1,0,0)$, then $\bbma_5=(1,0,0,1,0,0,0,0)$ and $\bbma'_5=(1,0,1,1, 0,0,\allowbreak0,0)$; i.e., they differ in one bit. 
    \item If $\bma_5=(1,1,0,1,0,1,0,0)$ and $\bma'_5=(1,0,1,1,0,1,0,0)$, then $\bbma_5=(1,1,0,1,0,0,0,0)$ and $\bbma'_5=(1,0,1,1,0,0,\allowbreak0,0)$; i.e., they differ in two bits.
\end{itemize}
If $\bbma_i=\bbma'_i$, then $\Pr[\calR_i(\bma_i) = \hw_i] = \Pr[\calR_i(\bma'_i) = \hw_i]$. 
If $\bbma_i$ and $\bbma'_i$ differ in one bit, then $\bd'_i = \bd_i + 1$. 
In this case, we obtain (\ref{eq:R_i_a_i_hr_i_2}) in the same way as \textbf{Case 1}. 

We need to be careful when $\bbma_i$ and $\bbma'_i$ differ in two bits. 
In this case, $\bd'_i = \bd_i$ (because $d_i = d'_i = \td_{max}$ after graph projection). 
Then we obtain $s_i = s'_i = \binom{\td_{max}}{2}$. 
Since the number of $2$-stars that involve a particular user in $\bbma_i$ is $\bd_i - 1$, we obtain $t'_i - t_i \leq \bd_i - 1$. Therefore,
\begin{align*}
|w'_i - w_i| = |t'_i - t_i| \leq \bd_i - 1 < \td_{max},
\end{align*}
and (\ref{eq:R_i_a_i_hr_i_2}) holds. 
Therefore, if $d'_i = d_i + 1$, then 
\alg{Local2Rounds$_{\triangle}$} provides $(\epsilon_1 + \epsilon_2)$-edge LDP in total. 

Assume that $d'_i = d_i - 1$. 
If $d_i > \td_{max}$, then $d_i = d'_i = \td_{max}$ after graph projection. 
Thus \alg{Local2Rounds$_{\triangle}$} provides $(\epsilon_1 + \epsilon_2)$-edge LDP in total in the same as above. 
If $d_i = \td_{max}$, then we obtain (\ref{eq:R_i_a_i_hr_i_2}) in the same way as \textbf{Case 1}, and therefore \alg{Local2Rounds$_{\triangle}$} provides $(\epsilon_1 + \epsilon_2)$-edge LDP in total.

\smallskip
In summary, \alg{Local2Rounds$_{\triangle}$} provides $(\epsilon_1 + \epsilon_2)$-edge LDP in both \textbf{Case 1} and \textbf{Case 2}. 
\alg{Local2Rounds$_{\triangle}$} also provides $(\epsilon_1 + \epsilon_2)$-relationship DP 
because it uses only the lower triangular part of the adjacency matrix $\bmA$. \qed

\subsection{Proof of Theorem~\ref{thm:local2rounds}}
  When the maximum degree $d_{max}$ of $G$ is at most $\tilde{d}_{max}$, no graph
  projection will occur.
  By Proposition~\ref{prop:triangle_emp_2rounds}, the estimate
  $f_\triangle(G,\varepsilon)$ by \alg{Local2Rounds$_\triangle$} is unbiased.

  By bias-variance
  decomposition~\eqref{eq:bias-var}, the expected $l_2$ loss
  $\E[l_2^2(\hf_\triangle(G, \epsilon), f_\triangle(G))]$ is equal to
  $\V[\hf_\triangle(G,\epsilon)]$.
  Recall that $p_1 = \frac{1}{1+e^{\epsilon_1}}$.
  %This quantity 
  $\V[\hf_\triangle(G,\epsilon)]$ 
  can be written as follows:
  \begin{align}
    &\V[\hf_\triangle(G,\epsilon)] \\
    &= \textstyle{\frac{1}{(1-2p_1)^2}\V\left[\sum_{i=1}^n
    \hat{w}_i\right]} \nonumber \\
    &= \textstyle{\frac{1}{(1-2p_1)^2}\V\left[\sum_{i=1}^n
    t_i - p_1 s_i +
    \Lap(\frac{\tilde{d}_{max}(1-p_1)}{\varepsilon_2})\right]} \nonumber \\
    &= \textstyle{\frac{1}{(1-2p_1)^2}\left(\V\left[\sum_{i=1}^n
    t_i - p_1 s_i \right] +
    \V\left[\sum_{i=1}^n
    \Lap(\frac{\tilde{d}_{max}(1-p_1)}{\varepsilon_2})\right]\right)} \nonumber \\
    &= \textstyle{\frac{1}{(1-2p_1)^2}\V\left[\sum_{i=1}^n
    t_i \right] + \frac{n}{(1-2p_1)^2}
    2\frac{\tilde{d}_{max}^2(1-p_1)^2}{\varepsilon_2^2}}.\label{eq:local2rounds-var}
  \end{align}

  In the last line, we are able to get rid of the $s_i$'s because they are
  deterministic. We are also able to sum the variances of the $\Lap$ random
  variables since they are independent; we are not able to do the same with the
  sum of the $t_i$s. 

  Recall the definition of $E'$ computed by the first round of
  \alg{Local2Rounds$_\triangle$}---the noisy edges released by randomized
  response. Now,
  \begin{align*}
    t_i &= \sum_{a_{i,j}=a_{i,k}=1, j<k<i} \textbf{1}((v_j,v_k) \in E'). \\
  \end{align*}
  This gives
  \begin{align*}
    \sum_{i=1}^n t_i &= \sum_{i=1}^n\sum_{\substack{a_{i,j}=a_{i,k}=1 \\ j<k<i }} \textbf{1}((v_j,v_k) \in
    E') \\
    &= \sum_{1 \leq j < k \leq n} \sum_{\substack{i > k \\ a_{i,j}=a_{i,k}=1
    }} \textbf{1}((v_j,v_k) \in E') \\
    &= \sum_{1 \leq j < k \leq n} |\{i : i>k, a_{i,j}=a_{i,k}=1\}| \textbf{1}((v_j,v_k)
    \in E'|.
  \end{align*}
  Let $c_{jk} = |\{i : i>k, a_{i,j}=a_{i,k}=1\}|$. Notice that $\textbf{1}( (v_j,v_k) \in E')$
  are independent events. Thus, the variance of the above expression is
  \begin{align}
    \V\left[\sum_{i=1}^n t_i\right] &= 
    %\V_P
    \V
    \left[\sum_{1 \leq j < k \leq n}
    c_{jk} \textbf{1}( (v_j,v_k) \in E') \right] \nonumber \\
    &= \sum_{1 \leq j < k \leq n} c_{jk}^2 \V[\textbf{1}( (v_j,v_k \in E')) ]
    \nonumber \\
    &= p_1 (1-p_1) \sum_{1 \leq j < k \leq n} c_{jk}^2.
    \label{eq:subgraph-interactive-var}
  \end{align}
  $c_{jk}$ is the number of ordered 2-paths from $j$ to $k$ in $G$. Because 
  %$d_{max}$ is the maximum
  %degree of vertex $j$, 
  the degree of user $v_j$ is at most $\td_{max}$, 
  $0 \leq c_{jk} \leq \td_{max}$. There are at most
  $n\td_{max}^2$ ordered
  2-paths in $G$, since there are only $\td_{max}$ 
  %vertices 
  nodes 
  to go to once a first is
  picked. Thus, $\sum_{1 \leq j < k \leq n} c_{jk} \leq n\td_{max}^2$. Using a Jensen's inequality
  style argument, the best way to maximize~\eqref{eq:subgraph-interactive-var} is to have
  all $c_{jk}$ be $0$ or $\td_{max}$. At most $n\td_{max}$ of the 
  %$c_{ij}$
  $c_{jk}$ 
  can be $\td_{max}$, and the rest are zero. Thus,
  \begin{align*}
    \V\left[\sum_{i=1}^n t_i \right] &=
    p_1(1-p_1) \sum_{1 \leq j < k \leq n} c_{ij}^2 \\ 
    &\leq p_1(1-p_1) n\td_{max} \times \td_{max}^2.
  \end{align*}
  Plugging this into~\eqref{eq:local2rounds-var}
  \begin{align*}
    \V[\hf_\triangle(G,\epsilon)] &\leq \frac{p_1(1-p_1)n\td_{max}^3}{(1-2p_1)^2} +
    \frac{2n\td_{max}^2(1-p_1)^2}{(1-2p_1)^2\varepsilon_2^2} \\
    &\leq O\left(\frac{p_1 n\td_{max}^3 + n\td_{max}^2/\varepsilon_2^2}{(1-2p_1)^2} \right) \\
    &\leq O\left(\frac{e^{\varepsilon_1}}{(1-e^{\varepsilon_1})^2} \left(n\td_{max}^3 +
    \frac{e^{\varepsilon_1}}{\epsilon_2^2}n\td_{max}^2\right)\right).
  \end{align*}
  \qed

\subsection{Proof of Theorem~\ref{thm:lower-bound}}
\label{sub:proof_thm_lower-bound}

\paragraph{Preliminaries.}
We begin by defining a Boolean version of the independent cube in Definition~\ref{def:mono-cube}, which we call the \textit{Boolean independent cube}. 
The Boolean independent cube 
works for functions $g : \{0,1\}^\kappa \rightarrow \reals$ in the local DP model, where each of $\kappa \in \nats$ users has a \textit{single bit} and obfuscates the bit to provide $\epsilon$-DP. 
As shown later, there is a one-to-one correspondence between the independent cube in Definition~\ref{def:mono-cube} and the Boolean independent cube. 
Based on this, we show a lower-bound for the Boolean independent cube, and use the lower-bound to prove Theorem~\ref{thm:lower-bound}.

Below we define the Boolean independent cube. 
For $i\in[\kappa]$, let $x_i \in \{0,1\}$ be a bit of user $v_i$. 
Let $X = (x_1, \ldots, x_\kappa)$. 
We assume user $v_i$ 
obfuscates $x_i$ using 
a randomizer
$\calS_i : \{0,1\} \rightarrow \mathcal{Z}_i$, where $\calS_i$ 
satisfies $\epsilon$-DP and $\mathcal{Z}_i$ is a range of $\calS_i$. 
Examples of $\calS_i$ include Warner's RR. 
Furthermore, we assume the one-round setting, where each $\calS_i$ is independent, 
and where the estimator $\hg$ for $g$ has the form
\begin{equation}\label{eq:one-round-lower-2}
  \hg(X) = \tilde{g}(\calS_1(x_1), \ldots, \calS_\kappa(x_\kappa)).
\end{equation}
$\tilde{g}$ is an aggregate function that takes $\calS_1(x_1), \ldots, \calS_\kappa(x_\kappa)$ as input and outputs $\hg(X)$. 

We will prove a lower bound which uses the following stripped-down form of an independent cube (Definition~\ref{def:mono-cube}).
\begin{definition}\label{def:mono-cube-boolean}[Boolean $(\kappa,D)$-independent cube]
  Let $g : \{0,1\}^\kappa \rightarrow \reals$, and $D \in \reals$.
  We say 
  $g$ has 
  %$\{0,1\}^\kappa$ is 
  a \emph{Boolean $(\kappa,D)$-independent cube} 
  %$\{0,1\}^\kappa$
  if for all
  $(x_1, \ldots, x_\kappa) \in \{0,1\}^\kappa$ we have 
  \[
    g(x_1, \ldots, x_\kappa) = g(0,0,\ldots,0) + \sum_{i=1}^\kappa x_i C_i,
  \]
  where $C_i \in \reals$ satisfies $|C_i| \geq D$ for any $i \in [\kappa]$.
\end{definition}

The following theorem applies to 
the Boolean independent cube 
and will help us establish
Theorem~\ref{thm:lower-bound}. We prove this theorem in
Section~\ref{sub:proof_thm_lower-bound-dir}.
\begin{theorem}\label{thm:lower-bound-dir}
  Let $g : \calX^{\kappa} \rightarrow \reals$ be a function 
  %and $\calA$ be 
  that has 
  a Boolean
  $(\kappa,D)$-independent cube. 
  %for $f$. 
  Let $\hg(X)$ be an estimator having the form
  of~\eqref{eq:one-round-lower-2}, where 
  each $\calS_i$ %satisfies 
  %$\calS_1, \ldots, \calS_\kappa$ 
  provides $\epsilon$-DP and 
  % the $\calS_i$ are mutually independent. 
  is 
  mutually 
  %are 
  independent. 
  Let $X$ be drawn uniformly from
  $\{0,1\}^\kappa$. 
  %If each $\calS_i$ provides $\epsilon$-DP, then 
  Over the randomness both in selecting $X$ and in 
  %the $\calS_i$, 
  $\calS_1, \ldots, \calS_\kappa$, 
  $\E_{X, \calS_1, \ldots, \calS_\kappa}[l_2^2(g(X), \hg(X))] =
  \Omega\left(\frac{e^\epsilon}{(e^\epsilon+1)^2} \kappa D^2\right)$.
\end{theorem}

\paragraph{Proof of Theorem~\ref{thm:lower-bound} using Theorem~\ref{thm:lower-bound-dir}.}
To prove Theorem~\ref{thm:lower-bound}, let $\calA$ be the $(n,D)$-independent 
cube (Definition~\ref{def:mono-cube}) for $f$ given in the statement of
Theorem~\ref{thm:lower-bound}. Let $G$ be the graph, 
and 
$\bmA$ be the corresponding symmetric adjacency matrix. 
Below we sometimes write $f$ as a function on neighbor lists $\bma_1, \ldots, \bma_n$ (rather than $G$) because there is a one-to-one correspondence between $G$ and $\bma_1, \ldots, \bma_n$.
Let $M$ be the perfect
matching that defines $\calA$. Let $n=2\kappa$.

The idea is to pair up users that $M$ matches to make a new function $g$ that has a Boolean $(\kappa,D)$-independent cube and new randomizers $\calS_1, \ldots, \calS_\kappa$ that
satisfy $\epsilon$-DP. 
In other words, we regard a pair of users in $M$ as a \textit{virtual} user (since $n=2\kappa$, there are $\kappa$ virtual users in total). 
Then we apply Theorem~\ref{thm:lower-bound-dir}. 

Assume that 
$M = \{(v_1, v_2), (v_3, v_4), \ldots, (v_{2\kappa-1}, v_{2\kappa})\}$ without loss of generality 
(we can construct $g$ and $\calS_1, \ldots, \calS_\kappa$ for arbitrary $M$ in the same way). 
For $x_1, \ldots, x_\kappa \in \{0,1\}$, 
define
\begin{align*}
  g(x_1, \ldots, x_\kappa) = & f(\bma_1 + x_1 \bme_2,~ \bma_2 + x_1 \bme_1,~ \ldots,  \\
  &\hspace{3.3mm}\bma_{2\kappa-1} + x_\kappa \bme_{2\kappa},~ \bma_{2\kappa} + x_\kappa \bme_{2\kappa-1}),
\end{align*}
where $\bme_i \in \{0,1\}^n$ is the $i$-th standard basis vector that has $1$ in the $i$-th coordinate and $0$ elsewhere. 
In other words, 
$x_i \in \{0,1\}$ indicates whether the $i$-th edge in $M$ should be
added to 
$G$. 
Thus, $g$ has a Boolean $(\kappa,D)$-independent cube, and 
there is a one-to-one correspondence between an $(n,D)$-independent cube $\calA$ in Definition~\ref{def:mono-cube} 
and 
$(\kappa,D)$-Boolean independent cube $\{0,1\}^\kappa$ in Definition~\ref{def:mono-cube-boolean}. 
Figure~\ref{fig:Bool-cube} shows a $(2,2)$-Boolean independent cube for $g$ corresponding to the $(4,2)$-independent cube for $f$ in Figure~\ref{fig:mono-cube}.

\begin{figure}[t]
  \centering
  \includegraphics[width=0.88\linewidth]{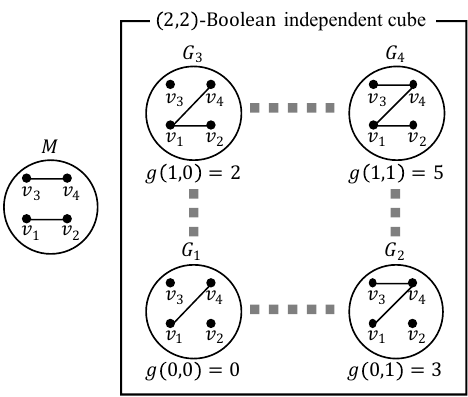}
  \vspace{-4mm}
  \caption{
    $(2,2)$-Boolean independent cube for $g$ corresponding to the $(4,2)$-independent cube for $f$ in Figure~\ref{fig:mono-cube}.
  }\label{fig:Bool-cube}
\end{figure}

Now, for $i \in [\kappa]$, define $\calS_i(x_i)$ for $x_i \in \{0,1\}$ by
\begin{align}
  \calS_i(x_i) = (\calR_{2i-1}(\bma_{2i-1} + x_i \bme_{2i}),
  \calR_{2i}(\bma_{2i} + x_i \bme_{2i-1})).
  \label{eq:S_i_R_2i_1}
\end{align}
In other words, $\calS_i(x_i)$ is simply the product of the outputs of users
$(v_{2i-1}, v_{2i})$, with $x_i$ indicating whether to add the edge in $M$ between them.

Assume that each $\calR_i$ is mutually independent and that $(\calR_1, \ldots, \calR_n)$ provides $\epsilon$-relationship DP in Definition~\ref{def:entire_edge_LDP}. 
Then by (\ref{eq:entire_edge_LDP}) and (\ref{eq:S_i_R_2i_1}), each $\calS_i$ provides $\epsilon$-DP and is mutually independent.

Define the estimator $\hat{g}$ by
\begin{align*}
  \hat{g}(x_1, \ldots, x_\kappa) &= \tilde{f}(\calS_1(x_1), \ldots,
  \calS_\kappa(x_\kappa)).
\end{align*}
Then by 
Theorem~\ref{thm:lower-bound-dir}, 
for 
$X=(x_1, \ldots, x_\kappa)$, 
\[
  \E_{X,\calS_1, \ldots, \calS_\kappa}[l_2^2(g(X), \hat{g}(X) )] \geq \Omega \left(
  \frac{e^{\varepsilon}}{(e^{\varepsilon}+1)^2} \kappa D^2 \right).
\]

Since 
there is a one-to-one correspondence between the ($n,D$)-independent cube $\calA$ 
and 
the ($\kappa,D$)-Boolean independent cube $\{0,1\}^\kappa$, we also have 
\[
  \E_{G,\calR_1, \ldots, \calR_n}[l_2^2(f(G), \hf(G))] \geq \Omega \left(
  \frac{e^{\varepsilon}}{(e^{\varepsilon}+1)^2} n D^2 \right),
\]
where $G$ is drawn uniformly from $\calA$, which proves Theorem~\ref{thm:lower-bound}. \qed

\subsection{Proof of
Theorem~\ref{thm:lower-bound-dir}}
\label{sub:proof_thm_lower-bound-dir}

Assume that $\calS_i : \{0,1\} \rightarrow \mathcal{Z}_i$.
For 
$X = (x_1, \ldots, x_\kappa) \in \{0,1\}^\kappa$, 
let $S(X) = (\calS_1(x_1), \cdots \calS_\kappa(x_\kappa))$ and
$Z = (z_1, \ldots, z_\kappa)$ with $z_i \in \mathcal{Z}_i$. 
We rewrite the quantity of interest as
\[
  \E_{X, S(X)}[l_2^2(g(X), \hg(X))] = \E_{X, S(X)}[(g(X)- \tilde{g}(S(X)))^2].
\]

By the law of total expectation, this quantity is the same as the expected value of the conditional expected value of $(g(X)- \tilde{g}(S(X)))^2$ given 
$S(X) = Z$:
\begin{align}
  &\E_{X, S(X)}[(g(X)- \tilde{g}(S(X)))^2] \nonumber\\
  &= \E_{S(X)} \E_X[(g(X) - \tilde{g}(Z) )^2 | S(X)=Z].
  \label{eq:E_SX_X_g_tg}
\end{align}
Let $\mu_Z = \E_X[g(X)|S(X)=Z]$. 
Then the inner expectation in (\ref{eq:E_SX_X_g_tg}) can be written as follows:
\begin{align*}
  &\; \E_X[(g(X) - \tilde{g}(Z) )^2 | S(X)=Z] \\
  &= \E_X[( (g(X) - \mu_Z) + (\mu_Z - \tilde{g}(Z)) )^2 | S(X)=Z] \\
  &= \E_X[(g(X) - \mu_Z)^2 |S(X)=Z] \\
  &\hspace{3.5mm}+ 2(\mu_Z - \tilde{g}(Z))\E_X[(g(X) - \mu_Z)|S(X)=Z] \\
  &\hspace{3.5mm}+ (\mu_Z - \tilde{g}(Z))^2 \\
  &= \E_X[(g(X) - \mu_Z)^2 |S(X)=Z] + (\mu_Z - \tilde{g}(Z))^2 \\
  &= \V_X[g(X)|S(X)=Z] + (\mu_Z - \tilde{g}(Z))^2.
\end{align*}
Thus, it suffices to show that $\V_X[g(X)|S(X)=Z] \geq
\Omega\left(\frac{e^\epsilon}{(1+e^\epsilon)^2} \kappa D^2 \right)$.
For $B = (b_1, \ldots, b_\kappa) \in \{0,1\}^\kappa$, we have
\begin{align*}
  \Pr[X=B|S(X)=Z] = \frac{\Pr[X=B]\Pr[S(X)=Z|X=B]}{\Pr[S(X)=Z]}.
\end{align*}
Since $\Pr[S(X)=Z]$ does not depend on $B$ and
$\Pr[X=B] = \frac{1}{2^\kappa}$, $\Pr[X=B|S(X)=Z]$ can also be expressed as 
\begin{align}
  \Pr[X=B|S(X)=Z] &\propto \Pr[S(X)=Z|X=B].
  \label{eq:X_B_SX_Z_propto}
\end{align}
Since $S_1, \ldots, S_\kappa$ are independently run, we have
\begin{align*}
  \Pr[S(X)=Z|X=B] 
  &=
  %\;
  \Pr[\calS_1(b_1) = z_1, \ldots, \calS_\kappa(b_\kappa) = z_\kappa] \\
  &= \prod_{i=1}^\kappa \Pr[\calS_i(b_i)=z_i].
\end{align*}
Define 
\begin{align*}
    % p_i = \frac{\Pr[\calS_i(x_i^{1})=z_i]}{\Pr[\calS_i(x_i^0)=z_i] +
    % \Pr[\calS_i(x_i^1)=z_i]}.
    p_i = \frac{\Pr[\calS_i(1)=z_i]}{\Pr[\calS_i(0)=z_i] +
    \Pr[\calS_i(1)=z_i]}.
\end{align*}
Because each $\calS_i$ satisfies $\epsilon$-DP, we have $\frac{1}{1+e^{\epsilon}} \leq p_i \leq \frac{e^\epsilon}{1+e^\epsilon}$.
By (\ref{eq:X_B_SX_Z_propto}) and $\sum_{B \in \{0,1\}^\kappa} \Pr[X=B|S(X)=Z] = 1$, 
we have
\begin{align}
    \Pr[X=B|S(X)=Z] 
    %\propto 
    = 
    \prod_{i=1}^\kappa (p_i)^{b_i}(1-p_i)^{1-b_i}. %\\
\label{eq:Pr_X_B_SX_Z_Bernoulli}
\end{align}
This means that $\Pr[X=B|S(X)=Z]$ is distributed according to the independent product of $Bernoulli(p_i)$ for $i \in [\kappa]$.

Now, because 
$g$ has 
a Boolean $(\kappa,D)$-independent cube, 
there are $C_1,
\ldots, C_\kappa \in \calS$ with $|C_i| \geq D$ such that
\begin{align*}
g(X) &= g(0, \ldots, 0) + \sum_{i=1}^\kappa
x_iC_i.
\end{align*}
By (\ref{eq:Pr_X_B_SX_Z_Bernoulli}), 
$x_i$ 
is an independent draw from $Bernoulli(p_i)$ 
given $S(X)=Z$. 
Thus, the variance of 
$g(X)$ given $S(X)=Z$ 
is
\begin{align*}
\V_X[g(X)|S(X)=Z] &= \sum_{i=1}^\kappa \V[x_i |S(X)=Z]C_i^2 \\
&\geq \sum_{i=1}^\kappa p_i(1-p_i) D^2 \\
&\geq \sum_{i=1}^\kappa \frac{e^\epsilon}{(1+e^\epsilon)^2} D^2 \\
&\geq \kappa \frac{e^\epsilon}{(1+e^\epsilon)^2} D^2.
\end{align*}
\qed
}

\end{document}